\documentclass[envcountsame,runningheads]{llncs}
\usepackage{geometry}
\usepackage{url}

\usepackage{graphicx, 
multicol, 
amssymb, 
amsmath, 
bussproofs,
xcolor, 
comment, 
color, 
enumerate,
enumitem,
proof,
upgreek,
}
\usepackage{cleveref}
\usepackage[all]{xy}

\usepackage{tikz}
\usetikzlibrary{trees}
\usetikzlibrary{fit}
\usepackage{tikz}
\usetikzlibrary{positioning,arrows,calc,decorations.markings}
\tikzset{
modal/.style={>=stealth',shorten >=1pt,shorten <=1pt,auto,node distance=1.5cm,semithick},
world/.style={circle,draw,minimum size=0.5cm,fill=gray!15},
point/.style={circle,draw,inner sep=0.5mm,fill=black},
reflexive above/.style={->,loop,looseness=7,in=120,out=60},
reflexive below/.style={->,loop,looseness=7,in=240,out=300},
reflexive left/.style={->,loop,looseness=7,in=150,out=210},
reflexive right/.style={->,loop,looseness=7,in=30,out=330}
}

\definecolor{tim}{RGB}{0, 0, 250}

\newcommand{\ifandonlyif}{\textit{iff} }
\newcommand{\iffi}{\textit{iff} }
\newcommand{\dfn}{Definition}
\newcommand{\fig}{Figure}
\newcommand{\figs}{Figures}
\newcommand{\sect}{Section}
\newcommand{\sects}{Sections}
\newcommand{\prp}{Proposition}
\newcommand{\thm}{Theorem}
\newcommand{\rmk}{Remark}
\newcommand{\eg}{Example}

\newcommand{\ih}{IH }
\newcommand{\lem}{Lemma}

\def\phi{\varphi}

\def\imp{\rightarrow}

\def\g{[\mathsf{G}]} 
\def\h{[\mathsf{H}]} 
\def\f{\langle\mathsf{F}\rangle} 
\def\p{\langle\mathsf{P}\rangle} 

\newcommand{\dneg}{{\ast}}
\newcommand{\bull}{{\bullet}}
\newcommand{\empdata}{\emptyset}

\newcommand{\coni}{\mathsf{(P_{1})}}
\newcommand{\conii}{\mathsf{(P_{2})}}
\newcommand{\coniii}{\mathsf{(P_{3})}}
\newcommand{\coniv}{\mathsf{(P_{4})}}
\newcommand{\conv}{\mathsf{(P_{5})}}
\newcommand{\convi}{\mathsf{(P_{6})}}
\newcommand{\convii}{\mathsf{(P_{7})}}

\newcommand{\grph}[1]{G(#1)}
\newcommand{\rest}{\restriction}

\newcommand{\lang}{\mathcal{L}}

\newcommand{\langatoms}{\mathcal{L}_{\land,\f,\p}}

\newcommand{\kt}{\mathsf{K_{t}}}
\newcommand{\ktp}{\mathsf{K_{t}P}}

\newcommand{\id}{(id)}

\def\wk{(w)}
\def\swk{(\bar{w})}
\def\slsub{(\bar{ls})}

\def\ctrr{(c_{r})}

\def\ctrl{(c_{l})}

\def\cut{(cut)}

\newcommand{\refl}{(ref)}

\def\negl{(\neg_{l})}
\def\negr{(\neg_{r})}

\def\botl{(\bot_{l})}
\def\topr{(\top_{r})}
\def\botr{(\bot_{r})}

\def\topl{(\top_{l})}

\def\negr{(\neg_{r})}
\def\negl{(\neg_{l})}
\newcommand{\disr}{(\vee_{r})}
\newcommand{\conr}{(\wedge_{r})}
\newcommand{\disl}{(\vee_{l})}
\newcommand{\conl}{(\wedge_{l})}
\newcommand{\impr}{(\rightarrow_{r})}
\newcommand{\impl}{(\rightarrow_{l})}
\def\gr{(\g_{r})}

\def\gld{(\g_{l})}
\def\hld{(\h_{l})}
\def\frd{(\f_{r})}
\def\prd{(\p_{r})}

\def\gl{(\g_{l})}

\def\pr{(\p_{r})}
\def\prii{(\p_{r2})}
\def\pl{(\p_{l})}
\def\hr{(\h_{r})}
\def\hl{(\h_{l})}

\def\fr{(\f_{r})}

\def\fl{(\f_{l})}
\newcommand{\ptsrd}{(pt_{\dseq})}
\newcommand{\ptsrl}{(pt_{\lseq})}

\def\sar{\Rightarrow}
\def\dar{\vdash}
\newcommand{\srus}{\mathsf{SR}}
\newcommand{\axs}{\mathsf{Ax}}

\newcommand{\dri}{(\delta_{1})}
\newcommand{\drii}{(\delta_{2})}
\newcommand{\driii}{(\delta_{3})}
\newcommand{\driv}{(\delta_{4})}
\newcommand{\drv}{(\delta_{5})}
\newcommand{\drvi}{(\delta_{6})}
\newcommand{\drvii}{(\delta_{7})}
\newcommand{\drviii}{(\delta_{8})}
\newcommand{\drix}{(\delta_{9})}

\newcommand{\deri}{(\rho_{1})}
\newcommand{\derii}{(\rho_{2})}
\newcommand{\deriii}{(\rho_{3})}
\newcommand{\deriv}{(\rho_{4})}
\newcommand{\derv}{(\rho_{5})}

\newcommand{\lsub}{(ls)}

\newcommand{\il}{(I_{l})}
\newcommand{\ir}{(I_{r})}
\newcommand{\ql}{(q_{l})}
\newcommand{\qr}{(q_{r})}
\newcommand{\dwl}{(w_{l})}
\newcommand{\dwr}{(w_{r})}
\newcommand{\asl}{(a_{l})}
\newcommand{\asr}{(a_{r})}
\newcommand{\pml}{(p_{l})}
\newcommand{\pmr}{(p_{r})}
\newcommand{\dcl}{(c_{l})}
\newcommand{\dcr}{(c_{r})}
\newcommand{\ml}{(m_{l})}
\newcommand{\mr}{(m_{r})}

\newcommand{\vA}{\dot{A}}
\newcommand{\vB}{\dot{B}}
\newcommand{\vC}{\dot{C}}
\newcommand{\vX}{\dot{X}}
\newcommand{\vP}{\dot{p}}
\newcommand{\vQ}{\dot{q}}
\newcommand{\vR}{\dot{r}}
\newcommand{\vY}{\dot{Y}}
\newcommand{\vZ}{\dot{Z}}
\newcommand{\vW}{\dot{W}}
\newcommand{\vL}{\dot{\lseq}}
\newcommand{\sls}{\dot{\Lambda}}
\newcommand{\ssX}{\mathcal{X}}
\newcommand{\ssY}{\mathcal{Y}}
\newcommand{\ssZ}{\mathcal{Z}}
\newcommand{\ssF}{\mathcal{F}}

\newcommand{\lw}{\dot{w}}
\newcommand{\lu}{\dot{u}}
\newcommand{\lv}{\dot{v}}
\newcommand{\lz}{\dot{z}}
\newcommand{\lx}{\dot{x}}
\newcommand{\ly}{\dot{y}}
\newcommand{\lr}{\dot{e}}
\newcommand{\lo}{\dot{o}}
\newcommand{\vs}{\dot{s}}

\newcommand{\ux}{\dot{X}}
\newcommand{\uy}{\dot{Y}}
\newcommand{\uz}{\dot{Z}}

\newcommand{\rx}{X}
\newcommand{\ry}{Y}
\newcommand{\rz}{Z}

\newcommand{\lseq}{\lambda} 
\newcommand{\dseq}{\delta} 
\newcommand{\ulseq}{\dot{\lseq}}
\newcommand{\varlseq}{\dot{\lseq}}

\newcommand{\np}{\mathrm{NP}}

\newcommand{\ptime}{\mathrm{PTIME}}

\def\dll{\mathfrak{L}_{1}}
\def\dlr{\mathfrak{L}_{2}}
\def\dl{\mathfrak{L}}
\def\ld{\mathfrak{D}}
\def\ldl{\mathfrak{D}_{1}}
\def\ldr{\mathfrak{D}_{2}}
\newcommand{\sub}{\sigma}
\newcommand{\subdl}{\sub^{\dl}}
\newcommand{\subld}{\sub^{\ld}}
\def\lt{\tau_{1}}
\def\rt{\tau_{2}}
\newcommand{\ftd}{\tau}

\newcommand{\seqcomp}{\otimes}
\newcommand{\ptcomp}[1]{\oplus_{#1}}

\newcommand{\wdth}[1]{w(#1)}
\newcommand{\lgth}[1]{\ell(#1)}
\newcommand{\qty}[1]{q(#1)}
\newcommand{\flgth}[1]{\ell(#1)}
\newcommand{\size}[1]{s(#1)}
\newcommand{\subseq}[2]{|^{#1}_{#2}}
\newcommand{\lgthdseq}[1]{\ell(#1)}

\newcommand{\dequiv}{\equiv_{\delta}}
\newcommand{\sst}[1]{\mathcal{S}(#1)}

\newcommand{\propset}{\mathrm{Prop}}
\newcommand{\lab}{\mathrm{Lab}}

\newcommand{\paxs}{\mathsf{P}}

\newcommand{\dkt}{\mathsf{DK_{t}}}
\newcommand{\dktp}{\mathsf{DK_{t}P}}

\newcommand{\gtkt}{\mathsf{G3K_{t}}}
\newcommand{\gtktii}{\mathsf{LK_{t}}}
\newcommand{\gtktp}{\mathsf{G3K_{t}P}}

\newcommand{\bigocirc}{\scalebox{1.5}{$\bigcirc$}}
\newcommand{\iso}{\cong}

\def\rel{\mathcal{R}}

\def\ru{(r)}

\newcommand{\scomp}{\otimes}
\newcommand{\hil}{\mathsf{H}}

\newcommand{\lproof}{\Uppi}
\newcommand{\dproof}{\Uppi}
\newcommand{\proves}[3]{#1, #2 \Vdash #3}

\begin{document}

\title{Realizing the Maximal Analytic Display Fragment of Labeled Sequent Calculi for Tense Logics\thanks{This research was supported by the European Research Council European Research Council (Grant Agreement no. 771779, DeciGUT).}}

\titlerunning{Realizing the Maximal Analytic Display Fragment of Labeled Sequent Calculi}

\author{Tim S. Lyon\orcidID{0000-0003-3214-0828}}

\authorrunning{Tim S. Lyon}

\institute{Computational Logic Group, Institute of Artificial Intelligence, Technische Universit\"at Dresden, Germany  \\ \email{timothy\_stephen.lyon@tu-dresden.de}}

\maketitle              

\begin{abstract}

We define and study translations between the maximal class of analytic display calculi for tense logics and labeled sequent calculi, thus solving an open problem about the translatability of proofs between the two formalisms. In particular, we provide $\ptime$ translations that map cut-free display proofs to and from special cut-free labeled proofs, which we dub `strict' labeled proofs. This identifies the space of cut-free display proofs with a polynomially equivalent subspace of labeled proofs, showing how calculi within the two formalisms polynomially simulate one another. We analyze the relative sizes of proofs under this translation, finding that display proofs become polynomially shorter when translated to strict labeled proofs, though with a potential increase in the length of sequents; in the reverse translation, strict labeled proofs may become polynomially larger when translated into display proofs. In order to achieve our results, we formulate labeled sequent calculi in a new way that views rules as `templates', which are instantiated with substitutions to obtain rule applications; we also provide the first definition of primitive tense structural rules within the labeled sequent formalism. Therefore, our formulation of labeled calculi more closely resembles how display calculi are defined for tense logics, which permits a more fine-grained analysis of rules, substitutions, and translations. This work establishes that every analytic display calculus for a tense logic can be viewed as a labeled sequent calculus, showing conclusively that the labeled formalism subsumes and extends the display formalism in the setting of primitive tense logics.

\end{abstract}

\section{Introduction}

Logical methods are at the forefront of research in computer science, artificial intelligence, and formal philosophy, giving rise to a diverse number of logics capable of executing reasoning fine-tuned for specific application scenarios. Logics have found extensive, meaningful applications in various domains such as the verification of software~\cite{McM18}, in ontology mediated querying~\cite{BaaHorLutSat17}, and in explainable AI~\cite{BerStr22}. Analytic proof systems are of chief importance to logics, being used to facilitate reasoning with logics, to establish non-trivial properties, and to design automated reasoning methods.

An analytic calculus is a collection of inference rules that generates proofs satisfying the \emph{subformula property}, i.e. every formula occurring in a proof is a subformula of the conclusion of the proof. This is a powerful constraint on the shape of proofs that has been exploited in various ways, e.g. in writing decision procedures~\cite{Dyc92,Gen35a,Gen35b} and in computing interpolants~\cite{LyoTiuGorClo20,Mae60}. One of the dominant formalisms (alongside tableaux~\cite{Fit72,Smu68}) for constructing analytic calculi is Gentzen's sequent formalism. Typically, sequent systems include a \emph{cut} rule, which, although helpful for establishing completeness (among other uses), deletes formulae from the premises to the conclusion when applied. This has the effect that proofs in a sequent calculus with cut do not exhibit the desirable subformula property, 
though, this property can be reclaimed by establishing cut-elimination. That is, one shows that the cut rule can be eliminated from any derivation without affecting the conclusion derived, and therefore, is superfluous in the underlying sequent system. 

Calculi built within the sequent formalism consist of rules that operate over sequents, i.e. formulae of the form $X \vdash Y$, where $X$ and $Y$ are sequences or (multi)sets of formulae. Despite the success of sequent systems in discovering new logics~\cite{Gir87}, in automated reasoning~\cite{Sla97,Dyc92}, and in confirming non-trivial logical properties~\cite{Mae60}, it was discovered that the structure of the sequent is insufficient to capture more expressive logics in an analytic or cut-free manner (e.g. the modal logic $\mathsf{S5}$ and bi-intuitionistic logic~\cite{BuiGor07}). In an attempt to recapture analyticity for more expressive logics of interest, various generalization of Gentzen's sequent formalism have been proposed; e.g. hypersequents~\cite{Avr96,Pot83}, linear nested sequents~\cite{Lel15}, nested sequents~\cite{Bul92,Kas94}, display sequents~\cite{Bel82,Wan94}, and labeled sequents~\cite{Sim94,Vig00}. We will focus on the relationship between the latter two formalisms in this paper in the context of tense logics.

A prominent generalization of Gentzen's sequent formalism is the display framework, originally defined by Belnap in 1982 under the name \emph{display logic}~\cite{Bel82}. Belnap's display calculi generalize Gentzen's calculi by expanding sequents with a host of new structural connectives---corresponding to pairs of dual connectives---along with rules for manipulating them. Incorporating structural connectives for pairs of dual connectives has proven fruitful in the design of analytic systems for large classes of logics~\cite{Bel82,Bro12,Gor98,Kra96,Wan94,Wol98}. This is in part due to a general cut-elimination theorem proved by Belnap, which states that a display calculus admits cut-elimination if eight (checkable) syntactic conditions are satisfied~\cite{Bel82}. We remark that Belnap's display formalism is among one of the most expressive generalizations of Gentzen's sequent formalism, providing a uniform and modular proof theory for substantial classes of logics. 

Another rich extension of Gentzen's sequent formalism is the labeled formalism. Whereas display calculi are obtained via algebraic considerations, labeled sequent calculi are normally extracted from a logic's relational semantics. The labeled sequent formalism was initiated by Kanger in the 1950's~\cite{Kan57}, though it was arguably the work of Simpson~\cite{Sim94} that provided the contemporary, general form of labeled sequents and their encompassing systems. Labeled sequents extend Gentzen sequents by prefixing formulae with labels (e.g. $w : A$) and incorporating relational atoms (e.g. $Rwu$) into the syntax of sequents. A labeled formulae $w : A$ encodes that a formula $A$ holds at a world $w$ in a Kripke model and a relational atom $Rwu$ encodes the accessibility relation. Like the display formalism, the labeled formalism has been used to provide analytic calculi for extensive classes of logics. Moreover, fundamental properties hold generally for such systems, e.g. labeled sequent systems admit the height-preserving admissibility of important structural rules (e.g. weakenings and contractions), have height-preserving invertible rules, and cut-elimination holds~\cite{Bor08,Hei05,Sim94,Vig00}.

There are a number of key differences between the display and labeled formalisms, e.g. display sequents are composed of structural connectives whereas labeled sequents encode graphs, display calculi are based on an algebraic semantics whereas labeled calculi are based on a relational semantics, display calculi employ structural rules and display rules while such rules are admissible or redundant in labeled calculi, etc. Despite these notional, motivational, and operational distinctions, both formalisms permit the definition of analytic calculi for wide classes of overlapping logics, implying the existence of an underlying relationship between the two formalisms. Normally, relationships between calculi are studied via translations, i.e. functions that stepwise translate a proof from one calculus into a proof in another, and it is informative to provide the computational complexity of the translations as well as compare the relative sizes of the input and output proofs.

 Preliminary studies relating display and labeled calculi for some normal modal logics were undertaken by Mints~\cite{Min97} and Restall~\cite{Res06}. A more intensive investigation studying the relationship between shallow-nested calculi (which are close relatives of display calculi) and labeled calculi for extensions of 
 the tense logic $\kt$ with path axioms (encoding Horn properties on Kripke frames) was given more recently in~\cite{CiaLyoRamTiu21}. Nevertheless, the renowned `Display Theorem I' by Kracht tells us that a substantially larger class of logics (subsuming the aforementioned logics), namely, \emph{primitive tense logics} admit analytic display calculi~\cite{Kra96}. Moreover, this result also tells us that if a tense logic has an analytic display calculus, then it is a primitive tense logic, meaning the set of display calculi for primitive tense logics forms the \emph{maximal analytic} set of display calculi for tense logics. Therefore, the aforementioned studies fall short of characterizing the relationship between all analytic display calculi for tense (and modal) logics and their labeled counterparts. Indeed, this problem was explicitly left open in Restall's 2006 paper~\cite{Res06} as well as in~\cite{CiaLyoRamTiu21}, and all aforementioned studies left the complexity of such translations unaddressed. In this paper, we develop and apply new techniques to solve this open problem and provide complexity bounds for our translations. In particular, we accomplish the following:

\begin{itemize}[leftmargin=*]

\item We provide the first characterization of primitive tense rules in the setting of labeled sequent calculi.

\item We provide a new formulation of labeled sequent calculi that matches the standard formulation of display calculi, namely, we take inference rules to be templates with variables that are instantiated to produce rule applications. This new formulation is crucial in defining translations between display and labeled proofs.

\item We establish a correspondence between the space of all analytic display proofs and a special subspace of labeled proofs, which we dub \emph{strict} labeled proofs. We show that $\ptime$ translations exist between every analytic display proof and each corresponding strict labeled proof. Furthermore, our translations are \emph{direct} and do not take any detours through other systems as in other approaches (cf.~\cite{CiaLyoRamTiu21}). 

\item We find that translating a display proof into a strict labeled proof yields a proof with the same number of sequents or less, but may quadratically increase the length of sequents appearing in the output proof; therefore, there is a trade-off between the number of sequents and their lengths when translating from display to labeled, which is a previously unrecognized finding. Conversely, translating a strict labeled proof into a display proof may polynomially increase the number of sequents in the output proof.


\item We observe that display structures (which only form part of a display seqeuent) \emph{are} labelled sequents, and that all display equivalent sequents translate to the same labeled sequent. This demonstrates a higher degree of syntactic bureaucracy in the display formalism and also shows that labeled sequents are canonical representations of display sequents since they do away with bureaucracy that obfuscates identities. 


\end{itemize}

This paper is organized as follows: In \sect~\ref{section-1}, we define and give the preliminaries for primitive tense logics. We then present Kracht's display calculi for primitive tense logics in \sect~\ref{section-2}. In \sect~\ref{section-3}, we reformulate and extend Boretti's labeled calculi for tense logics~\cite{Bor08} to capture primitive tense logics in a manner suitable for translations with display systems. In \sect~\ref{section-3-2}, we define labeled polytree sequents and strict proofs, which will be invaluable in translating labeled and display proofs. We then show how to translate between display and labeled notation in \sect~\ref{section-4}, and put our translations to use in \sects~\ref{section-5} and~\ref{section-6}, showing how to translate display proofs into strict labeled proofs and vice versa, respectively, as well as analyze the complexity of our translations. In \sect~\ref{conclusion}, we re-emphasize our findings and discuss future work.

\section{Tense Logics and Primitive Axioms}\label{section-1}


In this section, we introduce and define the class of logics our proof systems consider, namely, \emph{tense logics}. Tense logics were invented by Prior in the 1950's~\cite{Pri03} and extend classical propositional logic with modalities making reference to both future and past states of affairs. More concretely, tense logics employ the $\g$ and $\f$ modalities, which apply to a proposition that holds at every next moment, or at some next moment, respectively, and the $\h$ and $\p$ modalities, which apply to a proposition that holds at all prior moments, or at some prior moment, respectively. In other words, $\g$ is read as `in all future moments', $\f$ is read as `in some future moment', $\h$ is read as `in all past moments', and $\p$ is read as `in some past moment'. The remaining logical operators employed in our language 
 are the familiar classical operators $\top$, $\bot$, $\neg$, $\lor$, $\land$, and $\rightarrow$, that is, the language we use is an extension of the language of classical propositional logic.

\begin{definition}[The Language $\lang$]\label{def:language} We define our \emph{language} $\lang$ to consist of the formulae generated via the following grammar in BNF:
$$
A ::= p \ | \ \top \ | \ \bot \ | \ \neg A \ | \ (A \lor A) \ | \ (A \land A) \ | \ (A \rightarrow A) \ | \ \g A \ | \ \f A \ | \ \h A \ | \ \p A
$$
where $p$ ranges over a denumerable set $\propset$ of \emph{propositional atoms}, which we call \emph{atoms} for short. We use $p$, $q$, $r$, $\ldots$ (occasionally annotated) to denote atoms and $A$, $B$, $C$, $\ldots$ (occasionally annotated) to denote formulae from $\lang$. The bi-conditional operator is defined in the usual way as $A \leftrightarrow B := (A \rightarrow B) \land (B \rightarrow A)$.
\end{definition}

We define the \emph{length} of a formula $A$, denoted $\flgth{A}$, recursively as  $\flgth{p} = \flgth{\bot} = \flgth{\top} = 1$, $\flgth{A \odot B} = \flgth{A} + \flgth{B} + 1$ for $\odot \in \{\lor, \land, \rightarrow\}$, and $\flgth{\triangledown A} = \flgth{A} + 1$ for $\triangledown \in \{\g, \f, \h, \p\}$. Our language is interpreted over standard relational (i.e. Kripke) models for normal modal logics~\cite{Kri63}.

\begin{definition}[Relational Frame/Model] We define a \emph{relational frame} $F$ to be a tuple $(W,R)$ such that $W$ is a non-empty set of \emph{worlds} and $R \subseteq W \times W$ is the \emph{accessibility relation}. A \emph{relational model} $M$ is defined to be a tuple $(F,V)$ such that $F$ is a relational frame and $V : \propset \mapsto 2^{W}$ is a \emph{valuation function} mapping atoms to sets of worlds.
\end{definition}

We may think of the worlds in a relational model as possible moments in time with the accessibility relation holding between a moment and a (possible) future moment. Formulae are then interpreted at a specific moment in time (i.e. a world) and express a proposition concerning a present, past, or future state of affairs. The following definition makes the interpretation of formulae precise.

\begin{definition}[Semantic Clauses] Let $M = (W,R,V)$ be a relational model with $w \in W$. We define the \emph{satisfaction} of a formula $A \in \lang$ in $M$ at $w$, written $M,w \Vdash A$, recursively on the structure of $A$ as follows:
\begin{itemize}

\item $M, w \Vdash p$ \ifandonlyif $w \in V(p)$;

\item $M, w \Vdash \top$;

\item $M, w \not\Vdash \bot$;

\item $M, w \Vdash \neg A$ \ifandonlyif $M, w \not\Vdash A$;

\item $M, w \Vdash A \lor B$ \ifandonlyif $M, w \Vdash A$ or $M, w \Vdash B$;

\item $M, w \Vdash A \land B$ \ifandonlyif $M, w \Vdash A$ and $M, w \Vdash B$;

\item $M, w \Vdash A \rightarrow B$ \ifandonlyif $M, w \not\Vdash A$ or $M, w \Vdash B$;

\item $M, w \Vdash \g A$ \ifandonlyif for all $u \in W$, if $(w,u) \in R$, then $M, u \Vdash A$;

\item $M, w \Vdash \f A$ \ifandonlyif for some $u \in W$, $(w,u) \in R$ and $M, u \Vdash A$;

\item $M, w \Vdash \h A$ \ifandonlyif for all $u \in W$, if $(u,w) \in R$, then $M, u \Vdash A$;

\item $M, w \Vdash \p A$ \ifandonlyif for some $u \in W$, $(u,w) \in R$ and $M, u \Vdash A$.

\end{itemize}
A formula $A$ is \emph{globally true} in a relational model $M$ \ifandonlyif $M,w \Vdash A$ for all $w \in W$, and a formula $A$ is \emph{$\kt$-valid} \ifandonlyif it is globally true on all relational models.
\end{definition}


We note that the set of all $\kt$-valid formulae is axiomatizable (see~\cite{BlaRijVen01}), i.e. the following axiomatization $\hil\kt$ is sound and complete relative to our relational semantics.

\begin{definition}[$\hil\kt$, $\kt$] The axiomatization $\hil\kt$ consists of all classical propositional tautologies along with the following axioms and inference rules:
\begin{center}
\begin{multicols}{2}
\begin{itemize}

\item[A1] $\g (A \rightarrow B) \rightarrow (\g A \rightarrow \g B)$

\item[A2] $\h (A \rightarrow B) \rightarrow (\h A \rightarrow \h B)$

\item[A3] $A \rightarrow \g \p A$

\item[A4] $A \rightarrow \h \f A$

\item[R0] \AxiomC{$A$}
\AxiomC{$A \rightarrow B$}
\BinaryInfC{$B$}
\DisplayProof

\item[A5] $\g A \leftrightarrow \neg \f \neg A$

\item[A6] $\h A \leftrightarrow \neg \p \neg A$

\item[R1] \AxiomC{$A$}
\UnaryInfC{$\g A$}
\DisplayProof

\item[R2] \AxiomC{$A$}
\UnaryInfC{$\h A$}
\DisplayProof

\end{itemize}
\end{multicols}
\end{center}
The logic $\kt$ is defined to be the smallest set of formulae closed under substitutions of the axioms A1--A6 and applications of the inference rules R0--R2.
\end{definition}

The logic $\kt$ serves as the base logic in the class of tense logics we consider. All other logics considered can be characterized as extensions of $\kt$ with \emph{primitive tense axioms}~\cite{Kra96}.

\begin{definition}[Primitive Tense Axiom~\cite{Kra96}]\label{def:primitive-tense-axiom} A \emph{primitive tense axiom} is a formula of the form
$A \imp B$, where $A$ and $B$ are generated via the following grammar in BNF:
$$
A ::= p \ | \ \top \ | \ (A \land A) \ | \ (A \lor A) \ | \ \f A \ | \ \p A
$$
and $A$ contains each atom at most once.
\end{definition}

We will use $\ktp$ to denote the \emph{primitive tense logic} axiomatized by $\kt$ extended with a set of primitive tense axioms $\mathsf{P}$, and will use $\hil\ktp$ to denote the logic's axiomatization. The set of primitive tense axioms covers a wide array of well-known axioms such as the axiom $\mathsf{T}$ ($A \rightarrow \f A$), the axiom $\mathsf{4}$ ($\f \f A \rightarrow \f A$), and the axiom $\mathsf{5}$ ($\p \f A \rightarrow \f A$). Beyond capturing a diverse class of logics, the proof-theoretic significance of primitive tense logics was first identified by Kracht~\cite{Kra96}, who was able to show that the largest set of display calculi for tense logics satisfying a certain set of desirable properties (discussed in the subsequent section) is characterizable as those calculi that are sound and complete relative to primitive tense logics.

As remarked by Kracht~\cite{Kra96}, by making use of associativity, commutativity, and idempotency properties of $\lor$ and $\land$, along with standard equivalences, 
one can transform any primitive tense axiom into an equivalent normal form, defined below.

\begin{proposition}[Primitive Tense Normal Form]\label{prop:primitive-tense-normal-form} We define the language  $\langatoms$ to be the smallest set of formulae generated by the following grammar:
$$
A ::= p \ | \ \top \ | \ (A \land A) \ | \ \f A \ | \ \p A
$$
Let $A$, each $A_{i}$, and each $B_{i,j}$ be in $\langatoms$ with $A$ and each $A_{i}$ containing at most one occurrence of each atom. Every primitive tense axiom can be put into a normal form, shown below left. We define a \emph{simplified primitive tense axiom} to be a formula of the form shown below right.
$$
\bigwedge_{1 \leq i \leq n} (A_{i} \rightarrow \bigvee_{1 \leq j \leq m_{i}} B_{i,j})
\qquad\qquad
A \rightarrow \bigvee_{1 \leq j \leq m} B_{j}
$$
\end{proposition}



Since every primitive tense axiom is equivalent to a conjunction of simplified primitive tense axioms (by \cref{prop:primitive-tense-normal-form} above), each primitive tense logic $\ktp$ admits an axiomatization $\hil\ktp$ where $\mathsf{P}$ contains simplified primitive tense axioms only. We therefore assume w.l.o.g. that any set $\mathsf{P}$ of primitive tense axioms contains only simplified primitive tense axioms, unless specified otherwise.

\section{Display Calculi for Tense Logics}\label{section-2}


 In this section, we describe the display calculi for primitive tense logics introduced by Wansing~\cite{Wan94} and Kracht~\cite{Kra96}. One of the central ideas behind the construction of display systems is the use of \emph{structural connectives} (also called \emph{Gentzen toggles}~\cite{Kra96}) which stand for different logical connectives depending on their position within a sequent. The use of structural connectives dates back to the inception of sequents, where Gentzen employed the comma connective representing a conjunction on the left and a disjunction on the right of a sequent~\cite{Gen35a,Gen35b}; for example, $A, B \dar C$ is interpretd as $A \land B \imp C$ whereas $A \dar B, C$ is interpreted as $A \imp B \lor C$. 

Although Belnap discusses display systems for modal logics in his seminal paper introducing display calculi~\cite{Bel82}, their formulation was improved upon and simplified in the work of Wansing~\cite{Wan94}. While Belnap had formulated the structural connectives standing for modalities as \emph{binary connectives}, Wansing formulated them as \emph{unary connectives}. In particular, the bullet $\bull$ was used to represent a $\p$ in the antecedent of a sequent and to represent a $\g$ in the consequent; e.g. $\bull A \dar B$ is interpreted as $\p A \imp B$ and $A \dar \bull B$ is interpreted as $A \imp \g B$~\cite{Wan94}. 

The use of structural connectives in building sequents has proven beneficial. In the first place, the inclusion or exclusion of rules governing the behavior of structural connectives permits the supplementation of display calculi for diverse families of logics; e.g. bunched-implication logics~\cite{Bro12}, substructural logics~\cite{Gor98}, bi-intuitionistic logic~\cite{Wol98}, tense/temporal logics~\cite{Kra96,Wan94}, and relevance logics~\cite{Bel82}.  In the second place, the use of structural connectives supports the shifting of data within a sequent while at the same time preserving the subformula property. For instance, in the context of tense logics the sequents $\bull A \dar B$ and $A \dar \bull B$ are mutually derivable from one another; observe that the $\bull$ may be shifted to the antecedent or consequent, either \emph{displaying} $B$ as the entire consequent or $A$ as the entire antecedent, respectively, without changing the logical formulae that occur. The capacity to \emph{display} data within a display sequent (i.e. to shuffle data around) is what gives rise to the name of Belnap's formalism.
 
Another significant feature of the display formalism is the existence of a \emph{general cut-elimination theorem}. As shown by Belnap in~\cite{Bel82}, so long as a display calculus satisfies eight syntactic conditions (C1)--(C8), the calculus admits cut-elimination (see \dfn~\ref{def:C1-C8} and \thm~\ref{thm:general-cut-elim} below). Thus, the display formalism is well-suited for supplying a variety of logics with \emph{cut-free} calculi in a uniform and modular fashion. 

The combination of data by means of structural connectives gives rise to \emph{structures}, which serve as the entire antecedent or consequent of a display sequent. In the context of tense logics, we define a structure $X$ to be an object generated from the following grammar in BNF:
$$
X ::= A \ | \ I \ | \ {\dneg}X \ | \ {\bull} X \ | \ (X \circ X)
$$
where $A$ ranges over the formulae in $\lang$. We use $X$, $Y$, $Z$, $\ldots$ (occasionally annotated) to denote structures. A \emph{display sequent} is defined to be a formula of the form $X \vdash Y$ where $X$ and $Y$ are structures. We use $\dseq_{1}$, $\dseq_{2}$, $\dseq_{3}$, $\ldots$ (or, just $\dseq$) to denote display sequents. We recursively define the \emph{length} of a structure $X$ as follows: $\lgthdseq{A} = \lgthdseq{I} = 1$, $\lgthdseq{\dneg X} = \lgthdseq{X} + 1$, $\lgthdseq{\bull X} = \lgthdseq{X} + 1$, $\lgthdseq{X \circ Y} = \lgthdseq{X} + \lgthdseq{Y} + 1$. The \emph{length} of a display sequent is defined to be $\lgthdseq{X \vdash Y} = \lgthdseq{X} + \lgthdseq{Y}$.

As mentioned above, structural connectives act as a proxy for certain logical connectives dependent upon where they occur within a display sequent. This gives rise to structures and display sequents being translatable into logical formulae.

\begin{definition}[Formula Translation $\tau$~\cite{Kra96}]\label{def:formula-translation-display} We define the translations $\lt$ and $\rt$ from structures to formulae as follows:
\begin{center}
\begin{multicols}{2}
\begin{itemize}

\item $\lt(A) := A$

\item $\lt(I) := \top$

\item $\lt(\dneg X) := \neg \rt(X)$

\item $\lt(\bull X) := \p \lt(X)$

\item $\lt(X \circ Y) := \lt(X) \land \lt(Y)$

\end{itemize}

\begin{itemize}

\item $\rt(A) := A$

\item $\rt(I) := \bot$

\item $\rt(\dneg X) := \neg \lt(X)$

\item $\rt(\bull X) := \g \rt(X)$

\item $\rt(X \circ Y) := \rt(X) \lor \rt(Y)$

\end{itemize}
\end{multicols}
\end{center}
We define the translation $\tau$ of a display sequent as $\tau(X \vdash Y) := \lt(X) \rightarrow \rt(Y)$.
\end{definition}

We note that the structure $I$ is taken to stand for the \emph{empty structure}, which functions as the identity element with respect to the $\circ$ operator, that is, $I \circ X$, $X \circ I$, and $X$ are taken to be equivalent. This property of $I$ is formalized as rules in our display calculus (viz. $\il$ and $\ir$; see \fig~\ref{fig:structural-rules-DKt}) and corresponds to the fact that in the antecedent (or, consequent) a structure of the form $I \circ X$ or $X \circ I$ is equivalent to $\top \land \lt(X)$ and $\lt(X) \land \top$ (or, $\bot \lor \rt(X)$ and $\rt(X) \lor \bot$, respectively), which is equivalent to $\lt(X)$ (or, $\rt(X)$, respectively) by the formula translation above. We now define the notion of a \emph{substructure}, which will be used in the sequel. 

 
\begin{definition}[Substructure]\label{def:substructure} We define $W$ to be a \emph{substructure} of $X$ \ifandonlyif $W \in \sst{X}$ where $\sst{X}$ is defined recursively:

\begin{itemize}

\item $\sst{A} := \{A\}$

\item $\sst{I} := \{I\}$

\item $\sst{\dneg Y} := \{\dneg Y \} \cup \sst{Y}$

\item $\sst{\bull Y} := \{\bull Y\} \cup \sst{Y}$

\item $\sst{Y \circ Z} := \{Y \circ Z\} \cup \sst{Y} \cup \sst{Z}$

\end{itemize}
\end{definition}

We present Kracht's display calculus $\dkt$ for the minimal tense logic $\kt$, consisting of three sets of rules: \emph{logical rules} (see \fig~\ref{fig:logical-rules-DKt}), \emph{display rules} (see \fig~\ref{fig:display-rules-DKt}), and \emph{structural rules} (see \fig~\ref{fig:structural-rules-DKt}), which we introduce in turn. After introducing $\dkt$, we define extensions of $\dkt$ with rules corresponding to primitive tense axioms (cf.~\cite{Kra96}). Before we proceed however, let us comment on the types of symbols that occur within rules and their instances.

We let the symbols $\vP$, $\vR$, $\vQ$, $\ldots$ (occasionally annotated) be \emph{atomic variables}, which are instantiated with atoms in rule applications, the symbols $\vA$, $\vB$, $\vC$, $\ldots$ (occasionally annotated) be \emph{formula variables}, which are instantiated with formulae from $\lang$ in rule applications, and the symbols $\vX$, $\vY$, $\vZ$, $\ldots$ (occasionally annotated) be \emph{structure variables}, which are instantiated with structures in rule applications. (NB. Observe that atoms are formulae and formulae are structures, meaning formula variables can be instantiated with atoms and structure variables can be instantiated with formulae.) A \emph{schematic structure} is defined to be a combination of atomic variables, formula variables, and structure variables with logical and structural connectives, i.e. it is a formula generated via the following grammar in BNF:
$$
\ssX ::= \vP \ | \ \top \ | \ \bot \ | \ I \ | \ \vA \ | \ \triangledown \vA \ | \ \vA \odot \vB \ | \ \vX \ | \ \dneg \ssX \ | \ \bull \ssX \ | \ \ssX \circ \ssX
$$
where $\triangledown \in \{\neg, \g, \f, \h, \p\}$ and $\odot \in \{\lor, \land, \imp\}$. We use $\ssX$, $\ssY$, $\ssZ$, $\ldots$ (occasionally annotated) to denote schematic structures and define a \emph{schematic display sequent} to be a formula of the form $\ssX \dar \ssY$.


\begin{definition}[Substitution $\sub$] A \emph{substitution} $\sub$ is defined to be a function that maps atomic variables to atoms, formula variables to formulae, and structure variables to structures. We use postfix notation when applying substitutions, viz. we let $\vP\sub$ be the atom obtained by applying $\sub$ to $\vP$, $\vA\sub$ be the formula obtained from applying $\sub$ to $\vA$, and $\vX\sub$ be the structure obtained from applying $\sub$ to $\vX$. We extend substitutions to schematic structures and schematic display sequents in the expected way by applying them to each atomic variable, formula variable, and structure variable therein (cf.~\cite{Bel82,Kra96}). It immediately follows from the definition that for any substitution $\sub$, $I\sub = I$.

We always assume that if a substitution $\sub$ is applied to a set of schematic structures or schematic display sequents, then all atomic variables, formula variables, and structure variables contained therein are within the domain of the substitution. If we want to indicate what variables are mapped to which atoms, formulae, or structures, then we write a substitution $\sub$ as
$$
[X_{1} / \vX_{1}, \ldots, X_{n} / \vX_{n}, A_{1} / \vA_{1}, \ldots, A_{k} / \vA_{k}, p_{1} / \vP_{1}, \ldots, p_{m} / \vP_{m}]
$$
 indicating that $\vX_{i}\sub = X_{i}$, $\vA_{j}\sub = A_{j}$, and $\vP_{t}\sub = p_{t}$ for $i \in \{1, \ldots, n\}$, $j \in \{1, \ldots, k\}$, and $t \in \{1, \ldots, m\}$.
\end{definition}


A \emph{rule} or \emph{inference rule} (as shown below left) is a schema utilizing schematic display sequents that through the application of a substitution $\sub$ produces a \emph{rule instance} or \emph{rule application} (as shown below right). We always assume that if a substitution is applied to a rule that all structure, formula, and atomic variables are within the domain of the substitution.
\begin{center}
\begin{tabular}{c @{\hskip 1em} c}
\AxiomC{$\ssX_{1} \dar \ssY_{1}$}
\AxiomC{$\ldots$}
\AxiomC{$\ssX_{n} \dar \ssY_{n}$}
\TrinaryInfC{$\ssX \dar \ssY$}
\DisplayProof

&

\AxiomC{$(\ssX_{1} \dar \ssY_{1})\sub$}
\AxiomC{$\ldots$}
\AxiomC{$(\ssX_{n} \dar \ssY_{n})\sub$}
\TrinaryInfC{$(\ssX \dar \ssY)\sub$}
\DisplayProof
\end{tabular}
\end{center}

The display calculus $\dkt$ for the minimal tense logic $\kt$ consists of the logical rules given in \fig~\ref{fig:logical-rules-DKt}, the display rules given in \fig~\ref{fig:display-rules-DKt}, and the structural rules given in \fig~\ref{fig:structural-rules-DKt}. Certain logical rules, viz. $\id$, $\topr$, and $\botl$, serve as axioms and are called \emph{initial rules}; we refer to a display sequent generated by an initial rule as an \emph{initial sequent}. The remaining logical rules either introduce a logical connective to a formula, or between formulae, which occur in the premise(s). As usual, we refer to the explicitly presented formula(e) in the premise(s) of a rule as \emph{auxiliary} and refer to the explicitly presented formula in the conclusion of a rule as \emph{principal}. A notable feature of the display formalism is that in each logical rule the principal formula is either the entire antecedent or consequent of the concluding display sequent.\footnote{This fact is utilized in the general cut-elimination theorem~\cite{Bel82}, stated as \thm~\ref{thm:general-cut-elim} below.} \emph{Proofs} in display calculi are defined in the usual inductive way, where each instance of an initial rule is a proof, and further proofs may be obtained by successively applying inference rules~\cite{Kra96}. 
 We use $\dproof$ and annotated versions thereof to denote display proofs. There are various metrics that may be used to quantify aspects of proofs; we define a few notions below that will be helpful in relating display proofs to labeled proofs later on.

\begin{definition}[Quantity, Width, Size]\label{def:hght-wdth-size} We define the \emph{quantity} of a proof $\dproof$ to be the number of sequents it contains, i.e.
\begin{itemize}

\item if $\dproof$ is an initial rule of the form \AxiomC{}\RightLabel{$\ru$}
\UnaryInfC{$\dseq$}
\DisplayProof, then $\qty{\dproof} = 1$;

\medskip

\item if $\dproof =$ \AxiomC{$\dproof_{1} \ldots \dproof_{n}$}\RightLabel{$\ru$}\UnaryInfC{$\dseq$}\DisplayProof, then $\qty{\dproof} = \sum_{i=1}^{n} \qty{\dproof_{i}} + 1$.

\end{itemize}
We define the \emph{width} of a proof $\dproof$ to be equal to the maximum length among all display sequents occurring in the proof, i.e. $
\wdth{\lproof} = max\{\lgth{\dseq} \ | \ \dseq \in \dproof\}.$ Last, we define the \emph{size} of a proof $\dproof$ to be $\size{\dproof} = \qty{\dproof} \times \wdth{\dproof}$.
\end{definition}

\begin{figure}[t]
\noindent\hrule

\begin{center}
\begin{tabular}{c c c c c}

\AxiomC{}
\RightLabel{$\id$}
\UnaryInfC{$\vP \vdash \vP$}
\DisplayProof

&

\AxiomC{}
\RightLabel{$\topr$}
\UnaryInfC{$I \vdash \top$}
\DisplayProof

&

\AxiomC{}
\RightLabel{$\botl$}
\UnaryInfC{$\bot \vdash I$}
\DisplayProof

&

\AxiomC{$I \vdash \vY$}
\RightLabel{$\topl$}
\UnaryInfC{$\top \vdash \vY$}
\DisplayProof

&

\AxiomC{$\vX \vdash I$}
\RightLabel{$\botr$}
\UnaryInfC{$\vX \vdash \bot$}
\DisplayProof
\end{tabular}
\end{center}

\begin{center}
\begin{tabular}{c c c}
\AxiomC{$\dneg\vA \vdash \vY$}
\RightLabel{$\negl$}
\UnaryInfC{$\neg \vA \vdash \vY$}
\DisplayProof

&

\AxiomC{$\vX \vdash \dneg\vA$}
\RightLabel{$\negr$}
\UnaryInfC{$\vX \vdash \neg \vA$}
\DisplayProof

&

\AxiomC{$\vX,\vA \vdash \vB$}
\RightLabel{$\impr$}
\UnaryInfC{$\vX \vdash \vA \rightarrow \vB$}
\DisplayProof
\end{tabular}
\end{center}

\begin{center}
\begin{tabular}{c c c}
\AxiomC{$\ux \vdash \vA$}
\AxiomC{$\vB \vdash \vY$}
\RightLabel{$\impl$}
\BinaryInfC{$\vA \rightarrow \vB \vdash \dneg\vX \circ \vY$}
\DisplayProof

&

\AxiomC{$\vX \vdash \vA \circ \vB$}
\RightLabel{$\disr$}
\UnaryInfC{$\vX \vdash \vA \lor \vB$}
\DisplayProof

&

\AxiomC{$\vA \vdash \vY$}
\AxiomC{$\vB \vdash \vY$}
\RightLabel{$\disl$}
\BinaryInfC{$\vA \lor \vB \vdash \vY$}
\DisplayProof
\end{tabular}
\end{center}

\begin{center}
\begin{tabular}{c c c}
\AxiomC{$\vA \circ \vB \vdash \vY$}
\RightLabel{$\conl$}
\UnaryInfC{$\vA \land \vB \vdash \vY$}
\DisplayProof

&

\AxiomC{$\vX \vdash \vA$}
\AxiomC{$\vX \vdash \vB$}
\RightLabel{$\conr$}
\BinaryInfC{$\vX \vdash \vA \land \vB$}
\DisplayProof

&

\AxiomC{$\vA \vdash \vY$}
\RightLabel{$\gld$}
\UnaryInfC{$\g \vA \vdash \bull \vY$}
\DisplayProof
\end{tabular}
\end{center}

\begin{center}
\begin{tabular}{c c c}
\AxiomC{$\bull \vX \vdash \vA$}
\RightLabel{$\gr$}
\UnaryInfC{$\vX \vdash \g \vA$}
\DisplayProof

&

\AxiomC{$\vA \vdash \dneg \bull \dneg \vY$}
\RightLabel{$\fl$}
\UnaryInfC{$\f \vA \vdash \vY$}
\DisplayProof

&

\AxiomC{$\vX \vdash \vA$}
\RightLabel{$\frd$}
\UnaryInfC{$\dneg \bull \dneg \vX \vdash \f \vA$}
\DisplayProof
\end{tabular}
\end{center}

\begin{center}
\begin{tabular}{c @{\hskip .5em} c @{\hskip .5em} c @{\hskip .5em} c}
\AxiomC{$\dneg \bull \dneg \vX \vdash \vA$}
\RightLabel{$\hr$}
\UnaryInfC{$\vX \vdash \h \vA$}
\DisplayProof

&

\AxiomC{$\vA \vdash \vY$}
\RightLabel{$\hld$}
\UnaryInfC{$\h \vA \vdash \dneg \bull \dneg \vY$}
\DisplayProof

&

\AxiomC{$ \vA \vdash \bull \vY$}
\RightLabel{$\pl$}
\UnaryInfC{$\p \vA \vdash \vY$}
\DisplayProof

&

\AxiomC{$\vX \vdash \vA$}
\RightLabel{$\prd$}
\UnaryInfC{$\bull \vX \vdash \p \vA$}
\DisplayProof

\end{tabular}
\end{center}

\noindent\hrule
\caption{The logical rules for the display calculus $\dkt$~\cite{Kra96}.}\label{fig:logical-rules-DKt}
\end{figure}

The display rules (\fig~\ref{fig:display-rules-DKt}) describe what display sequents may be mutually derived from one another (indicated by the use of a double line); e.g. one may derive $\dneg Y \dar X$ from $\dneg X \dar Y$, or vice-versa. We refer to rules that can be applied top-down and bottom-up as \emph{reversible rules} and indicate them by using a double line. Display rules are the defining characteristic of display calculi and permit data to be shuffled within a display sequent, letting structures be displayed as the entire antecedent or consequent. This gives rise to an equivalence relation defined on sequents called \emph{display equivalence} and a property of display systems called the \emph{display property}. The former relation holds between two display sequents if they can be derived from one another by means of the display rules, and the meaning of the latter property is that one can always pick a substructure $Z$ within a display sequent $X \dar Y$ and transform the sequent into a display equivalent version of the form $Z \dar W$ or $W \dar Z$. Both of these notions are formally defined below. 

\begin{definition}[Display Equivalence, Display Property~\cite{Bel82}]\label{def:dis-equiv-property} Two display sequents $X \dar Y$ and $Z \dar W$ are \emph{display equivalent}, written $X \dar Y \dequiv Z \dar W$, \ifandonlyif they are mutually derivable from one another by means of 
 the display rules. We also define two display sequents $X \dar Y$ and $Z \dar W$ to be \emph{deductively equivalent}, written $X \dar Y \equiv Z \dar W$, \ifandonlyif they are mutually derivable from one another.

A display calculus has the \emph{display property} \ifandonlyif the calculus includes a set of rules (called \emph{display rules}) such that for any display sequent $X \vdash Y$, if $Z$ is a substructure of $X$ or $Y$, then either $Z \vdash W$ or $W \vdash Z$ are derivable for some structure $W$ using only the display rules. We note that in our setting the rules of \fig~\ref{fig:display-rules-DKt} are taken to be the display rules.
\end{definition}

The capacity to display a substructure as the entire antecedent or consequent of a sequent gives rise to two separate types of substructures, namely, \emph{antecedent parts (a-parts)} and \emph{consequent parts (c-parts)}. An a-part is a substructure that can be displayed as the entire antecedent, and a c-part is a substructure that can be displayed as the entire consequent. Formally, they are defined as follows.

\begin{definition}[A-part, C-part~\cite{Bel82}] 
 A substructure $Z$ of $X$ or $Y$ is an \emph{antecedent part (a-part)} \ifandonlyif for some structure $W$, $Z \vdash W$ is derivable from $X \vdash Y$ using only the display rules. A substructure $Z$ is defined to be a \emph{consequent part (c-part)} \ifandonlyif for some structure $W$, $W \vdash Z$ is derivable from $X \vdash Y$ using only the display rules.
\end{definition}

\begin{figure}[t]
\noindent\hrule

\begin{center}
\begin{tabular}{c c c}
\AxiomC{$\vX \circ \vY \vdash \vZ$}
\doubleLine
\RightLabel{$\dri$}
\UnaryInfC{$\vX \vdash \vZ  \circ \dneg\vY$}
\DisplayProof

&

\AxiomC{$\vX \circ \vY \vdash \vZ$}
\doubleLine
\RightLabel{$\drii$}
\UnaryInfC{$\vY \vdash \dneg\vX \circ \vZ$}
\DisplayProof

&

\AxiomC{$\vX \vdash \vY \circ \vZ$}
\doubleLine
\RightLabel{$\driii$}
\UnaryInfC{$\vX \circ \dneg\vZ \vdash \vY$}
\DisplayProof
\end{tabular}
\end{center}

\begin{center}
\begin{tabular}{c c c}
\AxiomC{$\vX \vdash \vY \circ \vZ$}
\doubleLine
\RightLabel{$\driv$}
\UnaryInfC{$\dneg\vY \circ \vX \vdash \vZ$}
\DisplayProof

&

\AxiomC{$\dneg\vX \vdash \vY$}
\doubleLine
\RightLabel{$\drv$}
\UnaryInfC{$\dneg\vY \vdash \vX$}
\DisplayProof

&

\AxiomC{$\vX \vdash \dneg\vY$}
\doubleLine
\RightLabel{$\drvi$}
\UnaryInfC{$\vY \vdash \dneg\vX$}
\DisplayProof
\end{tabular}
\end{center}

\begin{center}
\begin{tabular}{c c c}
\AxiomC{$\dneg\dneg\vX \vdash \vY$}
\doubleLine
\RightLabel{$\drvii$}
\UnaryInfC{$\vX \vdash \vY$}
\DisplayProof

&

\AxiomC{$\vX \vdash \dneg\dneg\vY$}
\doubleLine
\RightLabel{$\drviii$}
\UnaryInfC{$\vX \vdash \vY$}
\DisplayProof

&

\AxiomC{$\vX \vdash \bull \vY$}
\doubleLine
\RightLabel{$\drix$}
\UnaryInfC{$\bull \vX \vdash \vY$}
\DisplayProof

\end{tabular}
\end{center}

\noindent\hrule
\caption{The display rules for the display calculus $\dkt$~\cite{Kra96}.}\label{fig:display-rules-DKt}
\end{figure}

\begin{example} We give an example of the substructures of the structures $\bull (\dneg A \circ B)$ and $\dneg I \circ (C \circ D)$, as well as consider an example of a-parts and c-parts in a given display sequent.

\begin{itemize}

\item $\sst{\bull (\dneg A \circ B)} = \{\bull (\dneg A \circ B), \dneg A \circ B, \dneg A, A, B\}$ 

\item $\sst{\dneg I \circ (C \circ D)} = \{\dneg I \circ (C \circ D), \dneg I, I, C \circ D, C,D\}$

\end{itemize}

All display sequents below are display equivalent to one another since all are mutually derivable by means of the rules in \fig~\ref{fig:display-rules-DKt}. Moreover, we can see in the left derivation that the structures $\bull (\dneg A \circ B)$, $\dneg A \circ B$, and $B$ are a-parts, and in the right derivation $\dneg I \circ (C \circ D)$, $C \circ D$, and $C$ are c-parts.

\begin{center}
\begin{tabular}{c @{\hskip 1em} c}
\AxiomC{$\bull (\dneg A \circ B) \dar \dneg I \circ (C \circ D)$}
\doubleLine
\UnaryInfC{$\dneg A \circ B \dar \bull (\dneg I \circ (C \circ D))$}
\doubleLine
\UnaryInfC{$B \dar A \circ \bull (\dneg I \circ (C \circ D))$}
\DisplayProof

&

\AxiomC{$\bull (\dneg A \circ B) \dar \dneg I \circ (C \circ D)$}
\doubleLine
\UnaryInfC{$\bull (\dneg A \circ B) \circ I \dar C \circ D$}
\doubleLine
\UnaryInfC{$(\bull (\dneg A \circ B) \circ I) \circ \dneg D \dar C$}
\DisplayProof
\end{tabular}
\end{center}
\end{example}

The structural rules for $\dkt$ (\fig~\ref{fig:structural-rules-DKt}) ensure the proper behavior of our logical connectives; e.g. $\asl$ encodes the fact that $\land$ is associative (since by \dfn~\ref{def:formula-translation-display} the $\circ$ structural connective represents $\land$ in the antecedent) and $\pmr$ encodes the fact that $\lor$ is commutative (since by \dfn~\ref{def:formula-translation-display} the $\circ$ structural connective represents $\lor$ in the consequent). Additionally, by the work of Belnap~\cite{Bel82}, we know that the $\cut$ rule is eliminable in $\dkt$ since the calculus satisfies eight sufficient conditions necessitating the elimination of $\cut$. 

\begin{definition}[Conditions (C1)--(C8)]\label{def:C1-C8} Belnap's eight sufficient conditions ensuring the elimination of $\cut$ are as follows: 

\begin{itemize}

\item[(C1)] Each formula occurring in the premise of a rule instance is a subformula of some formula in the conclusion of the inference.

\item[(C2)] We say that a structure variable in the premise of an inference is \emph{congruent} to a structure variable in the conclusion iff the two structure variables are identical.

\item[(C3)] Each structure variable in the premise of an inference is congruent to at most one structure variable in the conclusion.

\item[(C4)] Congruent structure variables in an inference are either both a-parts or both c-parts.

\item[(C5)] For each rule, if a formula variable is in the conclusion of the rule, then it is either the entire antecedent or consequent.

\item[(C6)] Each rule is closed under the uniform substitution of arbitrary structures in c-parts for congruent structure variables.

\item[(C7)] Each rule is closed under the uniform substitution of arbitrary structures in a-parts for congruent structure variables.

\item[(C8)] Suppose there is an inference ending with $X \vdash A$ and an inference ending with $A \vdash Y$ with $A$ principal in both inferences. Then, (1) $X \vdash Y$ is identical to $X \vdash A$ or $A \vdash Y$, or (2) there exists a proof of $X \vdash Y$ from the premises of each inference where $\cut$ is only used on proper subformulae of $A$.
\end{itemize}
\end{definition}

\begin{theorem}[General Cut-Elimination~\cite{Bel82}]\label{thm:general-cut-elim} Any display calculus satisfying conditions (C1)--(C8) has the subformula property and admits $\cut$ elimination.
\end{theorem}

Although we have been discussing Kracht's display calculus $\dkt$ for the minimal tense logic $\kt$, the above cut-elimination theorem applies equally to extensions of $\dkt$ with rules corresponding to (simplified) primitive tense axioms. In order to define the rules capturing these axioms (referred to as \emph{primitive tense structural rules}), we first define a translation from logical formulae to schematic structures. The reason being, we obtain primitive tense structural rules by translating simplified primitive tense axioms into such rules, and the translation of logical formulae into schematic structures is a crucial component of this process.

\begin{figure}[t]
\noindent\hrule

\begin{center}
\begin{tabular}{c c c c}
\AxiomC{$\vX \vdash \vY$}
\doubleLine
\RightLabel{$\il$}
\UnaryInfC{$I \circ \vX \vdash \vY$}
\DisplayProof

&

\AxiomC{$\vX \vdash \vY$}
\doubleLine
\RightLabel{$\ir$}
\UnaryInfC{$\vX \vdash I \circ \vY$}
\DisplayProof

&

\AxiomC{$I \vdash \vY$}
\doubleLine
\RightLabel{$\ql$}
\UnaryInfC{$\dneg I \vdash \vY$}
\DisplayProof

&

\AxiomC{$\vX \vdash I$}
\doubleLine
\RightLabel{$\qr$}
\UnaryInfC{$\vX \vdash \dneg I$}
\DisplayProof
\end{tabular}
\end{center}

\begin{center}
\begin{tabular}{c c c}
\AxiomC{$\vX \vdash \vY$}
\RightLabel{$\dwl$}
\UnaryInfC{$\uz \circ \vX \vdash \vY$}
\DisplayProof

&

\AxiomC{$\vX \vdash \vY$}
\RightLabel{$\dwr$}
\UnaryInfC{$\vX \vdash \vY \circ \vZ$}
\DisplayProof

&

\AxiomC{$\vX \circ (\vY \circ \vZ) \vdash \vW$}
\doubleLine
\RightLabel{$\asl$}
\UnaryInfC{$(\vX \circ \vY) \circ \vZ \vdash \vW$}
\DisplayProof
\end{tabular}
\end{center}

\begin{center}
\begin{tabular}{c c c c}
\AxiomC{$\vX \vdash \vY \circ (\vZ \circ \vW)$}
\doubleLine
\RightLabel{$\asr$}
\UnaryInfC{$\vX \vdash (\vY \circ \vZ) \circ \vW$}
\DisplayProof

&

\AxiomC{$\vX \circ \vY \vdash \vZ$}
\RightLabel{$\pml$}
\UnaryInfC{$\vY \circ \vX \vdash \vZ$}
\DisplayProof

&

\AxiomC{$\vX \vdash \vY \circ \vZ$}
\RightLabel{$\pmr$}
\UnaryInfC{$\vX \vdash \vZ \circ \vY$}
\DisplayProof

&

\AxiomC{$\vX \circ \vX \vdash \vY$}
\RightLabel{$\dcl$}
\UnaryInfC{$\vX \vdash \vY$}
\DisplayProof
\end{tabular}
\end{center}

\begin{center}
\begin{tabular}{c c c c}
\AxiomC{$\vX \vdash \vY \circ \vY$}
\RightLabel{$\dcr$}
\UnaryInfC{$\vX \vdash \vY$}
\DisplayProof

&

\AxiomC{$I \vdash \vY$}
\RightLabel{$\ml$}
\UnaryInfC{$\bull I \vdash \vY$}
\DisplayProof

&

\AxiomC{$\vX \vdash I$}
\RightLabel{$\mr$}
\UnaryInfC{$\vX \vdash \bull I$}
\DisplayProof

&

\AxiomC{$\vX \vdash A$}
\AxiomC{$A \vdash \vY$}
\RightLabel{$\cut$}
\BinaryInfC{$\vX \vdash \vY$}
\DisplayProof
\end{tabular}
\end{center}

\noindent\hrule
\caption{The structural rules for the display calculus $\dkt$~\cite{Kra96}.}\label{fig:structural-rules-DKt}
\end{figure}

\begin{definition}[Structure Translation $\psi$~\cite{Kra96}] We define the translation function $\psi$ from formulae to schematic structures as follows:
\begin{itemize}

\item $\psi(\top) := I$

\item $\psi(p) := \vX_{p}$

\item $\psi(A \land B) := \psi(A) \circ \psi(B)$ 



\item $\psi(\f A) := \dneg \bull \dneg \psi(A)$

\item $\psi(\p A) := \bull \psi(A)$

\end{itemize}

\end{definition}

In the above definition, each atom $p$ is transformed into a unique structure variable $\vX_{p}$. Following Kracht~\cite{Kra96}, we make use of the translation $\psi$ to define primitive tense structural rules.

\begin{definition}[$\dktp$] Each simplified primitive tense axiom of the form shown below left corresponds to a primitive tense structural rule of the form shown below right.
\begin{center}
\begin{tabular}{c @{\hskip 2em} c}
$
\displaystyle{A \rightarrow \bigvee_{1 \leq j \leq m} B_{j}}
$

&

\AxiomC{$\psi(B_{1}) \dar \vX$}
\AxiomC{$\ldots$}
\AxiomC{$\psi(B_{m}) \dar \vX$}
\RightLabel{$\ptsrd$}
\TrinaryInfC{$\psi(A) \dar \vX$}
\DisplayProof
\end{tabular}
\end{center}
If $\paxs$ is a set of simplified primitive tense axioms, then we define $\dktp$ to be the extension of $\dkt$ with all primitive tense structural rules corresponding to the axioms in $\paxs$. We indicate that a display sequent $\dseq$ is provable in $\dktp$ with a display proof $\dproof$ by writing $\proves{\dktp}{\dproof}{\dseq}$. Note that $\dkt$ is equal to $\dktp$ where $\paxs = \emptyset$.
\end{definition}


\begin{theorem}[Display Theorem~\cite{Kra96}]\label{thm:display-theorem}
Each display calculus $\dktp$ possesses the display property.
\end{theorem}

The proof-theoretic significance of primitive tense axioms was identified by Kracht in~\cite{Kra96}. As discussed and proven there, primitive tense extensions of $\kt$ form the largest class of tense logics whose corresponding display calculus satisfies conditions (C1)--(C8), and---conversely---if a display calculus satisfies Belnap's conditions, then it is sound and complete relative to a primitive tense extension of $\kt$. 
 This result is stated in \thm~\ref{thm:display-theorem-I} below, and relies on the notion of a tense logic being \emph{properly displayed}.

\begin{definition}[Properly Displayed~\cite{Kra96}] Let $\srus$ be a set of structural rules such that $\dkt + \srus$ satisfies conditions (C1)---(C8) and let $\kt + \axs$ be the logic obtained from extending the Hilbert calculus $\hil\kt$ with the axioms in $\axs$. Then, we say that a display calculus $\dkt + \srus$ \emph{properly displays} $\kt + \axs$ \ifandonlyif it satisfies conditions (C1)--(C8) along with conditions (1) and (2) below.

(1) For every rule instance of the form shown below left in $\dkt + \srus$ the rule instance of the form shown below right is derivable in $\kt + \axs$: 
\begin{flushleft}
\AxiomC{$(\ssX_{1} \dar \ssY_{1})\sub$}
\AxiomC{$\ldots$}
\AxiomC{$(\ssX_{n} \dar \ssY_{n})\sub$}
\TrinaryInfC{$(\ssX \dar \ssY)\sub$}
\DisplayProof
\end{flushleft}
\begin{flushright}
\AxiomC{$\ftd((\ssX_{1} \dar \ssY_{1})\sub)$}
\AxiomC{$\ldots$}
\AxiomC{$\ftd((\ssX_{n} \dar \ssY_{n})\sub)$}
\TrinaryInfC{$\ftd((\ssX \dar \ssY)\sub)$}
\DisplayProof
\end{flushright}

(2) For every rule instance of the form shown below left (where we assume that an axiom is a rule instance with zero premises) in $\kt + \axs$, the rule instance shown below right is derivable in $\dkt + \srus$.
\begin{center}
\begin{tabular}{c @{\hskip 1em} c}
\AxiomC{$\vA_{1}\sub$}
\AxiomC{$\ldots$}
\AxiomC{$\vA_{n}\sub$}
\TrinaryInfC{$\vA\sub$}
\DisplayProof

&

\AxiomC{$I \dar \vA_{1}\sub$}
\AxiomC{$\ldots$}
\AxiomC{$I \dar \vA_{n}\sub$}
\TrinaryInfC{$I \dar \vA\sub$}
\DisplayProof
\end{tabular}
\end{center}
\end{definition}

As a consequence of the above definition, if a display calculus properly displays a logic, then it is sound and complete relative to that logic.

\begin{theorem}[Display Theorem I~\cite{Kra96}]\label{thm:display-theorem-I} Let $\srus$ be a set of structural rules such that $\dkt + \srus$ satisfies conditions (C1)---(C8) and let $\kt + \axs$ be the logic obtained from extending the Hilbert calculus $\hil\kt$ with the axioms in $\axs$. The display calculus $\dkt + \srus$ properly displays $\kt + \axs$ \ifandonlyif $\kt + \axs$ is axiomatizable with primitive tense axioms.
\end{theorem}


\fig~\ref{fig:derivable-rules-display} shows some derivable rules in $\dktp$, which will be of use in our translation work later on. All rules are derivable by means of display and/or structural rules.

\begin{figure}[t]

\noindent\hrule

\begin{center}
\begin{tabular}{c @{\hskip 1em} c @{\hskip 1em} c}
\AxiomC{$\dneg \bull \dneg \vX \dar \vY$}
\RightLabel{$\deri$}
\doubleLine
\UnaryInfC{$\vX \dar \dneg \bull \dneg \vY$}
\DisplayProof

&

\AxiomC{$\vX \dar \bull \vY \circ \bull \vZ$}
\RightLabel{$\derii$}
\UnaryInfC{$\vX \dar \bull (\vY \circ \vZ)$}
\DisplayProof

&

\AxiomC{$\vX \dar \dneg \bull \dneg \vY \circ \dneg \bull \dneg \vZ$}
\RightLabel{$\deriii$}
\UnaryInfC{$\vX \dar \dneg \bull \dneg (\vY \circ \vZ)$}
\DisplayProof
\end{tabular}
\end{center}

\begin{center}
\begin{tabular}{c @{\hskip 1em} c}
\AxiomC{$\bull \vX \circ \bull \vY \dar \vZ$}
\RightLabel{$\deriv$}
\UnaryInfC{$\bull (\vX \circ \vY) \dar \vZ$}
\DisplayProof

&

\AxiomC{$\dneg \bull \dneg \vX \circ \dneg \bull \dneg \vY \dar \vZ$}
\RightLabel{$\derv$}
\UnaryInfC{$\dneg \bull \dneg (\vX \circ \vY) \dar \vZ$}
\DisplayProof
\end{tabular}
\end{center}

\noindent\hrule

\caption{Derivable rules in $\dktp$.\label{fig:derivable-rules-display}}
\end{figure}

\begin{proposition}\label{prop:dis-derivable-rules}
The rules in \fig~\ref{fig:derivable-rules-display} are derivable in $\dktp$.
\end{proposition}

Last, we note that throughout the remainder of the paper we use distinct types of inference lines to indicate distinct types of inferences. In particular, a solid line `\begin{tikzpicture} \draw (-.25,0) -- (.25,0); \end{tikzpicture}' (as shown first below) is used to indicate a rule application, a dashed line `\begin{tikzpicture} \draw[dashed] (-.25,0) -- (.25,0); \end{tikzpicture}' (as shown second below) is used to indicate the application of a derivable or admissible rule, a double line `\begin{tikzpicture} \draw[double] (-.25,0) -- (.25,0); \end{tikzpicture}' (as shown third below) is used to indicate a reversible rule application, and a dotted line `\begin{tikzpicture} \draw[loosely dotted, thick] (-.275,0) -- (.275,0); \end{tikzpicture}' (as shown fourth below) is used to indicate that the premise and conclusion are identical or isomorphic. (NB. We define isomorphisms between sequents in the next section.) In the notation below, each $P$ and $P_{i}$ denote a premise and each $C$ is the conclusion; these may be display or labeled sequents (defined in the next section).
\begin{center}
\begin{tabular}{c @{\hskip 2em} c @{\hskip 2em} c @{\hskip 2em} c}
\AxiomC{$P_{1}$}
\AxiomC{$\ldots$}
\AxiomC{$P_{n}$}
\TrinaryInfC{$C$}
\DisplayProof

&

\AxiomC{$P_{1}$}
\AxiomC{$\ldots$}
\AxiomC{$P_{n}$}
\dashedLine
\TrinaryInfC{$C$}
\DisplayProof

&

\AxiomC{$P$}
\doubleLine
\UnaryInfC{$C$}
\DisplayProof

&

\AxiomC{$P$}
\dottedLine
\UnaryInfC{$C$}
\DisplayProof
\end{tabular}
\end{center}

\section{Labeled Calculi for Tense Logics}\label{section-3}


Labeled sequents generalize the syntax of Gentzen-style sequents through the incorporation of labels and semantic elements. This idea is rooted in the work of Kanger~\cite{Kan57}, who employed \emph{spotted formulae} in the construction of sequent systems for modal logics. Since then, large classes of modal and constructive logics have been supplemented with labeled sequent systems~\cite{CasSma02,Gab96,Min97,Sim94,Vig00}. The labeled formalism has been successful in generating modular systems that cover extensive classes of logics in a uniform manner, i.e. through the inclusion or exclusion of structural rules, one labeled system for a logic may be transformed into a labeled system for another logic~\cite{Sim94,Vig00}. Moreover, general results exist (e.g.~\cite{Bor08,Hei05}) which show that labeled systems commonly possess favorable properties such as admissible structural rules, invertible logical rules, and cut-admissibility (we formally define these properties below). In this section, we introduce labeled calculi with primitive tense structural rules. The base calculus for the logic $\kt$ is a notational variant of Boretti's~\cite{Bor08} labeled calculus $\gtkt$, though the primitive tense structural rules we define are entirely new.\footnote{We employ a schematic representation of labeled calculi whereby rules are instantiated by means of substitutions to better match the display formalism.}

We let $\lab = \{w, u, v, \ldots \}$ be a denumerable set of \emph{labels}. A \emph{labeled sequent} is defined to be a formula of the form $\rel, \Gamma \sar \Delta$, where $\rel$ is a (potentially empty) set of \emph{relational atoms} of the form $Rwu$, and $\Gamma$ and $\Delta$ are (potentially empty) multisets of \emph{labeled formulae} of the form $w : A$, where $w$ and $u$ range over $\lab$, and $A$ ranges over $\lang$.\footnote{Usually labeled calculi employ \emph{multisets} of relational atoms. However, since weakenings and contractions are admissible over relational atoms in our setting, it is well-known that sets can be used instead of multisets.} We use $\rel$ to denote sets of relational atoms, $\Gamma$ and $\Delta$ to denote multisets of labeled formulae, and $\lseq$ to denote labeled sequents, as well as use annotated versions of these symbols. We let $\lab(\rel)$, $\lab(\Gamma)$, and $\lab(\lseq)$ indicate the set of all labels in a set $\rel$, a multiset $\Gamma$, or a labeled sequent $\lseq$, respectively. We define the \emph{length} of a labeled sequent $\rel, \Gamma \sar \Delta$ to be $\lgth{\rel, \Gamma \sar \Delta} = |\rel| + |\Gamma| + |\Delta|$, where $|\rel|$, $|\Gamma|$, and $|\Delta|$ denote the cardinality of each (multi)set.

 Given two labeled sequents $\lseq_{1} := \rel_{1}, \Gamma_{1} \sar \Delta_{1}$ and $\lseq_{2} := \rel_{2}, \Gamma_{2} \sar \Delta_{2}$, we define the \emph{sequent composition} $\lseq_{1} \seqcomp \lseq_{2} := \rel_{1}, \rel_{2},\Gamma_{1}, \Gamma_{2} \sar \Delta_{1}, \Delta_{2}$. It follows by definition that sequent compositions are (1) associative, i.e. $(\lseq_{1} \seqcomp \lseq_{2}) \seqcomp \lseq_{3} = \lseq_{1} \seqcomp (\lseq_{2} \seqcomp \lseq_{3})$, and (2) commutative, i.e. $\lseq_{1} \seqcomp \lseq_{2} = \lseq_{2} \seqcomp \lseq_{1}$. Due to the associative property, we will often omit parentheses when writing successive sequent compositions. If $\{A_{1}, \ldots, A_{n}\} \subset \lang$, we define $w : \{A_{1}, \ldots, A_{n}\} := w : A_{1}, \ldots, w : A_{n}$.

As in the display setting, we utilize variables in the formulation of our inference rules and calculi. In particular, we let $\lw, \lu, \lv, \ldots$ denote \emph{label variables}, which will be instantiated with labels, we define \emph{schematic relational atoms} to be formulae of the form $R\lw\lu$ with $\lw$ and $\lu$ label variables, and we define \emph{schematic labeled formulae} to be formulae of the form $\lw : \ssF$ where $\lw$ is a label variable and $\ssF$ is generated by the following grammar in BNF:
$$
\ssF ::= \vP \ | \ \top \ | \ \bot \ | \ \vA \ | \ {\triangledown} \vA \ | \ \vA \odot \vB
$$
 with $\triangledown \in \{\neg, \g, \f, \h, \p\}$, $\odot \in \{\land, \lor, \imp\}$, $\vP$ an atomic variable, and $\vA, \vB$ formula variables. We use $\vL$ and annotated versions thereof to denote \emph{labeled sequent variables}, which will be instantiated with labeled sequents. Last, we define \emph{schematic labeled sequents} to be formulae of the form
$$
(\dot{\rel}_{1},\dot{\Gamma}_{1} \sar \dot{\Delta}_{1}) \seqcomp \cdots \seqcomp (\dot{\rel}_{n},\dot{\Gamma}_{n} \sar \dot{\Delta}_{n}) \seqcomp \vL_{1} \seqcomp \cdots \seqcomp \vL_{k}
$$
 where, for each $1 \leq i \leq n$ and $1 \leq j \leq k$, $\dot{\rel}_{i}$ is a (potentially empty) set of schematic relational atoms, $\dot{\Gamma}_{i}$ and $\dot{\Delta}_{i}$ are (potentially empty) multisets of schematic labeled formulae, and $\vL_{j}$ is a labeled sequent variable. We sometimes denote schematic labeled sequents by $\sls$ and annotated versions thereof. As can be seen in \figs~\ref{fig:labeled-logical-rules}--\ref{fig:labeled-structural-rules}, schematic labeled sequents are used to define inference rules, instances of which, are obtained via applications of substitutions.

\begin{definition}[Substitution $\sub$] A \emph{substitution} $\sub$ is defined to be a function (written in postfix notation) that satisfies the following: $\lw\sub \in \lab$, $\vP\sub \in \propset$, $\vA\sub \in \lang$, and $\vL\sub$ is a labeled sequent. We extend a substitution $\sub$ to sets of schematic relational atoms, multisets of schematic labeled formulae, and schematic labeled sequents in the expected way by applying $\sub$ to each label variable, atomic variable, formula variable, and labeled sequent variable occurring therein.

If we want to indicate what variables are mapped to which labeled sequents, formulae, atoms, or labels then we write a substitution $\sub$ as
$$
[\lseq_{1} / \vL_{1}, \ldots, \lseq_{r} / \vL_{r}, A_{1} / \vA_{1}, \ldots, A_{n} / \vA_{n}, p_{1} / \vP_{1}, \ldots, p_{k} / \vP_{k}, w_{1} / \lw_{1}, \ldots, w_{m} / \lw_{m}]
$$
 indicating that $\vL_{s}\sub = \lseq_{s}$, $\vA_{i}\sub = A_{i}$, $\vP_{j}\sub = p_{j}$, and $\lw_{t}\sub = w_{t}$ for $s \in \{1, \ldots, r\}$, $i \in \{1, \ldots, n\}$, $j \in \{1, \ldots, k\}$, and $t \in \{1, \ldots, m\}$.
\end{definition}

\begin{figure}[t]
\noindent\hrule

\begin{center}
\begin{tabular}{c c c} 
\AxiomC{}
\RightLabel{$\id$}
\UnaryInfC{$(\lw : \vP \sar \lw : \vP) \seqcomp \vL$}
\DisplayProof

&

\AxiomC{}
\RightLabel{$\botl$}
\UnaryInfC{$(\lw :\bot \sar ) \seqcomp \vL$}
\DisplayProof

&

\AxiomC{}
\RightLabel{$\topr$}
\UnaryInfC{$( \sar \lw : \top) \seqcomp \vL$}
\DisplayProof
\end{tabular}
\end{center}

\begin{center}
\begin{tabular}{c c}

\AxiomC{$( \sar \lw : \vA) \seqcomp \vL$}
\RightLabel{$\negl$}
\UnaryInfC{$( \lw : \neg \vA \sar ) \seqcomp \vL$}
\DisplayProof

&

\AxiomC{$(\lw : \vA \sar ) \seqcomp \vL$}
\RightLabel{$\negr$}
\UnaryInfC{$( \sar \lw : \neg \vA) \seqcomp \vL$}
\DisplayProof

\end{tabular}
\end{center}

\begin{center}
\begin{tabular}{c c}

\AxiomC{$(\lw : \vA, \lw : \vB \sar ) \seqcomp \vL$}
\RightLabel{$\conl$}
\UnaryInfC{$(\lw : \vA \wedge \vB  \sar ) \seqcomp \vL$}
\DisplayProof

&

\AxiomC{$( \sar \lw : \vA) \seqcomp \vL$}
\AxiomC{$( \sar \lw : \vB) \seqcomp \vL$}
\RightLabel{$\conr$}
\BinaryInfC{$( \sar \lw : \vA \wedge \vB) \seqcomp \vL$}
\DisplayProof

\end{tabular}
\end{center}

\begin{center}
\begin{tabular}{c @{\hskip 1em} c}
\AxiomC{$(\lw : \vA \sar ) \seqcomp \vL$}
\AxiomC{$(\lw : \vB \sar ) \seqcomp \vL$}
\RightLabel{$\disl$}
\BinaryInfC{$(\lw : \vA \vee \vB \sar ) \seqcomp \vL$}
\DisplayProof

&

\AxiomC{$( \sar \lw : \vA, \lw : \vB) \seqcomp \vL$}
\RightLabel{$\disr$}
\UnaryInfC{$( \sar \lw : \vA \vee \vB) \seqcomp \vL$}
\DisplayProof
\end{tabular}
\end{center}

\begin{center}
\begin{tabular}{c c}
\AxiomC{$( \sar \lw : \vA) \seqcomp \vL$}
\AxiomC{$( \lw : \vB \sar ) \seqcomp \vL$}
\RightLabel{$\impl$}
\BinaryInfC{$(\lw : \vA \imp \vB \sar ) \seqcomp \vL$}
\DisplayProof 

&

\AxiomC{$( \lw : \vA \sar \lw : \vB) \seqcomp \vL$}
\RightLabel{$\impr$}
\UnaryInfC{$( \sar \lw : \vA \imp \vB) \seqcomp \vL$}
\DisplayProof
\end{tabular}
\end{center}

\hrule
\caption{Initial and logical rules of $\gtkt$.}
\label{fig:labeled-logical-rules}
\end{figure}

 It follows from the above definition that applying a substitution to a schematic labeled sequent yields a labeled sequent as the result. We define a \emph{rule} or \emph{inference rule} to be a schema utilizing schematic labeled sequents, which produces a \emph{rule instance} or \emph{rule application} when a substitution is applied to each schematic labeled sequent occurring in the rule (see \eg~\ref{eg:labeled-rule-application} below). When a substitution $\sub$ is applied to a rule, we assume that every label variable, atomic variable, formula variable, and labeled sequent variable occurring in the rule is within the domain of $\sub$.

The labeled calculus $\gtkt$ is defined to be set of the inference rules presented in \figs~\ref{fig:labeled-logical-rules} and \ref{fig:labeled-modal-rules}. $\gtkt$ contains the initial rules $\id$, $\botl$, and $\topr$, as well as logical rules for $\neg$, $\lor$, $\land$, and $\imp$, all of which are shown in \fig~\ref{fig:labeled-logical-rules}. We refer to each rule $\id$, $\botl$, and $\topr$ as an \emph{initial rule} and refer to any labeled sequent generated by an initial rule as an \emph{initial sequent}. The rules governing the introduction of modal formulae are provided in \fig~\ref{fig:labeled-modal-rules}, and we note that $\gr$, $\hr$, $\fl$, and $\pl$ have a side condition, namely, for any application of the rule with a substitution $\sub$, $\lu\sub$ must be \emph{fresh}, i.e must not occur in the conclusion of the rule application. As in the display setting, the explicitly presented formula(e) in the premise(s) of a rule is (are) referred to as \emph{auxiliary} and the explicitly presented formula(e) in the conclusion is (are) called \emph{principal}. We also note that the labeled sequent variable $\vL$ in each inference rule provides the \emph{context} of the rule when instantiated. This formulation of labeled sequent rules is new and will be helpful in studying and defining translations with display calculi later on since substitutions may be explicitly taken into account in translations. 

\begin{example}\label{eg:labeled-rule-application} To provide the reader with intuition concerning the application of rules, we give two examples. First, let $\lseq = (\rel, \Gamma \sar \Delta)$ and $\sub = [\lseq / \vL, q / \vP, u / \lw]$. Applying $\sub$ to the rule $\id$ (shown below left) yields the rule instance below right.
\begin{center}
\begin{tabular}{c c c}
\AxiomC{}
\RightLabel{$\id$}
\UnaryInfC{$((\lw : \vP \sar \lw : \vP) \seqcomp \vL)\sub$}
\DisplayProof

&

$=$

&

\AxiomC{}
\RightLabel{$\id$}
\UnaryInfC{$(u : q \sar u : q) \seqcomp (\rel, \Gamma \sar \Delta)$}
\RightLabel{$=$}
\dottedLine
\UnaryInfC{$\rel, \Gamma, u : q \sar u : q, \Delta$}
\DisplayProof
\end{tabular}
\end{center}
Second, let $\lseq = (\rel, \Gamma \sar \Delta)$ and $\sub = [\lseq / \vL, \p p / \vA, w / \lw, v / \lu]$. Applying $\sub'$ to the rule $\gr$ (shown below left) yields the rule instance below right.
\begin{center}
\begin{tabular}{c c c}
\AxiomC{$((R\lw\lu \sar \lu : \vA) \seqcomp \vL)\sub'$}
\RightLabel{$\gr$}
\UnaryInfC{$(( \sar \lw : \g \vA) \seqcomp \vL)\sub'$}
\DisplayProof

&

$=$

&

\AxiomC{}
\RightLabel{$\id$}
\UnaryInfC{$(Rwv \sar v : \p p) \seqcomp (\rel, \Gamma \sar \Delta)$}
\RightLabel{$=$}
\dottedLine
\UnaryInfC{$\rel, Rwv, \Gamma, v : \p q \sar \Delta$}
\RightLabel{$\gr$}
\UnaryInfC{$\rel, \Gamma, w : \g \p q \sar \Delta$}
\RightLabel{$=$}
\dottedLine
\UnaryInfC{$( \sar w : \g \p p) \seqcomp (\rel, \Gamma \sar \Delta)$}
\DisplayProof
\end{tabular}
\end{center}
\end{example}


\emph{Proofs} are defined in the usual, inductive way: each instance of an initial rule is a proof, and applying a rule to the conclusion of a proof, or between conclusions of proofs, yields a proof. We use $\lproof$ and annotated versions thereof to denote labeled proofs. A labeled sequent is \emph{derivable} \iffi it is the conclusion of a proof. As with display proofs, we use various metrics to quantify the size of, or certain aspects of, proofs. We define the \emph{quantity}, \emph{width}, and \emph{size} of a proof just as in \dfn~\ref{def:hght-wdth-size} (though relative to labeled proofs), i.e. the quantity of a proof is equal to the number of sequents it contains, the width is defined to be the maximal length among all labeled sequents occurring in the proof, and the size is defined to be the product of the quantity and the width of the proof. For a labeled proof $\lproof$, we let $\qty{\lproof}$ denote the quantity, $\wdth{\lproof}$ denote the width, and $\size{\lproof}$ denote the size.

\begin{figure}[t]
\noindent\hrule

\begin{center}
\begin{tabular}{c c}
\AxiomC{$(R\lw\lu, \lu : \vA \sar ) \seqcomp \vL$}
\RightLabel{$\fl^{\dag}$}
\UnaryInfC{$(\lw : \f \vA \sar ) \seqcomp \vL$}
\DisplayProof

&

\AxiomC{$(R\lu\lw, \lu : \vA \sar ) \seqcomp \vL$}
\RightLabel{$\pl^{\dag}$}
\UnaryInfC{$(\lw : \p \vA \sar ) \seqcomp \vL$}
\DisplayProof
\end{tabular}
\end{center}

\begin{center}
\begin{tabular}{c c}
\AxiomC{$(R\lw\lu \sar \lu : \vA) \seqcomp \vL$}
\RightLabel{$\gr^{\dag}$}
\UnaryInfC{$( \sar \lw : \g \vA) \seqcomp \vL$}
\DisplayProof

&

\AxiomC{$(R\lu\lw \sar \lu : \vA) \seqcomp \vL$}
\RightLabel{$\hr^{\dag}$}
\UnaryInfC{$( \sar \lw : \h \vA) \seqcomp \vL$}
\DisplayProof
\end{tabular}
\end{center}

\begin{center}
\begin{tabular}{c c}
\AxiomC{$(R\lw\lu, \lw : \g \vA, \lu : \vA \sar ) \seqcomp \vL$}
\RightLabel{$\gl$}
\UnaryInfC{$(R\lw\lu, \lw : \g \vA \sar ) \seqcomp \vL$}
\DisplayProof

&

\AxiomC{$(R\lu\lw, \lw : \h \vA, \lu : \vA \sar ) \seqcomp \vL$}
\RightLabel{$\hl$}
\UnaryInfC{$(R\lu\lw, \lw : \h \vA \sar ) \seqcomp \vL$}
\DisplayProof
\end{tabular}
\end{center}

\begin{center}
\begin{tabular}{c c}
\AxiomC{$(R\lw\lu \sar \lu : \vA, \lw : \f \vA) \seqcomp \vL$}
\RightLabel{$\fr$}
\UnaryInfC{$(R\lw\lu \sar \lw : \f \vA) \seqcomp \vL$}
\DisplayProof

&

\AxiomC{$(R\lu\lw \sar \lw : \p \vA, \lu : \vA) \seqcomp \vL$}
\RightLabel{$\pr$}
\UnaryInfC{$(R\lu\lw \sar \lw : \p \vA) \seqcomp \vL$}
\DisplayProof
\end{tabular}
\end{center}

\hrule
\caption{Modal rules. The $\dag$ side condition stipulates that in a rule application under a substitution $\sub$, the label $\lu\sub$ must be fresh, i.e. it cannot occur in the conclusion of the rule application.}
\label{fig:labeled-modal-rules}
\end{figure}


 A rule is (\emph{quantity-preserving}) \emph{admissible} in a calculus \ifandonlyif for any instance of the rule, if each premise $\lseq_{i}$ with $1 \leq i \leq n$ is derivable with a proof $\lproof_{i}$, then the conclusion is derivable with a proof $\lproof$ (such that $\qty{\lproof} \leq \max\{\qty{\lproof_{i}} \ | \ 1 \leq i \leq n \}$). To simplify terminology, we write \emph{qp-admissible} rather than \emph{quantity-preserving admissible}. The (qp-)admissible structural rules for $\gtkt$ are shown in \fig~\ref{fig:labeled-structural-rules}. More specifically, the rules $\lsub$, $\ctrl$, $\ctrr$, and $\wk$ rules are qp-admissible in $\gtkt$ and $\cut$ is admissible~\cite{Bor08}. The $\lsub$ rule is called the \emph{labeled substitution rule} because it applies a \emph{label substitution} $(w/u)$ to any instance of the premise to obtain the conclusion. As usual, a \emph{label substitution} $(w/u)$ replaces each label $u$ occurring in a labeled sequent by a label $w$ (see~\cite{Vig00}). For example, $(Ruv, w : A \sar u : B)(w/u) = (Rwv, w : A \sar w : B)$.

\begin{example}\label{eg:struc-rules-example} To provide the reader with intuition, we give two examples of applications of structural rules. First, let $\lseq = (Ruv, w : A \sar u : B)$ and $\sub = [\lseq / \vL, w / \lw, u / \lv]$. Applying $\sub$ to the rule $\lsub$ (shown below left) yields the rule instance below right.
\begin{center}
\begin{tabular}{c c c}
\AxiomC{$\vL\sub$}
\RightLabel{$\lsub$}
\UnaryInfC{$\vL(\lw/\lv)\sub$}
\DisplayProof

&

$=$

&

\AxiomC{$Ruv, w : A \sar u : B$}
\RightLabel{$\lsub$}
\UnaryInfC{$(Ruv, w : A \sar u : B)(w/u)$}
\RightLabel{$=$}
\dottedLine
\UnaryInfC{$Rwv, w : A \sar w : B$}
\DisplayProof
\end{tabular}
\end{center}
Second, let $\lseq = (\rel, \Gamma \sar \Delta)$, $\lseq' = (Rwu \sar w : A)$, and $\sub = [\lseq / \vL, \lseq' / \vL']$. Applying $\sub'$ to the rule $\wk$ (shown below left) yields the rule instance below right.
\begin{center}
\begin{tabular}{c c c}
\AxiomC{$\vL\sub'$}
\RightLabel{$\wk$}
\UnaryInfC{$(\vL \seqcomp \vL')\sub'$}
\DisplayProof

&

$=$

&

\AxiomC{$\rel, \Gamma \sar \Delta$}
\RightLabel{$\wk$}
\UnaryInfC{$\rel, Rwu, \Gamma \sar w : A, \Delta$}
\RightLabel{$=$}
\dottedLine
\UnaryInfC{$(\rel, \Gamma \sar \Delta) \seqcomp (Rwu \sar w : A)$}
\DisplayProof
\end{tabular}
\end{center}
\end{example}

\begin{figure}[t]
\noindent\hrule
\begin{center}
\begin{tabular}{c c c}
\AxiomC{$\vL$}
\RightLabel{$\lsub$}
\UnaryInfC{$\vL(\lw/\lv)$}
\DisplayProof

&

\AxiomC{$\vL$}
\RightLabel{$\wk$}
\UnaryInfC{$\vL \seqcomp \vL'$}
\DisplayProof

&

\AxiomC{$(\lw : \vA, \lw : \vA \sar ) \seqcomp \vL$}
\RightLabel{$\ctrl$}
\UnaryInfC{$(\lw : \vA \sar ) \seqcomp \vL$}
\DisplayProof
\end{tabular}
\end{center}

\begin{center}
\begin{tabular}{c c c}
\AxiomC{$( \sar \lw : \vA, \lw : \vA) \seqcomp \vL$}
\RightLabel{$\ctrr$}
\UnaryInfC{$( \sar \lw : \vA) \seqcomp \vL$}
\DisplayProof

&

\AxiomC{$(\sar \lw : \vA) \seqcomp \vL$}
\AxiomC{$(\lw : \vA \sar ) \seqcomp \vL$}
\RightLabel{$\cut$}
\BinaryInfC{$\vL$}
\DisplayProof
\end{tabular}
\end{center}

\hrule
\caption{Admissible structural rules for $\gtkt$.}
\label{fig:labeled-structural-rules}
\end{figure}

As in the display setting, we not only want to provide a calculus for the minimal tense logic $\kt$, but also for extensions of $\kt$ with (simplified) primitive tense axioms, i.e. for any tense logic $\ktp$. To accomplish this goal, we define labeled versions of primitive tense structural rules, which have hitherto been undefined for labeled sequent calculi. Such rules are obtained in a similar fashion as in the display setting by means of transforming a simplified primitive tense axiom into a rule. This transformation is carried out via a translation function $\phi_{\lw}$.

\begin{definition}[Translation $\phi_{\lw}$]\label{def:phi-translation} We recursively define the translation function $\phi_{\lw}$ that maps formulae to schematic labeled sequents.
\begin{itemize}

\item $\phi_{\lw}(\top) := (\empdata \sar \empdata)$

\item $\phi_{\lw}(p) := \vL_{p}^{\lw}$

\item $\phi_{\lw}(A \land B) := \phi_{\lw}(A) \seqcomp \phi_{\lw}(B)$

\item $\phi_{\lw}(\f A) := (R\lw\lu \sar \empdata) \seqcomp \phi_{\lu}(A)$ with $\lu$ a new label variable

\item $\phi_{w}(\p A) := (R\lu\lw \sar \empdata) \seqcomp \phi_{\lu}(A)$ with $\lu$ a new label variable

\end{itemize}
 We note that each occurrence of $\vL_{p}^{\lw}$ is taken to be a unique labeled sequent variable; e.g. both occurrences of $\vL_{p}^{\lw}$ in $\phi_{\lw}(p \land p) = \vL_{p}^{\lw} \seqcomp \vL_{p}^{\lw}$ are taken to be unique labeled sequent variables.
\end{definition}

Labeled sequent variables obtained from the translation $\phi_{\lw}$ are annotated with label variables and atoms. This information is crucial for defining primitive tense structural rules in the labeled setting. First, the label variables tell us which labels must be fresh in a rule instance. Second, the atoms annotating labeled sequent variables tell us what substitutions are applicable, viz. labeled sequent variables annotated with the same atom must be instantiated with \emph{isomorphic} labeled sequents (see \dfn~\ref{def:iso} below). These are restrictions on what substitutions may be applied to a primitive tense structural rule and they ensure the soundness of rule applications.
 
\begin{definition}[Isomorphic]\label{def:iso} Let $\lseq_{1} = \rel_{1}, \Gamma_{1} \sar \Delta_{1}$ and $\lseq_{2} = \rel_{2}, \Gamma_{2} \sar \Delta_{2}$ be two labeled sequent. They are \emph{isomorphic}, written $\lseq_{1} \iso \lseq_{2}$, \iffi there exists a function $f : \lab(\lseq_{1}) \to \lab(\lseq_{2})$ such that 
\begin{multicols}{2}
\begin{enumerate}
\item[(1)] $f$ is bijective, 

\item[(2)] $Rwu \in \rel_{1}$ \iffi $Rf(w)f(u) \in \rel_{2}$,

\item[(3)] $w : A \in \Gamma_{1}$ \iffi $f(w) : A \in \Gamma_{2}$,

\item[(4)] $w : A \in \Delta_{1}$ \iffi $f(w) : A \in \Delta_{2}$.
\end{enumerate}
\end{multicols}
\end{definition}


 Let us now define primitive tense structural rules in the labeled setting. Every simplified primitive tense axiom can be transformed into a primitive tense structural rule, dubbed $\ptsrl$. By extending $\gtkt$ with the set of $\ptsrl$ rules corresponding to a set $\mathsf{P}$ of (simplified) primitive tense axioms, we obtain a labeled calculus for a primitive tense logic $\ktp$.


\begin{definition}[Primitive Tense Structural Rule]\label{def:primitive-tense-rule-labeled} Each simplified primitive tense axiom of the form shown below left corresponds to a primitive tense structural rule of the following form shown below right:
\begin{center}
\begin{tabular}{c c}
$\displaystyle{A \rightarrow \bigvee_{1 \leq j \leq m} B_{j}}$

&

\AxiomC{$ \phi_{\lw}(A) \seqcomp \phi_{\lw}(B_{1}) \seqcomp \vL$}
\AxiomC{$\ldots$}
\AxiomC{$ \phi_{\lw}(A) \seqcomp \phi_{\lw}(B_{m}) \seqcomp \vL$}
\RightLabel{$\ptsrl^{\dag}$}
\TrinaryInfC{$ \phi_{\lw}(A) \seqcomp \vL$}
\DisplayProof
\end{tabular}
\end{center}
 The side condition $\dag$ stipulates that for any instance of the rule under a substitution $\sub$, the following conditions must be satisfied:
 \begin{description}

\item[$\coni$] for each $1 \leq i \leq m$, if the label variable $\lu$ occurs in $\phi_{\lw}(B_{i})$, but does \emph{not} occur in some labeled sequent variable, then $\lu\sub$ must be fresh.

\item[$\conii$] for any two labeled sequent variables of the form $\vL_{p}^{\lu}$ and $\vL_{p}^{\lv}$ (i.e. for any two labeled sequent variables annotated with the same atom) occurring in the rule, $\vL_{p}^{\lu}\sub \iso \vL_{p}^{\lv}\sub$.

\item[$\coniii$] for any labeled sequent variable $\vL_{p}^{\lv}$, if $\vL_{p}^{\lv}\sub \neq (\emptyset \sar \emptyset)$, then $\lv\sub \in \lab(\vL_{p}^{\lv}\sub)$.

\end{description}
\end{definition}

 Contrary to the display setting, each $\ptsrl$ rule retains a copy of $\phi_{\lw}(A)$ in its premises, that is, the principal formulae are bottom-up preserved in rule applications. Formulating rules in this manner permits the qp-admissibility of contractions, so long as our calculus abides by the so-called `closure condition' (cf.~\cite[p.~29]{Bor08})\label{closure-cond}. To define the closure condition, let us take a $\ptsrl$ rule as in \dfn~\ref{def:primitive-tense-rule-labeled} above and suppose that we substitute certain label variables occurring in $\phi_{\lw}(A)$ by other label variables occurring in $\phi_{\lw}(A)$. If, after this substitution, all premises and the conclusion \emph{mutually contain} (1) at least two copies of a schematic relational atom of the form $(R\lw\lu \sar \empdata)$ or (2) at least two copies of a labeled sequent variable of the form $\lseq_{p}^{\lw}$, then we define a \emph{contraction} of $\ptsrl$ to be the rule obtained by deleting one of the duplicate copies, respectively. We then say that an extension of $\gtkt$ with a set of primitive tense structural rules and contractions thereof satisfies the \emph{closure condition} \iffi the calculus is closed under the contraction of every primitive tense structural rule and every contraction thereof. 
 
 We note that if we close the primitive tense structural rules of an extension of $\gtkt$ under contractions, then the resulting calculus is still finite. This follows from the fact that (1) we only consider extensions of $\gtkt$ with \emph{finitely} many primitive tense structural rules, and (2) because each primitive tense structural rule has only finitely many contractions as only a finite number of label variable substitutions (as described above) are possible. 
 
 To make the closure condition clearer, we provide an example of a primitive tense structural rule and a contraction of the rule. 
 
\begin{example} Let us consider the Euclidean axiom $\p \f p \rightarrow \f p$ (which is a simplified primitive tense axiom). The axiom's corresponding primitive tense structural rule is as follows:
\begin{center}
\AxiomC{$(R\lu\lw \sar \empdata) \seqcomp (R\lu\lv \sar \empdata) \seqcomp \vL_{p}^{\lv} \seqcomp (R\lw\lz \sar ) \seqcomp \vL_{p}^{\lz} \seqcomp \vL$}
\RightLabel{$(euc)$}
\UnaryInfC{$(R\lu\lw \sar \empdata) \seqcomp (R\lu\lv \sar \empdata) \seqcomp \vL_{p}^{\lv} \seqcomp \vL$}
\DisplayProof
\end{center}
 If we substitute the label $\lw$ for $\lu$ and $\lv$, then we obtain the following:
 \begin{center}
\AxiomC{$(R\lw\lw \sar \empdata) \seqcomp (R\lw\lw \sar \empdata) \seqcomp \vL_{p}^{\lw} \seqcomp (R\lw\lz \sar ) \seqcomp \vL_{p}^{\lz} \seqcomp \vL$}
\UnaryInfC{$(R\lw\lw \sar \empdata) \seqcomp (R\lw\lw \sar \empdata) \seqcomp \vL_{p}^{\lw} \seqcomp \vL$}
\DisplayProof
\end{center}
 Since the schematic labeled sequent $(R\lw\lw \sar \empdata)$ occurs twice, a contraction of $(euc)$ can be obtained by deleting the additional copy, yielding the following:
\begin{center} 
\AxiomC{$(R\lw\lw \sar \empdata) \seqcomp \vL_{p}^{\lw} \seqcomp (R\lw\lz \sar ) \seqcomp \vL_{p}^{\lz} \seqcomp \vL$}
\RightLabel{$(euc)'$}
\UnaryInfC{$(R\lw\lw \sar \empdata) \seqcomp \vL_{p}^{\lw} \seqcomp \vL$}
\DisplayProof
\end{center}
 We note that no contractions exist of the $(euc)'$ rule as no substitution of label variables can produce duplications. As specified above, contractions are only possible when all premises and the conclusion share duplicate formulae of the form $(R\lw\lu \sar \empdata)$ or $\vL_{p}^{\lw}$, yet, the conclusion of $(euc)'$ is of a form that precludes this possibility.
\end{example}

\begin{definition}[$\gtktp$]\label{def:gtktp} Let $\ktp$ be a primitive tense logic. We define $\gtktp$ to be $\gtkt$ extended with a $\ptsrl$ rule and all contractions thereof, for each simplified primitive tense axiom in $\paxs$. We indicate that a labeled sequent $\lseq$ is provable in $\gtktp$ with a labeled proof $\lproof$ by writing $\proves{\gtktp}{\lproof}{\lseq}$.
\end{definition}

\begin{theorem}\label{lem:admiss-rules-labeled} The $\lsub$, $\wk$, $\ctrl$, and $\ctrr$ rules are qp-admissible in $\gtktp$ and $\cut$ is admissible.
\end{theorem}

\begin{proof} Follows from \lem~2.3.4, \thm~2.3.6, \thm~2.3.8, and \thm~2.3.10 of \cite{Bor08}.
\end{proof}

\begin{remark}\label{rmk:qp-admiss-PTIME}
 The qp-admissibility of $\lsub$, $\wk$, $\ctrl$, and $\ctrr$ is proven in $\gtktp$ by showing that each rule can be permuted upward in any proof and deleted at initial rules (while preserving the quantity of the proof). Therefore, this method of proof, which relies on rule permutations, shows the existence of an algorithm that takes a proof with structural rules as input and returns a proof without such rules as output. Such proof transformations are computable in $\ptime$ in the size of the input proof. This is because if our input proof is $\lproof$, then there are most $\qty{\lproof}$ occurrences of structural rules in $\lproof$, and for each such occurrence at most $\qty{\lproof}$ many permutations need to be made, showing that the algorithm performs at most $\qty{\lproof} \times \qty{\lproof} \leq \size{\lproof}^2$ many permutations. Moreover, the width of the output proof will always be bounded by $\qty{\lproof} \times \wdth{\lproof} = \size{\lproof}$. This is straightforward to verify in the $\lsub$, $\ctrl$, and $\ctrr$ cases, which may reduce the length of labeled sequents when permuted upward; in the case of $\wk$, which weakens in a labeled sequent $\lseq$, $\lgth{\lseq} \leq \wdth{\lproof}$, so permuting numerous instances of $\wk$ upward may increase the length of the `longest' labeled sequent occurring in $\lproof$, and thus, the width of the output proof will be at most $\qty{\lproof} \times \wdth{\lproof}$. Since $\lsub$, $\wk$, $\ctrl$, and $\ctrr$ are qp-admissible, this shows that the size of the output proof is polynomial in the size of the input proof and that the proof transformations eliminating $\lsub$, $\wk$, $\ctrl$, and $\ctrr$ instances are computable in $\ptime$.
\end{remark}

\section{Labeled Polytrees}\label{section-3-2}

 In this section, we discuss special graphs, referred to as \emph{polytrees}.  Intuitively, a polytree is a directed graph such that the underlying graph (i.e. the graph obtained by replacing each directed edge by an undirected edge) is a tree (i.e. is a connected, cycle-free graph). This notion is important because labeled sequents with a polytree structure are notational variants of display sequents. The importance of polytrees was first observed in works on translations between labeled and shallow-nested calculi for tense logics with general path axioms~\cite{CiaLyoRamTiu21}, which form a small sub-class of the primitive tense logics and calculi we consider. 
 In contrast, this paper considers translations between more expressive labeled and display calculi, which requires novel methods of translation, greatly generalizing the work in~\cite{CiaLyoRamTiu21}, and leading to new insights (discussed throughout \sects~\ref{section-4}--\ref{conclusion}).

 Let us now define important terminology that will be of use throughout the remainder of the paper. We define a \emph{$w$-flat sequent} to be a labeled sequent of the form $w : A_{1}, \ldots, w : A_{n} \sar w : B_{1}, \ldots, w : B_{k},$ i.e. a labeled sequent without relational atoms and where all labeled formulae share the same label. When the label is not important, we sometimes refer to a $w$-flat sequent as a \emph{flat sequent}. We define the \emph{graph} of a labeled sequent $\lseq = \rel, \Gamma \sar \Delta$ to be $\grph{\lseq} = (V,E)$ such that
$$
V = \{(w,\lseq') \ | \ w \in \lab(\lseq) \text{ and } \lseq' = (\Gamma \restriction w \sar \Delta \restriction w)\}
\text{ and }
E = \{(w,u) \ | \ Rwu \in \rel\}.
$$
 In other words, the graph $\grph{\lseq}$ of a labeled sequent $\lseq$ is a graph where each node $w$ is a label decorated with $w$-flat sequent (obtained from the formulae labeled with $w$ in $\lseq$) and each edge is obtained from a relational atom.

\begin{definition}[Labeled Polytree] Let $\lseq = \rel, \Gamma \sar \Delta$ be a labeled sequent and $\grph{\lseq} = (V,E)$ be its graph. We define $\lseq$ to be a \emph{labeled polytree sequent} \iffi (1) if $\rel = \emptyset$, then $\lseq$ is a $w$-flat sequent for some $w \in \lab$, (2) if $\rel \neq \emptyset$, then $\lab(\Gamma, \Delta) \subseteq \lab(\rel)$, and (3) $\grph{\lseq}$ forms a \emph{polytree}, that is, $\grph{\lseq}$ is connected and free of (un)directed cycles. We define a \emph{labeled polytree proof} to be a proof containing only labeled polytree sequents.
\end{definition}

\begin{example} To provide intuition on labeled polytree sequents, we give an example. Suppose $\lseq_{i}$ is a $w_{i}$-flat sequent for $0 \leq i \leq 5$. The labeled polytree sequent
$$
\lseq = (Rw_{0}w_{1}, Rw_{4}w_{1}, Rw_{5}w_{1}, Rw_{1}w_{2}, Rw_{3}w_{2} \sar \emptyset) \seqcomp \lseq_{0} \seqcomp \lseq_{1} \seqcomp \lseq_{2} \seqcomp \lseq_{3} \seqcomp \lseq_{4} \seqcomp \lseq_{5}
$$
 can be pictured as shown on the left in \fig~\ref{fig:labeled-polytree-sequent-composition}. (NB. Please ignore the dashed box for the time being; this is discussed below.)
\end{example}

 Moreover, \fig~\ref{fig:labeled-polytree-sequent-composition} also demonstrates the concept of a \emph{$w$-partition}, whereby a labeled polytree sequent $\lseq$ may be split into two labeled polytree sequents $\lseq'$ and $\lseq''$ such that $\lseq = \lseq' \seqcomp \lseq''$. For instance, in \fig~\ref{fig:labeled-polytree-sequent-composition}, if $\lseq$ is the left-most labeled polytree sequent, then $\lseq$ may be partitioned into $\lseq'$ (shown in the middle) and $\lseq''$ (shown on the right) so long as $\lseq_{2} = \lseq_{2}' \seqcomp \lseq_{2}''$. 

\begin{definition}[$w$-partition]\label{def:w-partition} Let $\lseq_{1} = \rel_{1}, \Gamma_{1} \sar \Delta_{1}$ and $\lseq_{2} = \rel_{2}, \Gamma_{2} \sar \Delta_{2}$ be labeled sequents. We say that $\lseq_{1}$ and $\lseq_{2}$ are \emph{$w$-disjoint} \iffi $\lab(\lseq) \cap \lab(\lseq') = \{w\}$. We define $\lseq_{1} \ptcomp{w} \lseq_{2}$ to be a \emph{$w$-partition} of $\lseq$ \iffi (1) $\lseq = \lseq_{1} \seqcomp \lseq_{2}$, (2) $\lseq_{1}$ and $\lseq_{2}$ are labeled polytree sequents, and (3) $\lseq_{1}$ and $\lseq_{2}$ are $w$-disjoint.
\end{definition}

\begin{figure}[t]\label{fig:labeled-polytree-sequent-composition}

\begin{center}
\begin{tabular}{c @{\hskip 2em} c}
\begin{tikzpicture}
     \node[world] [] (w0) [label=above:$w_{0}$] {$\lseq_{0}$};
     \node[world] (w1) [below left=of w0,label=above:$w_{1}$] {$\lseq_{1}$};
     \node[world] (w2) [below right=of w0,label=above:$w_{2}$] {$\lseq_{2}$};
     \node[world] (w3) [below=of w2,label=below:$w_{3}$] {$\lseq_{3}$};
     \node[world] (w4) [below=of w1,label=below:$w_{4}$] {$\lseq_{4}$};
     \node[world] (w5) [right=of w4,label=below:$w_{5}$] {$\lseq_{5}$};

     \path[->,draw] (w0) -- (w1);
     \path[->,draw] (w0) -- (w2);
     \path[->,draw] (w3) -- (w2);
     \path[->,draw] (w4) -- (w1);
     \path[->,draw] (w5) -- (w1);  
     
    \node[draw, dashed, color=black, rounded corners, inner xsep=4pt, inner ysep=11pt, fit=(w1)(w4)(w5)] (c1) {}; 
\end{tikzpicture}

&

\begin{tikzpicture}
     \node[world] [] (w0) [label=above:$w_{0}$] {$\lseq_{0}$};
     \node[world] (w1) [below left=of w0,label=above:$w_{1}$] {$\lseq_{1}$};
     \node[world] (w2) [below right=of w0,label=above:$w_{2}$] {$\lseq_{2}'$};
          \node[world] (w22) [right=of w2,label=above:$w_{2}$] {$\lseq_{2}''$};
     \node[world] (w3) [below=of w22,label=below:$w_{3}$] {$\lseq_{3}$};
     \node[world] (w4) [below=of w1,label=below:$w_{4}$] {$\lseq_{4}$};
     \node[world] (w5) [right=of w4,label=below:$w_{5}$] {$\lseq_{5}$};

     \path[->,draw] (w0) -- (w1);
     \path[->,draw] (w0) -- (w2);
     \path[->,draw] (w3) -- (w22);
     \path[->,draw] (w4) -- (w1);
     \path[->,draw] (w5) -- (w1);  
     
     
     
\end{tikzpicture}
\end{tabular}
\end{center}

\caption{Examples of labeled polytree sequents $\lseq$, $\lseq'$, and $\lseq''$ (read from left to right). If $\lseq_{2} = \lseq_{2}' \seqcomp \lseq_{2}''$, then $\lseq'$ and $\lseq''$ serve as a $w_{2}$-partition of $\lseq$, i.e. $\lseq = \lseq' \seqcomp_{w} \lseq''$ (see \dfn~\ref{def:w-partition}). The dashed region contains the labeled polytree sequent $\lseq |_{w_{1}}^{w_{0}}$ that serves as a subpolytree sequent of $\lseq$ (see \dfn~\ref{def:sub-sequent}).}
\end{figure}

 The concept of a \emph{subpolytree sequent} will also be of use to us; in particular, to define certain recursive operations on the structure of a given labeled polytree sequent. For example, let us consider the labeled sequent $\lseq$ shown in the left of \fig~\ref{fig:labeled-polytree-sequent-composition}. The subpolytree sequent rooted at $w_{1}$ relative to $w_{0}$, denoted $\lseq |_{w_{1}}^{w_{0}}$, is the labeled sequent $(Rw_{4}w_{1}, Rw_{5}w_{1} \sar \emptyset) \seqcomp \lseq_{1} \seqcomp \lseq_{4} \seqcomp \lseq_{5}$ within the dashed region. 

\begin{definition}[Subpolytree Sequent]\label{def:sub-sequent} If $\lseq = \rel, \Gamma \sar \Delta$ is a labeled polytree sequent with $Rwu \in \rel$ (or, $Ruw \in \rel$), then $\lseq |_{u}^{w}$ is defined to be the unique labeled polytree sequent, referred to as a \emph{subpolytree sequent}, such that there exists a labeled polytree sequent $\lseq'$ and $\lseq = (\lseq' \seqcomp_{w} (Rwu \sar \emptyset)) \seqcomp_{u} \lseq |_{u}^{w}$ (respectively, $\lseq = (\lseq |_{u}^{w} \seqcomp_{w} (Ruw \sar \emptyset)) \seqcomp_{u} \lseq'$).
\end{definition}

 Let us now prove a few useful facts about the relationship between labeled polytree sequents and provability in $\gtktp$. It will be useful to identify what proofs in $\gtktp$ are labeled polytree proofs since---as shown in \sects~\ref{section-5} and \ref{section-6}---such proofs translate to and from display proofs in $\dktp$. Toward this end, it will be helpful to consider rules in $\gtktp$ that preserve the polytree structure of labeled sequents when applied.
 
 First, one can verify that for any rule $\ru$ in $\gtkt$, the conclusion of $\ru$ is a labeled polytree sequent \iffi every premise is a labeled polytree sequent, i.e. the property of being a labeled polytree sequent is both top-down and bottom-up preserved in $\gtkt$. For instance, if we consider the $\fl$, $\pl$, $\gr$, and $\hr$ rules, then we notice that every rule bottom-up introduces a relational atom from a label already occurring in the labeled sequent to a fresh label, and thus, the conclusion is a labeled polytree sequent \iffi the premise is a labeled polytree sequent, albeit, the premise has an additional protruding edge/relational atom. Regarding the remaining rules of $\gtkt$, each rule both top-down and bottom-up preserves the relational atoms of the premises and conclusion, respectively, meaning each premise will be a labeled polytree sequent \iffi the conclusion is.  Therefore, we have the following: 

\begin{proposition}\label{prop:gtkt-upward-preserves-polytree}
Let $\ru$ be a non-initial rule in $\gtkt$. The premises of $\ru$ are labeled polytree sequents \iffi the conclusion is.
\end{proposition}

 By the above proposition, we know that every rule $\ru$ in $\gtkt$ preserves the property of being a labeled polytree sequent. Yet, the question remains: is this also true for primitive tense structural rules? In general, the answer is `no' as primitive tense structural rules can introduce disconnectedness and (un)directed cycles when applied. For instance, the primitive tense structural rule $\refl$ corresponding to the reflexivity axiom $p \imp \f p$ is shown below left. As shown below right, by applying the substitution $\sub$ such that $\sub(\vL) = (\emptyset \sar w : \f p)$, $\sub(\vL_{p}^{\lw}) = \sub(\vL_{p}^{\lu}) = (w : p \sar \emptyset)$, $\sub(\lw) = w$, and $\sub(\lu) = w$, we obtain a valid instance of the rule which contains the loop $Rww$.
\begin{center}
\begin{tabular}{c @{\hskip 1em} c}
\AxiomC{$\vL_{p}^{\lw} \otimes (R\lw\lu \sar \emptyset) \otimes \vL_{p}^{\lu} \otimes \vL$}
\RightLabel{$\refl$}
\UnaryInfC{$\vL_{p}^{\lw} \otimes \vL$}
\DisplayProof

&

\AxiomC{$Rww, w : p, w : p \sar w :\f p$}
\RightLabel{$\refl$}
\UnaryInfC{$w : p \sar w :\f p$}
\DisplayProof
\end{tabular}
\end{center}
 Nevertheless, observe that if we apply a different substitution $\sub'$ to $\refl$ such that $\sub'(\vL) = (\emptyset \sar w : \f p)$, $\sub(\vL_{p}^{\lw}) = (w : p \sar \emptyset)$, $\sub(\vL_{p}^{\lu}) = (u : p \sar \emptyset)$, $\sub(\lw) = w$, and $\sub(\lu) = u$, then we obtain the following instance of the rule containing only labeled polytree sequents.
\begin{center}
\AxiomC{$Rwu, w : p, u : p \sar w :\f p$}
\RightLabel{$\refl$}
\UnaryInfC{$w : p \sar w :\f p$}
\DisplayProof
\end{center}
 This observation suggests that certain instances of primitive tense structural rules bottom-up preserve the polytree structure of the conclusion while others do not. In general, we find that if a certain set of conditions are satisfied, then each premise of a primitive tense structural rule will be labeled polytree sequent so long as the conclusion is. We list these conditions below and define the corresponding notion of a \emph{strict} $\ptsrl$ application.

\begin{definition}[Strict]\label{def:strictness} Let $\ptsrl$ be of the form shown in \dfn~\ref{def:primitive-tense-rule-labeled} with $C \in \{A, B_{1}, \ldots, B_{m}\}$. We say that an application of $\ptsrl$ under a substitution $\sub$ is \emph{strict} \iffi it satisfies conditions $\coni$--$\coniii$ as well as the following conditions:
\begin{description}

\item[$\coniv$] for any two label variables $\lu$ and $\lv$, if $\lu \neq \lv$, then $\lu\sub \neq \lv\sub$.

\item[$\conv$] for any two occurrences of labeled sequent variables $\vL_{p}^{\lv}$ and $\vL_{q}^{\lu}$, if $\lv\sub \neq \lu\sub$, then $\lab(\vL_{p}^{\lv}\sub) \cap \lab(\vL_{q}^{\lu}\sub) = \emptyset$; otherwise, $\lab(\vL_{p}^{\lv}\sub) \cap \lab(\vL_{q}^{\lu}\sub) = \{\lv\sub\}$.


\item[$\convi$] if $\phi_{\lw}(C)\sub \neq (\empdata \sar \empdata) \neq \vL\sub$, then $\lab(\phi_{\lw}(C)\sub) \cap \lab(\vL\sub) = \{\lw\sub\}$. 

\item[$\convii$] for any occurrence of a labeled sequent variable $\vL$, $\vL\sub$ is a labeled polytree sequent.

\end{description}
 A \emph{strict} proof in $\gtktp$ is a proof such that every instance of $\ptsrl$ is strict.
\end{definition}

 We aim to show that if a proof in $\gtktp$ is strict, then it is a labeled polytree proof. Let us therefore analyze the structures present in a primitive tense structural rule. Observe that the $\phi_{\lw}$ function takes a formula $A$ from the language $\langatoms$ and returns a schematic labeled sequent that forms a polytree. This follows from the fact that every time a $\f$ or $\p$ modality is encountered, the $\phi_{\lw}$ function adds a schematic relational atom to or from a fresh label variable (see \dfn~\ref{def:phi-translation} above). For instance, if we consider the formula $p \land \f q \land \f(r\land \p s \land \p t) \in \langatoms$, then  $\phi_{\lw}(p \land \f q \land \f(r\land \p s \land \p t)) =$
 $$
(R\lw\lu \sar \emptyset) \seqcomp (R\lw\lv \sar \emptyset) \seqcomp (R\dot{z}\lv \sar \emptyset) \seqcomp (R\dot{y}\lv \sar \emptyset) \seqcomp \vL_{\lw}^{p} \seqcomp \vL_{\lu}^{q} \seqcomp \vL_{\lv}^{r} \seqcomp \vL_{\dot{z}}^{s} \seqcomp \vL_{\dot{y}}^{t},
 $$
which can be pictured as the polytree shown top left in \fig~\ref{fig:phi-gives-polytree}. 

We define the \emph{graph} of a schematic labeled sequent $\sls$ to be $\grph{\sls} = (V,E)$ such that (1) $(\lw,X) \in V$ \iffi $\lw$ occurs in $\sls$ and $X = \{\lseq_{p}^{\lw} \ | \ \lseq_{p}^{\lw} \text{ occurs in } \sls \}$ and (2) $(\lw,\lu) \in E$ \iffi $R\lw\lu$ occurs in $\sls$. A schematic labeled sequent \emph{forms a polytree} \iffi its graph is a polytree, i.e. is connected and free of (un)directed cycles. The following lemma is straightforward to prove:

\begin{figure}[t]
\begin{center}
\begin{tabular}{c @{\hskip 2em} c}
\begin{tikzpicture}
     \node[world] [] (w0) [] {$\vL_{\lw}^{p}$};
     \node[world] (w1) [below left=of w0,xshift=.25cm,yshift=.25cm,] {$\vL_{\lu}^{q}$};
     \node[world] (w2) [below right=of w0,xshift=-.25cm,yshift=.25cm,] {$\vL_{\lv}^{r}$};
     \node[world] (w3) [below left=of w2,xshift=.75cm,yshift=.25cm] {$\vL_{\dot{z}}^{s}$};
     \node[world] (w4) [below right=of w2,xshift=-.75cm,yshift=.25cm] {$\vL_{\dot{y}}^{t}$};

     \path[->,draw] (w0) -- (w1);
     \path[->,draw] (w0) -- (w2);
     \path[->,draw] (w3) -- (w2);
     \path[->,draw] (w4) -- (w2);
     
\end{tikzpicture}

&

\begin{tikzpicture}
     \node[world] [] (w0) [label=above:$w_{0}$] {$\lseq_{0}$};
     \node[world] (w1) [below left=of w0,xshift=.25cm,yshift=.25cm,label=above:$w_{1}$] {$\lseq_{1}$};
     \node[world] (w2) [below right=of w0,xshift=-.25cm,yshift=.25cm,label=above:$w_{2}$] {$\lseq_{1}'$};
     \node[world] (w3) [below left=of w2,xshift=1cm,yshift=.25cm,label=below:$w_{5}$] {$\lseq_{2}'$};
     \node[world] (w4) [below left=of w1,xshift=1cm,yshift=.25cm,label=below:$w_{3}$] {$\lseq_{2}$};
     \node[world] (w5) [below right=of w1,xshift=-1cm,yshift=.25cm,label=below:$w_{4}$] {$\lseq_{3}$};
     \node[world] (w6) [below right=of w2,xshift=-1cm,yshift=.25cm,label=below:$w_{6}$] {$\lseq_{3}$};

     \path[->,draw] (w0) -- (w1);
     \path[->,draw] (w0) -- (w2);
     \path[->,draw] (w3) -- (w2);
     \path[->,draw] (w4) -- (w1);
     \path[->,draw] (w5) -- (w1);  
     \path[->,draw] (w6) -- (w2);  
     
    \node[draw, dashed, color=black, rounded corners, inner xsep=4pt, inner ysep=11pt, fit=(w1)(w4)(w5)] (c1) {}; 
    \node[draw, dashed, color=black, rounded corners, inner xsep=4pt, inner ysep=11pt, fit=(w2)(w3)(w6)] (c2) {}; 
\end{tikzpicture}
\end{tabular}
\end{center}
\begin{center}
\begin{tabular}{c c c}
\begin{tikzpicture}
     \node[world] [] (w0) [label=above:$w_{0}$] {$\lseq_{0}$};
     \node[world] (w1) [below=of w0,yshift=.5cm,label=left:$w_{1}$] {$\lseq_{1}''$};
     \node[world] (w4) [below left=of w1,xshift=-.25cm,yshift=.25cm,label=below:$w_{3}$] {$\lseq_{2}$};
     \node[world] (w5) [right=of w4,xshift=-.75cm,label=below:$w_{4}$] {$\lseq_{3}$};
     \node[world] (w6) [below right=of w1,xshift=.25cm,yshift=.25cm,label=below:$w_{6}$] {$\lseq_{3}'$};
     \node[world] (w3) [left=of w6,xshift=.75cm,label=below:$w_{5}$] {$\lseq_{2}'$};

     \path[->,draw] (w0) -- (w1);
     \path[->,draw] (w3) -- (w1);
     \path[->,draw] (w4) -- (w1);
     \path[->,draw] (w5) -- (w1);  
     \path[->,draw] (w6) -- (w1);  
     
    \node[draw, dashed, color=black, rounded corners, inner xsep=4pt, inner ysep=11pt, fit=(w4)(w5)] (c1) {}; 
    \node[draw, dashed, color=black, rounded corners, inner xsep=4pt, inner ysep=11pt, fit=(w3)(w6)] (c2) {}; 
\end{tikzpicture}

&

\begin{tikzpicture}
     \node[world] [] (w0) [label=above:$w_{0}$] {$\lseq_{0}$};
     \node[world] (w1) [below=of w0,yshift=.5cm,label=left:$w_{1}$] {$\lseq_{1} ''$};
     \node[world] (w5) [below=of w1,yshift=.25cm,yshift=.25cm,label=below:$w_{4}$] {$\lseq_{3}$};
     \node[world] (w4) [left=of w5,xshift=.45cm,label=below:$w_{3}$] {$\lseq_{2}''$};
     \node[world] (w6) [right=of w5,xshift=-.45cm,label=below:$w_{6}$] {$\lseq_{3}'$};

     \path[->,draw] (w0) -- (w1);
     \path[->,draw] (w4) -- (w1);
     \path[->,draw] (w5) -- (w1);  
     \path[->,draw] (w6) -- (w1);  
     
    \node[draw, dashed, color=black, rounded corners, inner xsep=4pt, inner ysep=11pt, fit=(w5)] (c1) {}; 
    \node[draw, dashed, color=black, rounded corners, inner xsep=4pt, inner ysep=11pt, fit=(w6)] (c2) {}; 
\end{tikzpicture}

&

\begin{tikzpicture}
     \node[world] [] (w0) [label=above:$w_{0}$] {$\lseq_{0}$};
     \node[world] (w1) [below=of w0,yshift=.5cm,label=left:$w_{1}$] {$\lseq_{1} ''$};
     \node[world] (w4) [below left=of w1,xshift=.75cm,yshift=.25cm,label=below:$w_{3}$] {$\lseq_{2}''$};
     \node[world] (w5) [below right=of w1,xshift=-.75cm,yshift=.25cm,label=below:$w_{4}$] {$\lseq_{3}''$};

     \path[->,draw] (w0) -- (w1);
     \path[->,draw] (w4) -- (w1);
     \path[->,draw] (w5) -- (w1);  
     
\end{tikzpicture}
\end{tabular}
\end{center}

\caption{A graphical depiction of the schematic labeled sequent $\phi_{\lw}(p \land \f q \land \f(r\land \p s \land \p t))$, which forms a polytree, is shown top left. The top right graph and the three bottom graphs demonstrate how strict labeled substitutions can be applied to merge two isomorphic subpolytree sequents $\lseq \subseq{w_{0}}{w_{1}}$ and $\lseq \subseq{w_{0}}{w_{2}}$ so that every labeled sequent generated is a labeled polytree sequent.\label{fig:zip-polytree}\label{fig:phi-gives-polytree}}
\end{figure}

\begin{lemma}\label{lem:phi-gives-polytree}
If $A \in \langatoms$, then the schematic labeled sequent $\phi_{\lw}(A)$ forms a polytree.
\end{lemma}

\begin{lemma}\label{lem:strict-ptsrl-is-polytree}
The premises and conclusion of a strict $\ptsrl$ instance are labeled polytree sequents.
\end{lemma}

\begin{proof} Let $\ptsrl$ be a primitive tense structural rule of the form shown below. We consider a strict instance of $\ptsrl$ under the substitution $\sub$ are argue that $\sls_{i}\sub$ (for $1 \leq i \leq m$) and $\sls\sub$ are labeled polytree sequents.
\begin{center}
\AxiomC{$\sls_{1}$}
\AxiomC{$\ldots$}
\AxiomC{$\sls_{m}$}
\RightLabel{$\ptsrl$}
\TrinaryInfC{$\sls$}
\DisplayProof
\end{center}
 First, by \lem~\ref{lem:phi-gives-polytree}, we know that the schematic relational atoms of all schematic labeled sequents occurring in $\ptsrl$ form a polytree. Let $L$ be the set of all labels in the range of $\sub$. Also, let $\grph{\sls_{i}\sub} = (V_{i},E_{i})$ and $\grph{\sls\sub} = (V,E)$. By condition $\coniv$, we know that $\sub$ is injective over the set of label variables occurring in its domain, meaning that $E_{i} \restriction L$ and $E \restriction L$ form polytrees. By condition $\convii$, all labeled sequent variables will be instantiated with labeled polytree sequents, and by conditions $\conv$ and $\convi$, we know that all such label polytree sequents (if not the empty sequent) will only intersect at a single label occurring in either $E_{i} \restriction L$ or $E \restriction L$. As `fusing' a single node of a polytree to a single node of another polytree forms a polytree, we are ensured that $\sls_{i}\sub$ (for $1 \leq i \leq m$) and $\sls\sub$ will be labeled polytree sequents.
\end{proof}

\begin{theorem}\label{thm:strict-is-polytree}
If $\lproof$ is a strict proof in $\gtktp$ of a labeled polytree sequent $\lseq$, then $\lproof$ is a labeled polytree proof.
\end{theorem}

\begin{proof} Follows from \lem~\ref{prop:gtkt-upward-preserves-polytree} and \lem~\ref{lem:strict-ptsrl-is-polytree}.
\end{proof}

We end this section by discussing qp-admissible structural rules that preserve the polytree structure of labeled sequents when applied. One can see that the conclusion of a $\ctrl$ and $\ctrr$ application is a labeled polytree sequent if the premise is; however, this is not generally true for $\lsub$ and $\wk$ (shown in \fig~\ref{fig:labeled-structural-rules}). We therefore define \emph{strict} versions of these rules that preserve polytree structure.

\begin{definition} We define $\slsub$ to be an instance of $\lsub$ under $\sub$ that satisfies the following condition: there exists a $u \in \lab(\vL\sub)$ such that either $Ru(\lw\sub),Ru(\lv\sub) \in \rel$ or $R(\lw\sub)u, R(\lv\sub)u \in \rel$ with $\vL\sub = \rel, \Gamma, \sar \Delta$. Effectively, $\slsub$ either identifies two children nodes of a single parent node or two parent nodes of a single child node in $\grph{\vL\sub}$, that is, it takes one of two forms shown below.
\begin{center}
\begin{tabular}{c c}
\AxiomC{$(R\lz\lu, R\lz\lv \sar ) \seqcomp \vL$}
\RightLabel{$\slsub$}
\UnaryInfC{$((R\lz\lu, R\lz\lv \sar ) \seqcomp \vL)(\lu/\lv)$}
\DisplayProof

&

\AxiomC{$(R\lu\lz, R\lv\lz \sar ) \seqcomp \vL$}
\RightLabel{$\slsub$}
\UnaryInfC{$((R\lu\lz, R\lz\lv \sar ) \seqcomp \vL)(\lu/\lv)$}
\DisplayProof
\end{tabular}
\end{center}

We define $\swk$ to be an instance of $\wk$ under $\sub$ that satisfies the following conditions: (1) $\vL'\sub$ is a label polytree sequent and (2) if $\vL\sub \neq (\empdata \sar \empdata) \neq \vL'\sub$, then $|\lab(\vL\sub) \cap \lab(\vL'\sub)| = 1$. 
 Effectively, $\swk$ introduces a labeled polytree sequent that intersects a single node (if one occurs) in $\grph{\vL\sub}$.
\end{definition}

\begin{example}\label{eg:zip-polytree} It is simple to verify that if the premise of a $\swk$ application is a labeled polytree sequent, then its conclusion is as well. We therefore show that $\slsub$ always yields a labeled polytree sequent when applied to one. Moreover, as it will be important for our work in \sect~\ref{section-5}, we also show how $\slsub$ can be repeatedly applied to `merge' two isomorphic subpolytree sequents. Suppose we are given a labeled polytree sequent $\lseq$ of the form:
$$
\lseq = (\rel \sar \emptyset) \seqcomp \lseq_{1} \seqcomp \lseq_{1}' \seqcomp \lseq_{2} \seqcomp \lseq_{2}' \seqcomp \lseq_{3} \seqcomp \lseq_{3}'
$$
where $\rel = Rw_{0}w_{1}, Rw_{0}w_{2}, Rw_{3}w_{1}, Rw_{4}w_{1}, Rw_{5}w_{2}, Rw_{6}w_{2}$ and for $i \in \{1,2,3\}$, $\lseq_{i} \iso \lseq_{i}'$. Also, suppose  $\lseq_{1}$, $\lseq_{1}'$, $\lseq_{2}$, $\lseq_{2}'$, $\lseq_{3}$, and $\lseq_{3}'$ is a $w_{0}$-flat, $w_{1}$-flat, $w_{2}$-flat, $w_{3}$-flat, $w_{4}$-flat, $w_{5}$-flat, $w_{6}$-flat sequent, respectively. Then, $\grph{\lseq}$ can be pictured as the top right polytree in \fig~\ref{fig:zip-polytree}. By applying the label substitutions $(w_{1}/w_{2})$, $(w_{3}/w_{5})$, and $(w_{4}/w_{6})$, we obtain the bottom left, bottom middle, and bottom right graphs, respectively, where $\lseq_{1}'' = \lseq_{1} \seqcomp \lseq_{1}'$, $\lseq_{2}'' = \lseq_{2} \seqcomp \lseq_{2}'$, and $\lseq_{3}'' = \lseq_{3} \seqcomp \lseq_{3}'$. Observe that each graph is a polytree, exemplifying that strict label substitutions can always be applied to preserve the labeled polytree property.
\end{example}

\begin{theorem}\label{thm:polytree-admiss} Let $\ru \in \{\slsub, \swk, \ctrl, \ctrr\}$. (1) If $\ru$ is applied to a labeled polytree sequent, then the conclusion is a labeled polytree sequent. (2) Eliminating $\ru$ from a strict proof via upward permutations yields a strict proof. 
\end{theorem}

\begin{proof} We prove claim (2) since claim (1) is straightforward. We consider a representative number of cases in our proof of claim (2); the remaining cases and sub-cases are tedious, but go through. We argue that any strict application of $\ptsrl$ 
followed by an application of $\ru$ can be replaced by applications of $\ru$ to the premises of $\ptsrl$ followed by a strict application of $\ptsrl$, yielding the same conclusion. We only consider cases when $\ru$ is either $\ctrl$ or $\slsub$ as the $\swk$ and $\ctrr$ cases are similar.

$\ctrl$. Suppose we have an instance of $\ptsrl$ followed by an instance of $\ctrl$, as shown below, such that $(\phi_{\lw}(A) \seqcomp \vL)\sub = ((\lw : \vB, \lw : \vB \sar ) \seqcomp \vL)\sub'$. 
\begin{center}
\AxiomC{$(\phi_{\lw}(A) \seqcomp \phi_{\lw}(B_{1}) \seqcomp \vL)\sub$}
\AxiomC{$\ldots$}
\AxiomC{$ (\phi_{\lw}(A) \seqcomp \phi_{\lw}(B_{m}) \seqcomp \vL)\sub$}
\RightLabel{$\ptsrl$}
\TrinaryInfC{$(\phi_{\lw}(A) \seqcomp \vL)\sub$}
\RightLabel{$=$}
\dottedLine
\UnaryInfC{$(\lu : \vB, \lu : \vB \sar ) \seqcomp \vL)\sub'$}
\RightLabel{$\ctrl$}
\UnaryInfC{$((\lu : \vB \sar ) \seqcomp \vL)\sub'$}
\DisplayProof
\end{center}
Either $(\lu : \vB, \lu : \vB)\sub'$ occurs in $\phi_{\lw}(A)\sub'$, $(\lu : \vB, \lu : \vB)\sub'$ occurs in $\vL\sub$, or one $(\lu : \vB)\sub'$ occurs in $\phi_{\lw}(A)\sub'$ and the other $(\lu : \vB)\sub'$ occurs in $\vL\sub$. We argue the first case as the last two cases are similar. Since $\phi_{\lw}(A)$ is a schematic labeled sequent that consists of schematic relational atoms and labeled sequent variables, it must be the case that either (1) there is some labeled sequent variable $\vL_{p}^{\lv}$ in $\phi_{\lw}(A)$ such that $(\lu : \vB, \lu : \vB)\sub'$ occurs within $\vL_{p}^{\lv}\sub$, or (2) there are two labeled sequent variables $\vL_{p}^{\lv}$ and $\vL_{q}^{\lz}$ in $\phi_{\lw}(A)$ such that one $(\lw : \vB)\sub'$ occurs in $\vL_{p}^{\lv}\sub$ and the other $(\lw : \vB)\sub'$ occurs in $\vL_{q}^{\lz}\sub$. In case (1), we know that $\vL_{p}^{\lv}\sub$ is of the form $\rel, \Gamma, u : B, u : B \sar \Delta$. We resolve the case by defining a new substitution $\sub''$ such that $\vL_{p}^{\lv}\sub'' = \rel, \Gamma, u : B \sar \Delta$ and for any other label, atomic, formula, or labeled sequent variable $x \neq \vL_{p}^{\lv}$ occurring in $\ptsrl$, $\sub''(x) = \sub(x)$. In the case (2), we know that $\vL_{p}^{\lv}\sub$ is of the form $\rel, \Gamma, u : B \sar \Delta$ and $\vL_{q}^{\lz}\sub$ is of the form $\rel', \Gamma', u : B \sar \Delta'$. We resolve the case by defining a new substitution $\sub''$ such that $\vL_{p}^{\lv}\sub'' = \rel, \Gamma \sar \Delta$, $\vL_{q}^{\lz}\sub'' = \rel, \Gamma, u : B \sar \Delta$, and for any other label, atomic, formula, or labeled sequent variable $x$ such that $\vL_{p}^{\lv} \neq x \neq \vL_{q}^{\lz} $ occurring in $\ptsrl$, $\sub''(x) = \sub(x)$. In both cases, we define the substitution $\sub_{i}$ such that $\lu\sub_{i} = \lu\sub$, $\vB\sub_{i} = \vB\sub$, and $\vL\sub_{i} = \vL\sub' \seqcomp \phi_{\lw}(B_{i})\sub$ for $1 \leq i \leq m$. 

 We can now permute $\ctrl$ above $\ptsrl$ as shown below. One can confirm that $\ptsrl$ is strict under $\sub''$ as all conditions $\coni$--$\convii$ (see \dfn~\ref{def:strictness}) still hold after the permutation; in particular, $\conii$ will hold because we assume that any atom occurs at most once in $A$ by \dfn~\ref{def:primitive-tense-axiom}.
\begin{center}
\AxiomC{$(\lu : \vB, \lu : \vB \sar ) \seqcomp \vL)\sub_{1}$}
\RightLabel{$\ctrl$}
\UnaryInfC{$((\lu : \vB \sar ) \seqcomp \vL)\sub_{1}$}
\RightLabel{$=$}
\dottedLine
\UnaryInfC{$(\phi_{\lw}(A) \seqcomp \phi_{\lw}(B_{1}) \seqcomp \vL)\sub''$}

\AxiomC{$\ldots$}

\AxiomC{$(\lu : \vB, \lu : \vB \sar ) \seqcomp \vL)\sub_{m}$}
\RightLabel{$\ctrl$}
\UnaryInfC{$((\lu : \vB \sar ) \seqcomp \vL)\sub_{m}$}
\RightLabel{$=$}
\dottedLine
\UnaryInfC{$ (\phi_{\lw}(A) \seqcomp \phi_{\lw}(B_{m}) \seqcomp \vL)\sub''$}

\RightLabel{$\ptsrl$}
\TrinaryInfC{$(\phi_{\lw}(A) \seqcomp \vL)\sub''$}
\DisplayProof
\end{center}


$\slsub$. Suppose we have an instance of $\ptsrl$ followed by an instance of $\slsub$, as shown below, such that $(\phi_{\lw}(A) \seqcomp \vL)\sub = \vL\sub'$. 
\begin{center}
\AxiomC{$(\phi_{\lw}(A) \seqcomp \phi_{\lw}(B_{1}) \seqcomp \vL)\sub$}
\AxiomC{$\ldots$}
\AxiomC{$ (\phi_{\lw}(A) \seqcomp \phi_{\lw}(B_{m}) \seqcomp \vL)\sub$}
\RightLabel{$\ptsrl$}
\TrinaryInfC{$(\phi_{\lw}(A) \seqcomp \vL)\sub$}
\RightLabel{$=$}
\dottedLine
\UnaryInfC{$((R\lz\lu, R\lz\lv \sar ) \seqcomp \vL)\sub'$}
\RightLabel{$\slsub$}
\UnaryInfC{$((R\lz\lu, R\lz\lv \sar ) \seqcomp \vL)(\lu/\lv)\sub'$}
\DisplayProof
\end{center}
 We let $\lz\sub' = z$, $\lu\sub' = u$, $\lv\sub' = v$ and assume w.l.o.g. that $\slsub$ takes the form shown above. Either $(R\lz\lu, R\lz\lv \sar )\sub'$ occurs in $\phi_{\lw}(A)\sub'$, $(R\lz\lu, R\lz\lv \sar )\sub'$ occurs in $\vL\sub$, or $(R\lz\lu \sar )\sub'$ occurs in $\phi_{\lw}(A)\sub'$ and $(R\lz\lv \sar )\sub'$  occurs in $\vL\sub$ (or, vice versa). We argue the first case as the remaining cases are similar. If $(R\lz\lu, R\lz\lv \sar )\sub'$ occurs in $\phi_{\lw}(A)\sub'$, then there are four possibilities: (1) there exist schematic relational atoms $R\lx\ly$ and $R\lx\lr$ in $\phi_{\lw}(A)$ such that $(R\lx\ly)\sub = (R\lz\lu)\sub'$ and $(R\lx\lr)\sub = (R\lz\lv)\sub'$, (2) there exists a schematic relational atom $R\lx\ly$ and labeled sequent variable $\lseq_{p}^{\lo}$ in $\phi_{\lw}(A)$ such that $(R\lx\ly)\sub = (R\lz\lu)\sub'$ and $(R\lz\lv)\sub'$ occurs in $\lseq_{p}^{\lo}\sub$ (or, vice versa), (3) there exists a labeled sequent variable $\lseq_{p}^{\lo}$ in $\phi_{\lw}(A)$ such that both $(R\lz\lu)\sub'$ and $(R\lz\lv)\sub'$ occur in $\lseq_{p}^{\lo}\sub$, or (4) there exist labeled sequent variables $\lseq_{p}^{\lo}$ and $\lseq_{q}^{\vs}$ such that $(R\lz\lu)\sub'$ occurs in $\lseq_{p}^{\lo}\sub$ and $(R\lz\lv)\sub'$ occurs in $\lseq_{q}^{\vs}\sub$. We consider case (1) and note that cases (2)--(4) can be argued in a similar manner.
 
By assumption, $\phi_{\lw}(A)$ is of the form $(R\lx\ly, R\lx\lr \sar ) \seqcomp \sls$. We assume that the labeled sequent variables $\lseq_{r_{1}}^{\ly}, \ldots, \lseq_{r_{n}}^{\ly}$ and $\lseq_{t_{1}}^{\lr}, \ldots, \lseq_{t_{k}}^{\lr}$ occur in $\sls$ and note that the remaining cases (where no labeled sequent variable annotated with the label variable $\ly$ or $\lr$ occurs) are similar. Therefore, $\sls$ is of the form $\lseq_{r_{1}}^{\ly} \seqcomp \cdots \seqcomp \lseq_{r_{n}}^{\ly} \seqcomp \lseq_{t_{1}}^{\lr} \seqcomp \cdots \seqcomp \lseq_{t_{k}}^{\lr} \seqcomp \sls'$.
 By the closure condition (see p.~\pageref{closure-cond}) and \dfn~\ref{def:gtktp}, we know that $\gtktp$ contains a contraction of the rule $\ptsrl$, which we dub $(pt_{\lseq}')$, that contains $(R\lx\ly \sar ) \seqcomp \lseq_{r_{1}}^{\ly} \seqcomp \cdots \seqcomp \lseq_{r_{n}}^{\ly} \seqcomp \sls'$ rather than $(R\lx\ly, R\lx\lr \sar ) \seqcomp \sls$. Let $\lseq_{r_{i}}^{\ly} = \rel_{i}, \Gamma_{i} \sar \Delta_{i}$ and $\lseq_{t_{j}}^{\lr} = \rel_{j}', \Gamma_{j}' \sar \Delta_{j}'$ for $1 \leq i \leq n$ and $1 \leq j \leq k$. We define $\sub''$ such that
$$
\lseq_{r_{1}}^{\ly}\sub'' = \rel_{1}, \bigcup_{1 \leq j \leq k} \rel_{j}', \Gamma_{1}, \biguplus_{1 \leq j \leq k} \Gamma_{j}', \sar \Delta_{1}, \biguplus_{1 \leq j \leq k} \Delta_{j}',
$$
 and for all remaining atomic, formula, and labeled sequent variables $x$ occurring in $(pt_{\lseq}')$, $\sub''(x) = \sub(x)$. In essence, the substitution $\sub''$ `shifts' all of the labeled polytree sequents rooted at $\lr\sub$ to $\ly\sub$. We now define the substitution $\sub_{i}$ such that $\vL\sub_{i} = \vL\sub' \seqcomp \phi_{\lw}(B_{1})\sub$ and $(R\lz\lu, R\lz\lv \sar )\sub_{i} = (R\lz\lu, R\lz\lv \sar )\sub'$ for $1 \leq i \leq m$. We can permute $\slsub$ upward to derive the same conclusion, if we apply $\slsub$ under $\sub_{i}$ to the $i^{th}$ premise of $\ptsrl$ under $\sub$, and then apply $(pt_{\lseq}')$ under $\sub'$. Moreover, it can be checked that $(pt_{\lseq}')$ under $\sub''$ is strict.
\end{proof}

\section{Translating Display and Labeled Notation}\label{section-4}

 In this section, we define $\ptime$ functions that translate display sequents into labeled polytree sequents and vice versa, which will be employed in translating display and labeled proofs in the following two sections. Our translations provide new insights into the relationship between the display and labeled formalisms. Interestingly, we find that display structures are inter-translatable with labeled polytree sequents, and display sequents likewise inter-translate with (compositions of) labeled polytree sequents. Moreover, display equivalent sequents translate to isomorphic labeled sequents (\lem~\ref{lem:dis-equiv-iso}). All of this suggests that the labeled formalism is more compact and exhibits far less bureaucracy than the display formalism since seemingly distinct objects in the display formalism are easily recognized as identical in the labeled formalism.\footnote{The excessive bureaucracy of display calculi, which obfuscates identities on proofs, has been discussed in prior works~\cite{CiaLyoRamTiu21,Sto04}.} 

 The first half of this section is dedicated to defining the translation of display sequents into labeled polytree sequents and proving certain properties thereof. The latter half of the section considers the reverse translation.

\begin{definition}[Translation $\dl_{w}$]\label{def:dis-to-lab-trans} Let $X \dar Y$ be a display sequent. We define the translation $\dl_{w}(X \dar Y) := \dl^{1}_{w}(X) \scomp \dl^{2}_{w}(Y)$, where $\dl^{1}_{w}$ and $\dl^{2}_{w}$ translate display structures into labeled (polytree) sequents recursively as follows:
\vspace*{-1em}
\begin{multicols}{2}
\begin{itemize}

\item $\dl^{w}_{1}(I) := \emptyset \sar \emptyset$

\item $\dl^{w}_{1}(A) := w : A \sar \emptyset$

\item $\dl^{w}_{1}(\dneg X) := \dl^{w}_{2}(X)$

\item $\dl^{w}_{1}(\bull X) := (Ruw \sar \emptyset) \scomp \dl^{u}_{1}(X)$

\item $\dl^{w}_{1}(X \circ Y) := \dl^{w}_{1}(X) \scomp \dl^{w}_{1}(Y)$

\end{itemize}

\begin{itemize}

\item $\dl^{w}_{2}(I) := \emptyset \sar \emptyset$

\item $\dl^{w}_{2}(A) := \emptyset \sar w : A$

\item $\dl^{w}_{2}(\dneg X) := \dl^{w}_{1}(X)$

\item $\dl^{w}_{2}(\bull X) := (Rwv \sar \emptyset) \scomp \dl^{v}_{2}(X)$

\item $\dl^{w}_{2}(X \circ Y) := \dl^{w}_{2}(X) \scomp \dl^{w}_{2}(Y)$

\end{itemize}
\end{multicols}
\noindent
We note that $u$ and $v$ are fresh labels in the above translations.
\end{definition}

\begin{example} We provide an example of translating a display sequent into a labeled sequent. 
\begin{eqnarray*}
& & \dl_{w}(\bull (\dneg p \circ \p q) \dar \dneg \bull q) \\
&  & \dl^{w}_{1}(\bull (\dneg p \circ \p q)) \seqcomp \dl^{w}_{2}(\dneg \bull q)\\
&  & (Ruw \sar \empdata) \seqcomp \dl^{u}_{1}(\dneg p \circ \p q) \seqcomp \dl^{w}_{1}(\bull q)\\
&  & (Ruw \sar \empdata) \seqcomp \dl^{u}_{1}(\dneg p) \seqcomp \dl^{u}_{1}(\p q) \seqcomp (Rvw \sar \empdata) \seqcomp \dl^{v}_{1}(q)\\
&  & (Ruw \sar \empdata) \seqcomp \dl^{u}_{2}(p) \seqcomp (u : \p q \sar \empdata) \seqcomp (Rwv \sar \empdata) \seqcomp (v : q \sar \empdata)\\
&  & (Ruw \sar \empdata) \seqcomp (\empdata \sar u : p) \seqcomp (u : \p q \sar \empdata) \seqcomp (Rwv \sar \empdata) \seqcomp (v : q \sar \empdata)\\
&  & Ruw, Rvw, u :\p q, v : q \sar u : p
\end{eqnarray*}
\end{example}

 Observe that when $\dl_{w}$ encounters a $\bull$ connective, it adds a relational atom to or from a \emph{fresh} label, thus ensuring that the resulting labeled sequent is free of (un)directed cycles. In addition, as $\dl_{w}(X \dar Y)$ translates both the antecedent $X$ and consequent $Y$ from the same label $w$, the resulting labeled sequent will be connected. Hence, the above translation will always yield a labeled polytree sequent (as seen in the example above). We also find that $\lgth{\dl_{w}(\dseq)} \leq \lgth{\dseq}$ for two reasons: first, there is a one-to-one correspondence between each bullet $\bull$ and formula occurrence in $\dseq$ and each relational atom and labeled formula in $\dl_{w}(\dseq)$. Second, the length of a labeled sequent is defined by the number of relational atoms and labeled formulae it contains, though $\dseq$ may still contain other connectives (e.g. $\dneg$ and $\circ$) that contribute to its length. (NB. Recall that the length of a display sequent was defined at the beginning of \sect~\ref{section-2} and the length of a labeled sequent was defined at the beginning of \sect~\ref{section-3}.)

\begin{lemma}\label{lem:ltrans-label-sub}
Let $\dseq$ be a display sequent. Then,
\begin{enumerate}

\item $\dl_{w}(\dseq)$ is a labeled polytree sequent;

\item $\lgth{\dl_{w}(\dseq)} \leq \lgth{\dseq}$;

\item for $w,u \in \lab$, $\dl_{w}(\dseq) \iso \dl_{u}(\dseq)$.

\end{enumerate}
\end{lemma}

The $\dl_{w}$ translation also sheds light on which display calculus rules are rendered redundant in the labeled setting, as stated in the lemma below.


\begin{lemma}\label{lem:dis-equiv-iso}
Suppose that $\dseq$ is derivable from $\dseq'$ by applying a display rule, $\il$, $\ir$, $\ql$, $\qr$, $\asl$, $\asr$, $\pml$, or $\pmr$. Then, $\dl_{w}(\dseq) \iso \dl_{w}(\dseq')$.
\end{lemma}

\begin{proof} We show the $\dri$, $\drix$, $\il$, and $\qr$ cases; the remaining cases are similar. 

\begin{center}
\begin{minipage}{.45\textwidth}
\begin{eqnarray*}
\dri & & \dl_{w}(X \circ Y \dar Z)\\
& & \dl^{w}_{1}(X \circ Y) \seqcomp  \dl^{w}_{2}(Z)\\
& & \dl^{w}_{1}(X) \seqcomp \dl^{w}_{1}(Y) \seqcomp  \dl^{w}_{2}(Z)\\
& & \dl^{w}_{1}(X) \seqcomp  \dl^{w}_{2}(Z) \seqcomp \dl^{w}_{1}(Y)\\
& & \dl^{w}_{1}(X) \seqcomp  \dl^{w}_{2}(Z) \seqcomp \dl^{w}_{2}(\dneg Y)\\
& & \dl^{w}_{1}(X) \seqcomp  \dl^{w}_{2}(Z \circ \dneg Y)\\
& & \dl_{w}(X \dar Z \circ \dneg Y)
\end{eqnarray*}
\end{minipage}
\begin{minipage}{.45\textwidth}
\begin{eqnarray*}
\drix & & \dl_{w}(X \dar \bull Y)\\
& & \dl^{w}_{1}(X) \seqcomp  \dl^{w}_{2}(\bull Y)\\
& & \dl^{w}_{1}(X) \seqcomp (Rwu \sar \empdata) \seqcomp \dl^{u}_{2}(Y)\\
& & (Rwu \sar \empdata) \seqcomp \dl^{w}_{1}(X) \seqcomp \dl^{u}_{2}(Y)\\
& & \dl^{u}_{1}(\bull X) \seqcomp \dl^{u}_{2}(Y)\\
& & \dl_{u}(\bull X \dar Y)\\
& & \dl_{w}(\bull X \dar Y)
\end{eqnarray*}
\end{minipage}
\end{center}

\begin{center}
\begin{minipage}{.45\textwidth}
\begin{eqnarray*}
\il & & \dl_{w}(X \dar Y)\\
& & \dl^{w}_{1}(X) \seqcomp  \dl^{w}_{2}(Y)\\
& & (\empdata \sar \empdata) \seqcomp \dl^{w}_{1}(X) \seqcomp  \dl^{w}_{2}(Y)\\
& & \dl^{w}_{1}(I) \seqcomp \dl^{w}_{1}(X) \seqcomp  \dl^{w}_{2}(Y)\\
& & \dl^{w}_{1}(I \circ X) \seqcomp  \dl^{w}_{2}(Y)\\
& & \dl_{w}(I \circ X \dar Y)
\end{eqnarray*}
\end{minipage}
\begin{minipage}{.45\textwidth}
\begin{eqnarray*}
\qr & & \dl_{w}(X \dar I)\\
& & \dl^{w}_{1}(X) \seqcomp \dl^{w}_{2}(I)\\
& & \dl^{w}_{1}(X) \seqcomp (\empdata \sar \empdata) \phantom{xxxxxxxxx}\\
& & \dl^{w}_{1}(X) \seqcomp \dl^{w}_{1}(I)\\
& & \dl^{w}_{1}(X) \seqcomp \dl^{w}_{2}(\dneg I)\\
& & \dl_{w}(X \dar \dneg I)
\end{eqnarray*}
\end{minipage}
\end{center}
 We note that the last step in the $\drix$ derivation follows from \lem~\ref{lem:ltrans-label-sub}. All remaining steps in each derivation either follow from the definition of $\dl_{w}$ or by the properties of sequent composition.
\end{proof}

 Let us now discuss the reverse translation, i.e. translating labeled sequents to display sequents. Our translation relies on a couple operations, which we now define. First, given that $X_{1}, \ldots, X_{n}$ are structures, we define
$$
\mathop{\bigocirc}_{i = 1}^{n} X_{i} := X_{1} \circ \cdots \circ X_{n}.
$$
 Second, if $\Gamma = \{w_{1} : A_{1}, \ldots, w_{n} : A_{n}\}$ is a multiset of labeled formulae, then we define $\Gamma \rest w := A_{1} \circ \cdots \circ A_{n}$. 
 If $w$ is not a label in $\Gamma$, then $\Gamma \rest w := I$, which includes the case when $\Gamma = \emptyset$. Utilizing these two operations, we define the function $\ld_{w}$ as follows.

\begin{definition}[Translation $\ld_{w}$]\label{def:translation-lab-to-dis} Let $\lseq = \lseq_{1} \ptcomp{w} \lseq_{2}$ be a labeled polytree sequent with $\lseq_{1} := \rel_{1},\Gamma_{1} \sar \Delta_{1}$, $\lseq_{2} := \rel_{2},\Gamma_{2} \sar \Delta_{2}$, and $w \in \lab(\lseq)$. We define $\ld_{w}(\lseq) = \ldl^{w}(\lseq_{1}) \dar \ldr^{w}(\lseq_{2})$.
\[
  \ldl^{w}(\lseq) :=
  \begin{cases}
  (\Gamma_{1} \rest w) \circ \dneg (\Delta_{1} \rest w) & \text{if $\rel = \empdata$}; \\
  (\Gamma_{1} \rest w) \circ \dneg (\Delta_{1} \rest w) \circ \mathop{\bigocirc}\limits_{i = 1}^{n} \dneg \bull \dneg (\ldl^{u_{i}}(\lseq |_{u_{i}}^{w})) \circ \mathop{\bigocirc}\limits_{j = 1}^{m}  \bull  (\ldl^{v_{j}}(\lseq |_{v_{j}}^{w})) & \text{otherwise}. 
  \end{cases}
\]
\[
  \ldr^{w}(\lseq) :=
  \begin{cases}
  \dneg (\Gamma_{2} \rest w) \circ (\Delta_{2} \rest w) & \text{if $\rel = \empdata$}; \\
  \dneg (\Gamma_{2} \rest w) \circ (\Delta_{2} \rest w) \circ \mathop{\bigocirc}\limits_{i = 1}^{n} \bull ( \ldr^{u_{i}}(\lseq |_{u_{i}}^{w})) \circ \mathop{\bigocirc}\limits_{j = 1}^{m}  \dneg \bull \dneg (\ldr^{v_{j}}(\lseq |_{v_{j}}^{w})) & \text{otherwise}. 
  \end{cases}
\]
In the cases where $\rel \neq \empdata$, we assume that $Rwu_{1}, \ldots, Rwu_{n}$ are all relational atoms occurring in $\lseq$ of the form $Rwy$ and that $Rv_{1}w, \ldots, Rv_{m}w$ are all relational atoms occurring in $\lseq$ of the form $Ryw$.
\end{definition}

\begin{example} Let us take the labeled sequent $Ruw, Rvw, u :\p q, v : q \sar u : p$, which admits the following $w$-partition:
$$
(Ruw, u :\p q \sar u : p) \ptcomp{w} (Rvw, v : q \sar \empdata).
$$
 We translate this into a display sequent (as shown below) by applying \dfn~\ref{def:translation-lab-to-dis}.
\begin{eqnarray*}
& & \ld_{w}(Ruw, Rvw, u :\p q, v : q \sar u : p)\\
& & \ldl^{w}(Ruw, u :\p q \sar u : p) \dar \ldr^{w}(Rvw, v : q \sar \empdata)\\
& & (u :\p q \rest w) \circ \dneg (u :p \rest w) \circ \bull (\ldl^{u}(u :\p q \sar u : p)) \dar \ldr^{w}(Rvw, v : q \sar \empdata)\\
& & I \circ \dneg I \circ \bull ((u :\p q \rest u) \circ \dneg (u : p \rest u)) \dar \ldr^{w}(Rvw, v : q \sar \empdata)\\
& & I \circ \dneg I \circ \bull (\p q \circ \dneg p) \dar \dneg (v : q \rest w) \circ (\empdata \rest w) \circ \dneg \bull \dneg (\ldr^{v}(v : q \sar \empdata))\\
& & I \circ \dneg I \circ \bull (\p q \circ \dneg p) \dar \dneg (\empdata) \circ I \circ \dneg \bull \dneg ( \dneg (v : q \rest v) \circ (\empdata \rest v))\\
& & I \circ \dneg I \circ \bull (\p q \circ \dneg p) \dar \dneg I \circ I \circ \dneg \bull \dneg (\dneg q \circ I)
\end{eqnarray*}
With minor effort, the deductively equivalent display sequent $\bull (\dneg p \circ \p q) \dar \dneg \bull q$ can be derived from $I \circ \dneg I \circ \bull (\p q \circ \dneg p) \dar \dneg I \circ I \circ \dneg \bull \dneg (\dneg q \circ I)$, which was the display sequent that translated to $Ruw, Rvw, u :\p q, v : q \sar u : p$ in the previous example.
\end{example}

\begin{remark}\label{rmk:dis-seq-polynomial-in-lpt} For any labeled polytree sequent $\lseq$ and label $w \in \lab(\lseq)$, there exists a polynomial $p$ such that $\lgth{\ld_{w}(\lseq)} \leq p(\lgth{\lseq})$.
\end{remark}

 The above translation is defined relative to a given $w$-partition of the input labeled polytree sequent $\lseq$, and so, the question naturally arises if the choice of partition affects the result of the translation. As it so happens, each partition of the input may yield a distinct display sequent as the output, yet, all such display sequents are mutually derivable from one another using only polynomially many inferences in $\lgth{\lseq}$. We prove this fact in \lem~\ref{lem:ld-equivalence-relative-to-labels} below, but first prove two helpful lemmata. 

\begin{lemma}\label{lem:ld-partition-invariance}
Let $\lseq$ be a labeled polytree sequent with $w \in \lab(\lseq)$. For any two partitions $\lseq_{1} \ptcomp{w} \lseq_{2}$ and $\lseq_{1}' \ptcomp{w} \lseq_{2}'$ of $\lseq$, $\ld_{w}(\lseq_{1} \ptcomp{w} \lseq_{2}) \equiv \ld_{w}(\lseq_{1}' \ptcomp{w} \lseq_{2}')$.
\end{lemma}

\begin{proof} Let $\lseq := \rel, \Gamma \sar \Delta$ be a labeled polytree sequent with $w \in \lab(\lseq)$. Furthermore, let $\lseq_{1} \ptcomp{w} \lseq_{2}$ and $\lseq_{1}' \ptcomp{w} \lseq_{2}'$ be $w$-partitions of $\lseq$ with $\lseq_{1} := \rel_{1}, \Gamma_{1} \sar \Delta_{1}$, $\lseq_{2} := \rel_{2}, \Gamma_{2} \sar \Delta_{2}$, $\lseq_{1}' := \rel_{1}', \Gamma_{1}' \sar \Delta_{1}'$, and $\lseq_{2}' := \rel_{2}', \Gamma_{2}' \sar \Delta_{2}'$. It follows that $\Gamma = \Gamma_{1}, \Gamma_{2} = \Gamma_{1}', \Gamma_{2}'$ and $\Delta = \Delta_{1}, \Delta_{2} = \Delta_{1}', \Delta_{2}'$. In the translations $\ld_{w}(\lseq_{1} \ptcomp{w} \lseq_{2})$ and $\ld_{w}(\lseq_{1}' \ptcomp{w} \lseq_{2}')$, we let $X$ and $Y$ be as follows.
$$
X := \mathop{\bigocirc}\limits_{i = 1}^{n} \dneg \bull \dneg (\ldl^{u_{i}}(\lseq |_{u_{i}}^{w})) \circ \mathop{\bigocirc}\limits_{j = 1}^{m}  \bull  (\ldl^{v_{j}}(\lseq |_{v_{j}}^{w}))
$$
$$
Y := \mathop{\bigocirc}\limits_{i = 1}^{n} \bull ( \ldr^{u_{i}}(\lseq |_{u_{i}}^{w})) \circ \mathop{\bigocirc}\limits_{j = 1}^{m}  \dneg \bull \dneg (\ldr^{v_{j}}(\lseq |_{v_{j}}^{w}))
$$
 Let us now prove that $\ld_{w}(\lseq_{1} \ptcomp{w} \lseq_{2}) \equiv \ld_{w}(\lseq_{1}' \ptcomp{w} \lseq_{2}')$. We note that the $\pml$ and $\pmr$ rules are applied a sufficient number of times below to re-write $(\Gamma_{2} \rest w) \circ  (\Gamma_{1} \rest w)$ as $(\Gamma_{2}' \rest w) \circ (\Gamma_{1}' \rest w)$ and $(\Delta_{1} \rest w) \circ  (\Delta_{2} \rest w)$ as $(\Delta_{1}' \rest w) \circ (\Delta_{2}' \rest w)$.
\begin{center}
\AxiomC{$\ld_{w}(\lseq_{1} \ptcomp{w} \lseq_{2})$}
\dottedLine
\RightLabel{=}
\UnaryInfC{$(\Gamma_{1} \rest w) \circ \dneg (\Delta_{1} \rest w) \circ X \dar \dneg (\Gamma_{2} \rest w) \circ (\Delta_{2} \rest w) \circ Y$}
\RightLabel{$\drii$}
\doubleLine
\UnaryInfC{$(\Gamma_{2} \rest w) \circ (\Gamma_{1} \rest w) \circ \dneg (\Delta_{1} \rest w) \circ X \dar (\Delta_{2} \rest w) \circ Y$}
\RightLabel{$\pml \times n_{1}$}
\doubleLine
\UnaryInfC{$\dneg (\Delta_{1} \rest w) \circ (\Gamma_{2}' \rest w) \circ (\Gamma_{1}' \rest w) \circ X \dar (\Delta_{2} \rest w) \circ Y$}
\RightLabel{$\driv$}
\doubleLine
\UnaryInfC{$(\Gamma_{2}' \rest w) \circ (\Gamma_{1}' \rest w) \circ X \dar (\Delta_{1} \rest w) \circ (\Delta_{2} \rest w) \circ Y$}
\RightLabel{$\pmr \times n_{2}$}
\doubleLine
\UnaryInfC{$(\Gamma_{2}' \rest w) \circ (\Gamma_{1}' \rest w) \circ X \dar (\Delta_{1}' \rest w) \circ (\Delta_{2}' \rest w) \circ Y$}
\RightLabel{$\driii$}
\doubleLine
\UnaryInfC{$\ast (\Delta_{1}' \rest w) \circ (\Gamma_{2}' \rest w) \circ (\Gamma_{1}' \rest w) \circ X \dar (\Delta_{2}' \rest w) \circ Y$}
\RightLabel{$\pml \times n_{3}$}
\doubleLine
\UnaryInfC{$(\Gamma_{2}' \rest w) \circ (\Gamma_{1}' \rest w) \circ \ast (\Delta_{1}' \rest w) \circ X \dar (\Delta_{2}' \rest w) \circ Y$}
\RightLabel{$\drii$}
\doubleLine
\UnaryInfC{$(\Gamma_{1}' \rest w) \circ \ast (\Delta_{1}' \rest w) \circ X \dar \ast (\Gamma_{2}' \rest w) \circ (\Delta_{2}' \rest w) \circ Y$}
\dottedLine
\RightLabel{=}
\UnaryInfC{$\ld_{w}(\lseq_{1}' \ptcomp{w} \lseq_{2}')$}
\DisplayProof
\end{center}
\end{proof}

\begin{lemma}\label{lem:D1-D2-Equivalence-Rules} Let $\lseq$ be a labeled polytree sequent. The following rules are derivable in $\dktp$:


\begin{center}
\begin{tabular}{c c c c}
\AxiomC{$\ldl^{w}(\lseq) \dar X$}
\doubleLine
\UnaryInfC{$\dneg\ldr^{w}(\lseq) \dar X$}
\DisplayProof

&

\AxiomC{$\dneg \ldl^{w}(\lseq) \dar X$}
\doubleLine
\UnaryInfC{$\ldr^{w}(\lseq) \dar X$}
\DisplayProof

&

\AxiomC{$X \dar \ldr^{w}(\lseq)$}
\doubleLine
\UnaryInfC{$X \dar \dneg \ldl^{w}(\lseq)$}
\DisplayProof

&

\AxiomC{$X \dar \dneg\ldr^{w}(\lseq)$}
\doubleLine
\UnaryInfC{$X \dar \ldl^{w}(\lseq)$}
\DisplayProof
\end{tabular}
\end{center}


\end{lemma}

\begin{proof} We prove the lemma by induction on the number of labels in $\lseq$, and prove the claim for the first rule; the remaining cases are similar.

\textit{Base case.} If $\lseq = \empdata \sar \empdata$, then the claim follows trivially by the $\ql$ rule. Let us suppose that $\lseq = \Gamma \sar \Delta$, where every labeled formula in $\Gamma$ and $\Delta$ has the same label $w$. Then, the rule is derived as shown below, where $\ldl^{w}(\lseq) = (\Gamma \rest w) \circ \dneg (\Delta \rest w)$ and $\dneg\ldr^{w}(\lseq) = \dneg (\dneg (\Gamma \rest w) \circ (\Delta \rest w))$.

\begin{center}
\AxiomC{$(\Gamma \rest w) \circ \dneg (\Delta \rest w) \dar X$}
\doubleLine
\RightLabel{$\dri$}
\UnaryInfC{$(\Gamma \rest w) \dar X \circ \dneg \dneg (\Delta \rest w)$}
\doubleLine
\RightLabel{$\driv$}
\UnaryInfC{$\dneg X \circ (\Gamma \rest w) \dar \dneg \dneg (\Delta \rest w)$}
\doubleLine
\RightLabel{$\drviii$}
\UnaryInfC{$\dneg X \circ (\Gamma \rest w) \dar (\Delta \rest w)$}
\doubleLine
\RightLabel{$\driv$}
\UnaryInfC{$(\Gamma \rest w) \dar X \circ (\Delta \rest w)$}
\doubleLine
\RightLabel{$\il$}
\UnaryInfC{$I \circ (\Gamma \rest w) \dar X \circ (\Delta \rest w)$}
\RightLabel{$\dri$}
\UnaryInfC{$I \dar X \circ (\Delta \rest w) \circ \dneg (\Gamma \rest w)$}
\doubleLine
\RightLabel{$\driii$}
\UnaryInfC{$I \circ \dneg (\dneg (\Gamma \rest w) \circ (\Delta \rest w)) \dar X$}
\doubleLine
\RightLabel{$\il$}
\UnaryInfC{$\dneg (\dneg (\Gamma \rest w) \circ (\Delta \rest w)) \dar X$}
\DisplayProof
\end{center}

\textit{Inductive step.} Let us suppose that $\lseq$ contains $n+1$ labels. By \dfn~\ref{def:translation-lab-to-dis}, $\ldl^{w}(\lseq) \dar X$ is equal to the following display sequent.
$$
(\Gamma_{1} \rest w) \circ \dneg (\Delta_{1} \rest w) \circ \mathop{\bigocirc}\limits_{i = 1}^{n} \dneg \bull \dneg (\ldl^{u_{i}}(\lseq |_{u_{i}}^{w})) \circ \mathop{\bigocirc}\limits_{j = 1}^{m}  \bull  (\ldl^{v_{j}}(\lseq |_{v_{j}}^{w})) \dar X
$$
Observe that each $\ldl^{u_{i}}(\lseq |_{u_{i}}^{w})$ and $\ldl^{v_{j}}(\lseq |_{v_{j}}^{w})$ is an a-part. Hence, we can invoke the display theorem (\thm~\ref{thm:display-theorem}) to display each such structure; then, by applying \ih (which is permitted since each labeled sequent $\lseq |_{u_{i}}^{w}$ and $\lseq |_{v_{j}}^{w}$ has a fewer number of labels), we obtain the following equivalent display sequent.
$$
(\Gamma_{1} \rest w) \circ \dneg (\Delta_{1} \rest w) \circ \mathop{\bigocirc}\limits_{i = 1}^{n} \dneg \bull \dneg (\dneg \ldr^{u_{i}}(\lseq |_{u_{i}}^{w})) \circ \mathop{\bigocirc}\limits_{j = 1}^{m}  \bull  (\dneg \ldr^{v_{j}}(\lseq |_{v_{j}}^{w})) \dar X
$$
Using display and reversible structural rules, we may derive the following display sequent.
$$
\dneg ( \dneg (\Gamma_{1} \rest w) \circ (\Delta_{1} \rest w) \circ \mathop{\bigocirc}\limits_{i = 1}^{n} \bull ( \ldr^{u_{i}}(\lseq |_{u_{i}}^{w})) \circ \mathop{\bigocirc}\limits_{j = 1}^{m}  \dneg \bull \dneg (\ldr^{v_{j}}(\lseq |_{v_{j}}^{w})) ) \dar X
$$
The above display sequent is equal to $\dneg\ldr^{w}(\lseq) \dar X$, thus giving us our desired conclusion. Also, we note that all of the rules used to derive the conclusion are reversible, implying that the premise may be derived from the conclusion as well.
\end{proof}

\begin{lemma}\label{lem:ld-equivalence-relative-to-labels}
Let $\lseq = \rel, \Gamma \sar \Delta$ be a labeled polytree sequent with $w,u \in \lab(\lseq)$. Then, 
 there exists a polynomial $p$ such that $\ld_{w}(\lseq)$ and $\ld_{u}(\lseq)$ are mutually derivable from one another with at most $p(\lgth{\lseq})$ many reversible rule applications.
\end{lemma}

\begin{proof} Let $\lseq  := \rel, \Gamma \sar \Delta$. We first argue that $\ld_{w}(\lseq) \equiv \ld_{u}(\lseq)$ by induction on the length of the minimal path of relational atoms in $\rel$ from $w$ to $u$. After, we argue that only polynomially many reversible rules were applied to derive $\ld_{u}(\lseq)$ from $\ld_{w}(\lseq)$. Let us consider two partitions of $\lseq$, namely, (i) $\lseq = \lseq_{1} \ptcomp{w} \lseq_{2}$ with $\lseq_{1} := \rel_{1},\Gamma_{1} \sar \Delta_{1}$ and $\lseq_{2} := \rel_{2},\Gamma_{2} \sar \Delta_{2}$, and (ii) $\lseq = \lseq_{1}' \ptcomp{u} \lseq_{2}'$. 
 By \lem~\ref{lem:ld-partition-invariance}, $\lseq_{1} \ptcomp{w} \lseq_{2}$ and $\lseq_{1}' \ptcomp{u} \lseq_{2}'$ may be any arbitrary $w$-partition and $u$-partition of $\lseq$.

\textit{Base case.} We assume that length of the minimal path between $w$ and $u$ is is one as the case when the length of the minimal path is zero (i.e. $w = u$) follows from \lem~\ref{lem:ld-partition-invariance} above. We prove that $\ld_{w}(\lseq) \equiv  \ld_{u}(\lseq)$ and have four cases to consider: either $Rwu \in \rel_{1}$, $Ruw \in \rel_{1}$, $Rwu \in \rel_{2}$, or $Ruw \in \rel_{2}$. We argue the first case as the others are proven in a similar fashion.

The display sequent $\ld_{w}(\lseq) = \ldl^{w}(\lseq_{1}) \dar \ldr^{w}(\lseq_{2})$ is of the form shown below top, and we let $u_{1} = u$ for convenience.
$$
\ld_{w}(\lseq) = (\Gamma_{1} \rest w) \circ \dneg (\Delta_{1} \rest w) \circ \dneg \bull \dneg \ldl^{u}(\lseq |^{w}_{u}) \circ X \dar \dneg (\Gamma_{2} \rest w) \circ (\Delta_{2} \rest w) \circ Y
$$
$$
X := \mathop{\bigocirc}\limits_{i = 2}^{n} \dneg \bull \dneg \ldl^{u_{i}}(\lseq |_{u_{i}}^{w}) \circ \mathop{\bigocirc}\limits_{j = 1}^{m} \bull \ldl^{v_{j}}(\lseq |_{v_{j}}^{w})
\quad
Y := \mathop{\bigocirc}\limits_{i = 1}^{k} \bull \ldr^{x_{i}}(\lseq |_{x_{i}}^{w}) \circ \mathop{\bigocirc}\limits_{j = 1}^{l} \dneg \bull \dneg \ldr^{y_{j}}(\lseq |_{y_{j}}^{w})
$$
We may now apply the following sequence of rules to the display sequent above to derive the conclusion shown below.

\begin{center}
\AxiomC{$(\Gamma_{1} \rest w) \circ \dneg (\Delta_{1} \rest w) \circ \dneg \bull \dneg \ldl^{u}(\lseq |^{w}_{u}) \circ X \dar \dneg (\Gamma_{2} \rest w) \circ (\Delta_{2} \rest w) \circ Y$}
\doubleLine
\RightLabel{$\drii$}
\UnaryInfC{$(\Gamma_{2} \rest w) \circ (\Gamma_{1} \rest w) \circ \dneg (\Delta_{1} \rest w) \circ \dneg \bull \dneg \ldl^{u}(\lseq |^{w}_{u}) \circ X \dar (\Delta_{2} \rest w) \circ Y$}
\doubleLine
\RightLabel{$\pml$}
\UnaryInfC{$\dneg (\Delta_{1} \rest w) \circ (\Gamma_{1} \rest w) \circ (\Gamma_{2} \rest w)\circ \dneg \bull \dneg \ldl^{u}(\lseq |^{w}_{u}) \circ X \dar (\Delta_{2} \rest w) \circ Y$}
\doubleLine
\RightLabel{$\driv$}
\UnaryInfC{$(\Gamma_{1} \rest w) \circ (\Gamma_{2} \rest w)\circ \dneg \bull \dneg \ldl^{u}(\lseq |^{w}_{u}) \circ X \dar (\Delta_{1} \rest w) \circ (\Delta_{2} \rest w) \circ Y$}
\RightLabel{=}
\dottedLine
\UnaryInfC{$(\Gamma \rest w) \circ \dneg \bull \dneg \ldl^{u}(\lseq |^{w}_{u}) \circ X \dar (\Delta \rest w) \circ Y$}
\RightLabel{$\drii$}
\doubleLine
\UnaryInfC{$\dneg \bull \dneg \ldl^{u}(\lseq |^{w}_{u}) \circ X \dar \dneg (\Gamma \rest w) \circ (\Delta \rest w) \circ Y$}
\DisplayProof
\end{center}
By the display theorem (\thm~\ref{thm:display-theorem}), we can display each a-part $\ldl^{u_{i}}(\lseq |_{u_{i}}^{w})$ and $\ldl^{v_{j}}(\lseq |_{v_{j}}^{w})$ in $X$ and apply \lem~\ref{lem:D1-D2-Equivalence-Rules} to derive the following display sequent.
$$
\dneg \bull \dneg \ldl^{u}(\lseq |^{w}_{u}) \circ \mathop{\bigocirc}\limits_{i = 2}^{n} \dneg \bull \dneg \dneg \ldr^{u_{i}}(\lseq |_{u_{i}}^{w}) \circ \mathop{\bigocirc}\limits_{j = 1}^{m} \bull \dneg \ldr^{v_{j}}(\lseq |_{v_{j}}^{w}) \dar \dneg (\Gamma \rest w) \circ (\Delta \rest w) \circ Y
$$
Applying display and reversible structural rules, the following display sequent is derivable from the one above.
$$
\dneg \bull \dneg \ldl^{u}(\lseq |^{w}_{u}) \circ \mathop{\bigocirc}\limits_{i = 2}^{n} \dneg \bull \ldr^{u_{i}}(\lseq |_{u_{i}}^{w}) \circ \mathop{\bigocirc}\limits_{j = 1}^{m} \bull \dneg \ldr^{v_{j}}(\lseq |_{v_{j}}^{w}) \dar \dneg (\Gamma \rest w) \circ (\Delta \rest w) \circ Y
$$
 Again, we may apply display and reversible structural rules to derive the following display sequent from the one above.
$$
\dneg \bull \dneg \ldl^{u}(\lseq |^{w}_{u}) \dar \dneg (\Gamma \rest w) \circ (\Delta \rest w) \circ Y'
$$
$$
Y' := \mathop{\bigocirc}\limits_{i = 2}^{n} \bull \ldr^{u_{i}}(\lseq |_{u_{i}}^{w}) \circ \mathop{\bigocirc}\limits_{i = 1}^{k} \bull \ldr^{x_{i}}(\lseq |_{x_{i}}^{w}) \circ \mathop{\bigocirc}\limits_{j = 1}^{m} \dneg \bull \dneg \ldr^{v_{j}}(\lseq |_{v_{j}}^{w}) \circ \mathop{\bigocirc}\limits_{j = 1}^{l} \dneg \bull \dneg \ldr^{y_{j}}(\lseq |_{y_{j}}^{w})
$$
By using $\deri$ from \fig~\ref{fig:derivable-rules-display}, we can derive 
$$
\ldl^{u}(\lseq |^{w}_{u}) \dar \dneg \bull \dneg (\dneg (\Gamma \rest w) \circ (\Delta \rest w) \circ Y') = \ldl^{u}(\lseq |^{w}_{u}) \dar \ldr^{u}(\lseq |^{u}_{w} \seqcomp (Rwu \sar \empdata))
$$ 
from the former display sequent above. By \lem~\ref{lem:ld-partition-invariance}, since $\lseq |^{w}_{u} \ptcomp{u} (\lseq |^{u}_{w} \seqcomp (Rwu \sar \empdata))$ is a $u$-partition of $\lseq$, we know that $\ld_{u}(\lseq |^{w}_{u} \ptcomp{u} (\lseq |^{u}_{w} \seqcomp (Rwu \sar \empdata))) \equiv \ld_{u}(\lseq_{1}' \seqcomp \lseq_{2}')$, and so, we have shown that $\ld_{w}(\lseq) \equiv \ld_{u}(\lseq)$ since only display and reversible structural rules were used above.

\textit{Inductive step.} Suppose that the length of the minimal path of relational atoms between $w$ and $u$ is $n+1$. Then, there exists a label $v$ such that the length of the minimal path from $w$ to $v$ is at most length $n$ and length of the minimal path from $v$ to $u$ is $1$. Hence, by \ih, $\ld_{w}(\lseq) \equiv \ld_{v}(\lseq)$ and $\ld_{v}(\lseq) \equiv \ld_{u}(\lseq)$, which implies the desired result and concludes the inductive step.

Last, observe that only reversible rules were used in the proofs of 
 \lem~\ref{lem:ld-partition-invariance}, \lem~\ref{lem:D1-D2-Equivalence-Rules}, and in the derivation of $\ld_{u}(\lseq)$ from $\ld_{w}(\lseq)$ above. Furthermore, the number of rules applied only depends on the length $n$ of the path from $w$ to $u$ in $\rel$ along with the number of incoming and outgoing edges (i.e. relational atoms) from nodes (i.e. labels) along this path. The above proof describes an algorithm that takes this path with incoming and outgoing edges as input and computes successive rule applications by processing ever smaller terminal segments of this path until $\ld_{u}(\lseq)$ is reached. Since the length of the path $n \leq \lgth{\lseq}$, one can verify that for a polynomial $p$, $p(\lgth{\lseq})$-many reversible inferences occur in the derivation of $\ld_{u}(\lseq)$ from $\ld_{w}(\lseq)$.
\end{proof}

\section{From Display to Labeled Proofs}\label{section-5}

 We show how to translate each display proof in $\dktp$ into a strict labeled polytree proof in $\gtktp$. Moreover, we argue that this translation is computable in $\ptime$ and does not increase the size of the proof. This result is significant for the following reasons: first, we establish that each labeled calculus $\gtktp$ can polynomially simulate each display calculus $\dktp$, showing that display calculi cannot outperform (in terms of complexity) labeled calculi. Second, we identify what subspace of labeled proofs correspond to display proofs, and in the following section, we show that each labeled proof within this space can be translated in $\ptime$ into a display proof as well, giving a characterization result. In essence, we find that display calculi are a restriction of labeled calculi and that the space of display proofs is a fragment of the space of labeled proofs.
 
 Let us first show how to translate proofs from the base system $\dkt$ into $\gtkt$. In the sequel, we explain how this translation can be augmented to translate primitive tense structural rules as well, giving the main result of the section (\thm~\ref{thm:labeled-p-simulates-display}).

\begin{theorem}\label{thm:dkt-to-gtkt}
If $\proves{\dkt}{\dproof}{\dseq}$, then there exists a $\ptime$ function $f$ in $\size{\dproof}$ 
 such that $\proves{\gtkt}{f(\dproof)}{\dl_{w}(\dseq)}$, $\size{f(\dproof)} = \mathcal{O}(\size{\dproof}^{2})$, and $f(\dproof)$ is a labeled polytree proof.
\end{theorem}

\begin{proof} Suppose $\proves{\dkt}{\dproof}{\dseq}$. We show by induction on the quantity of $\dproof$ that it can be translated in $\ptime$ into a labeled polytree proof $f(\dproof)$ of $\dl_{w}(\dseq)$ in $\gtkt$ such that $\size{f(\dproof)} \leq \size{\dproof}$. By \lem~\ref{lem:dis-equiv-iso}, we need not translate instances of $\il$, $\ir$, $\ql$, $\qr$, $\asl$, $\asr$, $\pml$, $\pmr$, or the display rules, as all such rules produce isomorphic or identical labeled polytree sequents under the $\dl_{w}$ translation. Below, when we invoke the qp-admissibility of $\wk$ (\thm~\ref{lem:admiss-rules-labeled}), we note that $\wk$ will always be an instance of $\swk$, which ensures that the polytree structure of labeled sequents in the target proof will be preserved (\thm~\ref{thm:polytree-admiss}).

\textit{Base case.} We show how the initial rules $\id$, $\topr$, and $\botl$ are translated from $\dkt$ to initial rules in $\gtktii$. Observe that each initial sequent is a labeled polytree sequent.
\begin{center}
\begin{tabular}{c c c @{\hskip 1em} c c c}
\AxiomC{}
\RightLabel{$\id$}
\UnaryInfC{$p \dar p$}
\DisplayProof

&

$\leadsto$

&

\AxiomC{}
\RightLabel{$\id$}
\UnaryInfC{$w : p \sar w : p$}
\RightLabel{=}
\dottedLine
\UnaryInfC{$\dll^{w}(p) \seqcomp \dlr^{w}(p)$}
\RightLabel{=}
\dottedLine
\UnaryInfC{$\dl_{w}(p \dar p)$}
\DisplayProof

&

\AxiomC{}
\RightLabel{$\topr$}
\UnaryInfC{$I \dar \top$}
\DisplayProof

&

$\leadsto$

&

\AxiomC{}
\RightLabel{$\topr$}
\UnaryInfC{$\empdata \sar w : \top$}
\RightLabel{=}
\dottedLine
\UnaryInfC{$\dll^{w}(I) \seqcomp \dlr^{w}(\top)$}
\RightLabel{=}
\dottedLine
\UnaryInfC{$\dl_{w}(I \dar \top)$}
\DisplayProof
\end{tabular}
\end{center}
\begin{center}
\begin{tabular}{c c c}
\AxiomC{}
\RightLabel{$\botl$}
\UnaryInfC{$\bot \dar I$}
\DisplayProof

&

$\leadsto$

&

\AxiomC{}
\RightLabel{$\botl$}
\UnaryInfC{$w : \bot \sar \empdata$}
\RightLabel{=}
\dottedLine
\UnaryInfC{$\dll^{w}(\bot) \seqcomp \dlr^{w}(I)$}
\RightLabel{=}
\dottedLine
\UnaryInfC{$\dl_{w}(\bot \dar I)$}
\DisplayProof
\end{tabular}
\end{center}

\textit{Inductive step.} For the inductive step, we give a few interesting cases of translating inferences from the given display proof into the target labeled proof. The remaining cases can be found in the appendix as they are either simple or similar. In each case, it can be checked that the result is a labeled polytree proof.
\begin{center}
\begin{tabular}{c c c}
\AxiomC{$I \dar \ry$}
\RightLabel{$\topl$}
\UnaryInfC{$\top \dar \ry$}
\DisplayProof

&

$\leadsto$

&

\AxiomC{}
\RightLabel{\ih}
\UnaryInfC{$\dl_{w}(I \dar \ry)$}
\RightLabel{=}
\dottedLine
\UnaryInfC{$\dll^{w}(I) \seqcomp \dlr^{w}(\ry)$}
\RightLabel{$\swk$}
\dashedLine
\UnaryInfC{$(w : \top \sar \empdata) \seqcomp \dll^{w}(I) \seqcomp \dlr^{w}(\ry)$}
\RightLabel{=}
\dottedLine
\UnaryInfC{$\dll^{w}(\top) \seqcomp \dlr^{w}(\ry)$}
\RightLabel{=}
\dottedLine
\UnaryInfC{$\dl_{w}(\top \dar \ry)$}
\DisplayProof
\end{tabular}
\end{center}
\begin{center}
\begin{tabular}{c c c}
\AxiomC{$\rx \dar \dneg A$}
\RightLabel{$\negr$}
\UnaryInfC{$\rx \dar \neg A$}
\DisplayProof

&

$\leadsto$

&

\AxiomC{}
\RightLabel{\ih}
\UnaryInfC{$\dl_{w}(\rx \dar \dneg A)$}
\RightLabel{=}
\dottedLine
\UnaryInfC{$\dlr^{w}(\rx) \seqcomp \dll^{w}(A)$}
\RightLabel{$\negr$}
\UnaryInfC{$\dll^{w}(\rx) \seqcomp \dlr^{w}(\neg A)$}
\RightLabel{=}
\dottedLine
\UnaryInfC{$\dl_{w}(\rx \dar \neg A)$}
\DisplayProof
\end{tabular}
\end{center}
When translating the $\impl$ rule below, observe that we invoke the qp-admissibility of $\swk$ (\lem~\ref{lem:admiss-rules-labeled}) to ensure that the contexts of the premises in the labeled proof match.
\begin{flushleft}
\begin{tabular}{c c}
\AxiomC{$\rx \dar A$}
\AxiomC{$B \dar \ry$}
\RightLabel{$\impl$}
\BinaryInfC{$A \rightarrow B \dar \dneg\rx \circ \ry$}
\DisplayProof
&

$\leadsto$
\end{tabular}
\end{flushleft}
\begin{flushright}
\begin{tabular}{c}
\AxiomC{}
\RightLabel{\ih}
\UnaryInfC{$\dl_{w}(\rx \dar A)$}
\RightLabel{=}
\dottedLine
\UnaryInfC{$\dll^{w}(\rx) \seqcomp \dlr^{w}(A)$}
\RightLabel{=}
\dottedLine
\UnaryInfC{$\dll^{w}(\rx) \seqcomp (\empdata \sar w : A)$}
\RightLabel{$\swk$}
\dashedLine
\UnaryInfC{$\dll^{w}(\rx) \seqcomp \dlr^{w}(\ry) \seqcomp (\empdata \sar w : A)$}

\AxiomC{}
\RightLabel{\ih}
\UnaryInfC{$\dl_{w}(B \dar \ry)$}
\RightLabel{=}
\dottedLine
\UnaryInfC{$\dll^{w}(B) \seqcomp \dlr^{w}(\ry)$}
\RightLabel{=}
\dottedLine
\UnaryInfC{$(w : B \sar \empdata) \seqcomp \dlr^{w}(\ry)$}
\RightLabel{$\swk$}
\dashedLine
\UnaryInfC{$(w : B \sar \empdata) \seqcomp \dll^{w}(\rx) \seqcomp \dlr^{w}(\ry)$}

\RightLabel{$\impl$}
\BinaryInfC{$\dll^{w}(\rx) \seqcomp \dlr^{w}(\ry) \seqcomp (w : A \imp B \sar \empdata)$}
\RightLabel{=}
\dottedLine
\UnaryInfC{$\dlr^{w}(\dneg \rx) \seqcomp \dlr^{w}(\ry) \seqcomp (w : A \imp B \sar \empdata)$}
\RightLabel{=}
\dottedLine
\UnaryInfC{$\dlr^{w}(\dneg \rx \circ \ry) \seqcomp (w : A \imp B \sar \empdata)$}
\RightLabel{=}
\dottedLine
\UnaryInfC{$\dll^{w}(A \imp B) \seqcomp \dlr^{w}(\dneg \rx \circ \ry)$}
\RightLabel{=}
\dottedLine
\UnaryInfC{$\dl_{w}(A \rightarrow B \sar \dneg \rx \circ \ry)$}
\DisplayProof
\end{tabular}
\end{flushright}
\lem~\ref{lem:dis-equiv-iso} is invoked in translating the $\gr$ and $\prd$ inferences below.
\begin{center}
\begin{tabular}{c c c}
\AxiomC{$\bull \rx \dar A$}
\RightLabel{$\gr$}
\UnaryInfC{$\rx \dar \g A$}
\DisplayProof

&

$\leadsto$

&

\AxiomC{}
\RightLabel{\ih}
\UnaryInfC{$\dl_{w}(\bull \rx \dar A)$}
\RightLabel{$\iso$~(\lem~\ref{lem:dis-equiv-iso})}
\dottedLine
\UnaryInfC{$\dl_{w}(\rx \dar \bull A)$}
\RightLabel{=}
\dottedLine
\UnaryInfC{$\dll^{w}(\rx) \seqcomp \dlr^{w}(\bull A)$}
\RightLabel{=}
\dottedLine
\UnaryInfC{$\dll^{w}(\rx) \seqcomp (Rwu \sar u : A)$}
\RightLabel{$\gr$}
\UnaryInfC{$\dll^{w}(\rx) \seqcomp (\empdata \sar w : \g A)$}
\RightLabel{=}
\dottedLine
\UnaryInfC{$\dll^{w}(\rx) \seqcomp \dlr^{w}(\g A)$}
\RightLabel{=}
\dottedLine
\UnaryInfC{$\dl_{w}(\rx \dar \g A)$}
\DisplayProof
\end{tabular}
\end{center}
\begin{center}
\begin{tabular}{c c c}
\AxiomC{$X \dar A$}
\RightLabel{$\prd$}
\UnaryInfC{$\bull X \dar \p A$}
\DisplayProof

&

$\leadsto$

&

\AxiomC{}
\RightLabel{\ih}
\UnaryInfC{$\dl_{w}(\rx \dar A)$}
\RightLabel{=}
\dottedLine
\UnaryInfC{$\dll^{w}(\rx) \seqcomp \dlr^{w}(A)$}
\RightLabel{=}
\dottedLine
\UnaryInfC{$\dll^{w}(\rx) \seqcomp (\empdata \sar w : A)$}
\RightLabel{$\prii$}
\UnaryInfC{$\dll^{w}(\rx) \seqcomp (Ruw \sar u : \p A)$}
\RightLabel{=}
\dottedLine
\UnaryInfC{$\dll^{w}(\rx) \seqcomp (Ruw \sar \empdata) \seqcomp (\empdata \sar u : \p A)$}
\RightLabel{=}
\dottedLine
\UnaryInfC{$\dll^{w}(\bull \rx) \seqcomp \dlr^{w}(A)$}
\RightLabel{=}
\dottedLine
\UnaryInfC{$\dl_{u}(\bull \rx \dar A)$}
\RightLabel{$\iso$~(\lem~\ref{lem:ltrans-label-sub})}
\dottedLine
\UnaryInfC{$\dl_{w}(\bull \rx \dar A)$}
\DisplayProof
\end{tabular}
\end{center}
Since $\swk$ is not an explicit rule in $\gtktii$, we apply the qp-admissibility of $\swk$ below to resolve the case of translating $\dwl$.
\begin{center}
\begin{tabular}{c c c}
\AxiomC{$\rx \vdash \ry$}
\RightLabel{$\dwl$}
\UnaryInfC{$\rz \circ \rx \vdash \ry$}
\DisplayProof

&

$\leadsto$

&

\AxiomC{}
\RightLabel{\ih}
\UnaryInfC{$\dl_{w}(\rx \vdash \ry)$}
\RightLabel{=}
\dottedLine
\UnaryInfC{$\dll^{w}(\rx) \seqcomp \dlr^{w}(\ry)$}
\RightLabel{$\swk$}
\dashedLine
\UnaryInfC{$\dll^{w}(\rz) \seqcomp \dll^{w}(\rx) \seqcomp \dlr^{w}(\ry)$}
\RightLabel{=}
\dottedLine
\UnaryInfC{$\dll^{w}(\rz \circ \rx) \seqcomp \dlr^{w}(\ry)$}
\RightLabel{=}
\dottedLine
\UnaryInfC{$\dl_{u}(\rz \circ \rx \dar \ry)$}
\RightLabel{$\iso$~(\lem~\ref{lem:ltrans-label-sub})}
\dottedLine
\UnaryInfC{$\dl_{w}(\rz \circ \rx \dar \ry)$}
\DisplayProof
\end{tabular}
\end{center}
 As can be seen below, $\dcr$ relies on the qp-admissibility of $\slsub$, $\ctrl$, and $\ctrr$ in $\gtkt$. After $\dll^{w}(\rx) \seqcomp  \dlr^{w}(\ry) \seqcomp \dlr^{w}(\ry)$ is obtained, we apply $\slsub$ to identify the two isomorphic copies of $\dlr^{w}(\ry)$, making both copies \emph{identical}, and then the contraction rules $\ctrl$ and $\ctrr$ are applied to contract identical labeled formulae, yielding the desired conclusion. We note that these operations can be applied in a manner to produce a labeled polytree proof as the output (see \eg~\ref{eg:zip-polytree}).
\begin{center}
\begin{tabular}{c c c}
\AxiomC{$\rx \vdash \ry \circ \ry$}
\RightLabel{$\dcr$}
\UnaryInfC{$\rx \vdash \ry$}
\DisplayProof

&

$\leadsto$

&

\AxiomC{}
\RightLabel{\ih}
\UnaryInfC{$\dl_{w}(\rx \vdash \ry \circ \ry)$}
\RightLabel{=}
\dottedLine
\UnaryInfC{$\dll^{w}(\rx) \seqcomp \dlr^{w}(\ry \circ \ry)$}
\RightLabel{=}
\dottedLine
\UnaryInfC{$\dll^{w}(\rx) \seqcomp  \dlr^{w}(\ry) \seqcomp \dlr^{w}(\ry)$}
\RightLabel{$\slsub$, $\ctrl$, $\ctrr$}
\dashedLine
\UnaryInfC{$\dll^{w}(\rx) \seqcomp \dlr^{w}(\ry)$}
\RightLabel{=}
\dottedLine
\UnaryInfC{$\dl_{u}(\rx \dar \ry)$}
\DisplayProof
\end{tabular}
\end{center}
 It is straightforward to check that the above translation is in $\ptime$ relative to $\size{\dproof}$ since each display sequent in $\dproof$ is translated in $\ptime$ and each qp-admissible rule application transforms the target proof in $\ptime$ (\rmk~\ref{rmk:qp-admiss-PTIME}). Also, since each inference translated derives a labeled polytree sequent, $f(\dproof)$ will be a labeled polytree proof by \prp~\ref{prop:gtkt-upward-preserves-polytree}. 
 One can confirm in each case that $\qty{f(\dproof)} \leq \qty{\dproof}$ since either a single inference is performed in the target proof, an qp-admissible rule (that does not appear in the target proof) is applied, or a redundant, ineffectual inference occurs (see \lem~\ref{lem:dis-equiv-iso}). By \lem~\ref{lem:ltrans-label-sub}, we know that for each $\dseq$ in the input proof $\lgth{\dl_{w}(\dseq)} \leq \lgth{\dseq}$. Nevertheless, it is possible for the length of a labeled polytree sequent to increase in the target proof if the qp-admissibility of $\swk$ is applied. However, the length can only increase by at most $\wdth{\dproof} \times \qty{\dproof}$ since applying the qp-admissibility of $\swk$ in the target proof runs an algorithm (see \rmk~\ref{rmk:qp-admiss-PTIME}) that permutes $\swk$ upward and `deposits' at most $\wdth{\dproof}$ amount of syntactic material in each labeled sequent (increasing its length) at most $\qty{\dproof}$ many times. Putting all of the above together, we have that
$$
\size{f(\dproof)} = \qty{f(\dproof)} \times \wdth{f(\dproof)} \leq \qty{\dproof} \times \qty{\dproof} \times \wdth{\dproof}
$$
 showing that $\size{f(\dproof)} = \mathcal{O}(\size{\dproof}^{2})$.
\end{proof}

Let us now consider translating display proofs to labeled (polytree) proofs in the presence of primitive tense structural rules. We carry out this translation by considering an instance of a primitive tense structural rule $\ptsrd$ under a substitution $\sub$, and show that the instance can be translated via $\dl_{w}$ to an instance of $\ptsrl$ under a related substitution, denoted $\sub^{\dl}$. 
 This is sufficient to prove the main result of this section (\thm~\ref{thm:labeled-p-simulates-display}), which states that each proof in $\dktp$ can be translated in $\ptime$ into a strict labeled polytree proof in $\gtktp$ such that the quantity of the proof does not grow.

 Let us consider an instance of a primitive tense structural rule $\ptsrd$ under a substitution $\sub$, as shown below. 
\begin{center}
\AxiomC{$(\psi(B_{1}) \dar \vX)\sub$}
\AxiomC{$\ldots$}
\AxiomC{$(\psi(B_{m}) \dar \vX)\sub$}
\RightLabel{$\ptsrd$}
\TrinaryInfC{$(\psi(A) \dar \vX)\sub$}
\DisplayProof
\end{center}
 Let us assume that $\rx_{p_{1}}, \ldots, \rx_{p_{n}}, \vX$ are all of the structure variables occurring in the $\ptsrd$ rule above with $\sub$ a substitution of the following form:
$$
\sub := [X_{1} / \vX_{p_{1}}, \ldots, X_{n} / \vX_{p_{n}}, X / \vX].
$$
 To translate the above instance of $\ptsrd$ to an instance of $\ptsrl$ under $\dl_{w}$, we define the substitution $\sub^{\dl}$ accordingly:
$$
\sub^{\dl} := [\dl_{1}^{u_{1}}(X_{1}) / \varlseq_{p_{1}}^{\lu_{1}}, \ldots, \dl_{1}^{u_{n}}(X_{n}) / \varlseq_{p_{n}}^{\lu_{n}}, \dl_{2}^{w}(X) / \varlseq, u_{1} / \lu_{1}, \ldots, u_{n} / \lu_{n}, w / \lw].
$$
 In the definition of $\sub^{\dl}$ above, we let all labels in $\dl_{1}^{u_{1}}(X_{1})$, $\ldots$, $\dl_{1}^{u_{n}}(X_{n})$, and $\dl_{2}^{w}(X)$ be fresh with the exception of $u_{1}$, $\ldots$, $u_{n}$, and $w$, respectively. We note that if $\vX_{p_{i}}\sub = X_{i}$ for $1 \leq i \leq n$, then for each labeled sequent variable $\varlseq_{p_{i}}^{\lv_{1}}, \ldots, \varlseq_{p_{i}}^{\lv_{k}}$ in $\ptsrl$ annotated with $p_{i}$, we have $\varlseq_{p_{i}}^{\lv_{j}}\sub^{\dl} = \dl_{1}^{v_{j}}(X_{i})$ and $\lv_{j}\sub^{\dl} = v_{j}$ for $1 \leq j \leq k$. In other words, the substitution of a structure for a single structure variable in $\sub$ may give rise to multiple substitutions of labeled polytree sequents for multiple labeled sequent variables in $\sub^{\dl}$. 

\begin{example}\label{eg:ptsrl-ptsrd} We give an example of translating an instance of $\ptsrd$ into an instance of $\ptsrl$. The $\ptsrd$ rule (shown below left) corresponds to the simple primitive tense axiom $\f p \rightarrow \p (p \land \f p)$, where we let $A = \f p$ and $B = \p (p \land \f p)$. An instance of the rule under the substitution $\sub = [q / \vX_{p}, \p p \lor q / \vX]$ is shown below right.
\begin{center}
\begin{tabular}{c c}
\AxiomC{$\bull (\vX_{p} \circ \dneg \bull \dneg \vX_{p}) \dar \vX$}
\RightLabel{$\ptsrd$}
\UnaryInfC{$\dneg \bull \dneg \vX_{p} \dar \vX$}
\DisplayProof

&

\AxiomC{$\bull (q \circ \dneg \bull \dneg q) \dar \p p \lor q$}
\RightLabel{$\ptsrd$}
\UnaryInfC{$\dneg \bull \dneg q \dar \p p \lor q$}
\DisplayProof
\end{tabular}
\end{center}
 If we transform $\f p \rightarrow \p (p \land \f p)$ into $\ptsrl$, we obtain the rule shown below.
\begin{center}
\AxiomC{$(R\lw\lu \sar ) \seqcomp (R\lv\lw \sar ) \seqcomp (R\lv\lz \sar ) \seqcomp \vL_{p}^{\lu} \seqcomp \vL_{p}^{\lv} \seqcomp \vL_{p}^{\lz} \dar \vL$}
\RightLabel{$\ptsrd$}
\UnaryInfC{$(R\lw\lu \sar ) \seqcomp \vL_{p}^{\lu} \dar \vL$}
\DisplayProof
\end{center}
 To obtain the corresponding instance of $\ptsrd$, we apply the following substitution
$$
\sub^{\dl} := [\lseq_{u} / \vL_{p}^{\lu}, \lseq_{v} / \vL_{p}^{\lv}, \lseq_{z} / \vL_{p}^{\lz}, \lseq / \vL, u / \lu, v / \lv, z / \lz]
$$
 where $\lseq_{u} = (u : q \sar )$, $\lseq_{v} = (v : q \sar )$, $\lseq_{z} = (z : q \sar )$, and $\lseq = ( \sar w : \p p \lor r)$. Applying $\sub^{\dl}$ to $\ptsrl$ gives the following strict instance of the rule.
\begin{center}
\AxiomC{$Rwu, Rvw, Rvz, u : q, v : q, z : q \sar w : \p p \lor q$}
\RightLabel{$\ptsrl$}
\UnaryInfC{$Rwu,  u : q \sar w : \p p \lor q$}
\DisplayProof
\end{center}
 It is interesting to observe that $\ptsrd$ under $\sub$ translates via $\dl_{w}$ into $\ptsrl$ under $\sub^{\dl}$. Since $\ptsrl$ retains a copy of $\phi_{\lw}(A)$ in its premise, we must include a copy of $\psi(A)$ when translating the premise of $\ptsrd$ under $\sub$ as shown below.
$$
\dl_{w}(\underbrace{\bull (\dneg \bull \dneg q)}_{\psi(A)\sub} \circ \underbrace{\bull (q \circ \dneg \bull \dneg q)}_{\psi(B)\sub} \dar \vX) = Rwu, Rvw, Rvz, u : q, v : q, z : q \sar w : \p p \lor q
$$
Alternatively, the conclusion of $\ptsrd$ under $\sub$ translates to the conclusion of $\ptsrl$ under $\sub^{\dl}$ without modification: $\dl_{w}(\dneg \bull \dneg q \dar \p p \lor q) = Rwu,  u : q \sar w : \p p \lor q$.
\end{example}

As demonstrated in the example above, instances of $\ptsrd$ translate to strict instances of $\ptsrl$. We use this fact below to prove the main result of the section (\thm~\ref{thm:labeled-p-simulates-display}), however, it will be helpful to first prove the following lemma.


\begin{lemma}\label{lem:translating-antecedents} Let us consider an instance of a primitive tense structural rule as shown below. 
\begin{center}
\AxiomC{$(\psi(B_{1}) \dar \vX)\sub$}
\AxiomC{$\ldots$}
\AxiomC{$(\psi(B_{m}) \dar \vX)\sub$}
\RightLabel{$\ptsrd$}
\TrinaryInfC{$(\psi(A) \dar \vX)\sub$}
\DisplayProof
\end{center}
and let $C \in \{A, B_{1}, \ldots, B_{m}\}$. Then, 
\begin{enumerate}

\item $\dl^{w}_{1}(\psi(C)\sub) = \phi_{\lw}(C)\subdl$

\item $\dl_{w}((\psi(C) \dar \vX)\sub) = (\phi_{\lw}(C) \scomp \ulseq)\subdl$

\end{enumerate}
\end{lemma}

\begin{proof} We first prove claim 1 by induction on the complexity of $C$. 

\textit{Base case.} If $C$ is an atom $p_{i}$, then $\psi(p_{i}) = \vX_{p_{i}}$ and the proof shown below left demonstrates the claim. If $A$ is the logical constant $\top$, so that $\psi(A) = \psi(\top) = I$, then the proof shown below right demonstrates the claim.\\
\begin{minipage}{.45\textwidth}
\begin{eqnarray*}
& & \dl^{w}_{1}(\psi(p_{i})\sub)\\
& & \dl^{w}_{1}(\vX_{p_{i}}\sub)\\
& & \dl^{w}_{1}(X_{i})\\
& & \ulseq_{p_{i}}^{\lw}\subdl\\
& & \phi_{\lw}(p)\subdl
\end{eqnarray*}
\end{minipage}
\begin{minipage}{.45\textwidth}
\begin{eqnarray*}
& & \dl^{w}_{1}(\psi(\top)\sub) \\
& & \dl^{w}_{1}(I\sub)\\
& & \dl^{w}_{1}(I)\\
& & \empdata \sar \empdata\\
& & \phi_{\lw}(\top)\subdl
\end{eqnarray*}
\end{minipage}\\

\textit{Inductive step.} We consider the cases where $C = D \land E$, $C = \f D$, and $C = \p D$ in turn. First, if $A = B\land C$, so that $\psi(C) = \psi(D \land E) = \psi(D) \circ \psi(E)$, then the following proof establishes the case.\\
\begin{minipage}{.95\textwidth}
\begin{eqnarray*}
& & \dl^{w}_{1}(\psi(D\land E)\sub)\\
& & \dl^{w}_{1}((\psi(D) \circ \psi(E))\sub)\\
& & \dl^{w}_{1}(\psi(D)\sub) \seqcomp \dl^{w}_{1}(\psi(E)\sub)\\
& & \phi_{\lw}(D)\subdl \otimes \phi_{\lw}(E)\subdl\\
& & (\phi_{\lw}(D) \otimes \phi_{\lw}(E))\subdl\\
& & \phi_{\lw}(D \land E)\subdl
\end{eqnarray*}
\end{minipage}\\

If $C = \f D$, so that $\psi(C) = \psi(\f D) = \dneg\bullet\dneg \psi(D)$, then the proof shown below left establishes the case, and if $C = \p D$, so that $\psi(C) = \psi(\p D) = \bull \psi(D)$, then the proof shown below right establishes the case.\\
\begin{minipage}{.45\textwidth}
\begin{eqnarray*}
& & \dl^{w}_{1}(\psi(\f D)\sigma)\\
& & \dl^{w}_{1}((\dneg\bullet\dneg\psi(D))\sigma)\\
& & \dl^{w}_{1}(\dneg\bullet\dneg(\psi(D)\sigma))\\
& & (Rwu \sar \empdata) \otimes \dl^{u}_{1}(\psi(D)\sigma)\\
& & (Rwu \sar \empdata) \otimes \phi_{\lu}(D)\subdl\\
& & ((Rwu \sar \empdata) \otimes \phi_{\lu}(D))\subdl\\
& & \phi_{\lw}(\f D)\subdl
\end{eqnarray*}
\end{minipage}
\begin{minipage}{.45\textwidth}
\begin{eqnarray*}
& & \dl^{w}_{1}(\psi(\p D)\sigma)\\
& & \dl^{w}_{1}((\bullet\psi(D))\sigma)\\
& & \dl^{w}_{1}(\bullet(\psi(D)\sigma))\\
& & (Ruw \sar \empdata) \otimes \dl^{u}_{1}(\psi(D)\sigma)\\
& & (Ruw \sar \empdata) \otimes \phi_{\lu}(D)\subdl\\
& & ((Ruw \sar \empdata) \otimes \phi_{\lu}(D))\subdl\\
& & \phi_{\lw}(\p D)\subdl
\end{eqnarray*}
\end{minipage}\\

\noindent
The following establishes claim 2 and invokes claim 1 in the third to fourth line. 
\begin{eqnarray*}
& & \dl_{w}((\psi(C) \dar \rx)\sigma) \\
& & \dl_{w}(\psi(C)\sub \dar \rx\sub) \\
& & \dl_{1}^{w}(\psi(C)\sub) \scomp \dl_{2}^{w}(\rx\sub) \\
& & \phi_{\lw}(C)\subdl \scomp \dl_{2}^{w}(X) \\
& & \phi_{\lw}(C)\subdl \scomp \ulseq\subdl \\
& & (\phi_{\lw}(C) \scomp \ulseq)\subdl 
\end{eqnarray*}
\end{proof}

\begin{theorem}\label{thm:labeled-p-simulates-display} If $\proves{\dktp}{\dproof}{\dseq}$, then there exists a $\ptime$ function $f$ in $\size{\dproof}$ such that $\proves{\gtktp}{f(\dproof)}{\dl_{w}(\dseq)}$, $\size{f(\dproof)} = \mathcal{O}(\size{\dproof}^{2})$, and $f(\dproof)$ is a strict labeled polytree proof.
\end{theorem}

\begin{proof} The proof is obtained by extending the inductive step of \thm~\ref{thm:dkt-to-gtkt} with the $\ptsrd$ case. Therefore, let us consider an instance of $\ptsrd$ as shown below top. We translate it into the $\ptsrl$ instance shown below bottom, where each equality step is obtained from \lem~\ref{lem:translating-antecedents}.
\begin{flushleft}
\begin{tabular}{c c}
\AxiomC{$(\psi(B_{1}) \dar \vX)\sub$}
\AxiomC{$\ldots$}
\AxiomC{$(\psi(B_{m}) \dar \vX)\sub$}
\RightLabel{$\ptsrd$}
\TrinaryInfC{$(\psi(A) \dar \vX)\sub$}
\DisplayProof

&

$\leadsto$
\end{tabular}
\end{flushleft}
\begin{flushright}
\AxiomC{$\dl_{w}((\psi(B_{1}) \dar \vX)\sub)$}
\RightLabel{=}
\dottedLine
\UnaryInfC{$(\phi_{\lw}(B_{1}) \seqcomp \ulseq)\subdl$}
\RightLabel{$\swk$}
\dashedLine
\UnaryInfC{$(\phi_{\lw}(A) \seqcomp \phi_{\lw}(B_{1}) \seqcomp \ulseq)\subdl$}
\AxiomC{$\cdots$}
\AxiomC{$\dl_{w}((\psi(B_{m}) \dar \vX)\sub)$}
\RightLabel{=}
\dottedLine
\UnaryInfC{$(\phi_{\lw}(B_{m}) \seqcomp \ulseq)\subdl$}
\RightLabel{$\swk$}
\dashedLine
\UnaryInfC{$(\phi_{\lw}(A) \seqcomp \phi_{\lw}(B_{m}) \seqcomp \ulseq)\subdl$}
\RightLabel{$\ptsrl$}
\TrinaryInfC{$(\phi_{\lw}(A) \seqcomp \ulseq)\subdl$}
\RightLabel{=}
\dottedLine
\UnaryInfC{$\dl_{w}((\psi(A) \dar \vX)\sub)$}
\DisplayProof
\end{flushright}
By the definition of $\dl_{w}$ and $\sub^{\dl}$, it is straightforward to confirm that the conditions $\coni$--$\convii$ hold, showing that the application of $\ptsrl$ is strict. Since our proof extends the proof of \thm~\ref{thm:dkt-to-gtkt}, only introduces strict $\ptsrl$ instances, and the target proof $f(\dproof)$ derives a labeled polytree sequent, we know $f(\dproof)$ will be a strict labeled polytree proof by \thm~\ref{thm:strict-is-polytree}. Moreover, as seen above, each instance of $\ptsrd$ translates to a single instance of $\ptsrl$ and may invoke the qp-admissibility of $\swk$. By an argument similar to the one at the end of \thm~\ref{thm:dkt-to-gtkt}, one can confirm that the translation is computed in $\ptime$ and $\size{f(\dproof)} = \mathcal{O}(\size{\dproof}^{2})$.
\end{proof}

\section{From Labeled to Display Proofs}\label{section-6}

We now study the reverse translation, showing how to transform strict labeled polytree proofs into display proofs. We find that translating a strict labeled polytree proof $\lproof$ gives a display proof $f(\lproof)$ with a higher quantity due to the introduction and use of display rules and structural rules in $f(\lproof)$. Nevertheless, the translation occurs within $\ptime$, and so we definitively confirm that the largest class of analytic (i.e. cut-free) display calculi for tense logics can be polynomially simulated by analytic labeled sequent calculi. 
 This result solves an open problem first discussed by Restall in 2006~\cite{Res06}, which was also left as an open problem in~\cite{CiaLyoRamTiu21}, concerning the existence of embeddings between labeled and display proofs.\footnote{In~\cite{CiaLyoRamTiu21}, it was shown how to translate labeled proofs into `one-sided' display calculus (i.e. shallow nested calculus) proofs for a small subclass of primitive tense logics, namely, extension of $\kt$ with \emph{path axioms} of the form $\langle ? \rangle_{1} \cdots \langle ? \rangle_{n} p \rightarrow \langle ? \rangle_{n+1} p$ with $\langle ? \rangle_{i} \in \{\f,\p\}$ for $1 \leq i \leq n+1$.}


\begin{theorem}\label{thm:gtkt-to-dkt}
If $\proves{\gtkt}{\lproof}{\lseq}$ and $\lproof$ is a labeled polytree proof, then there exists a $\ptime$ function $f$ in $\size{\lproof}$ and polynomial $p$ such that $\proves{\dkt}{f(\lproof)}{\ld_{w}(\lseq)}$ and $\size{f(\lproof)} = \mathcal{O}(p(\size{\lproof}))$.
\end{theorem}

\begin{proof} We prove the result by induction on the quantity of $\lproof$ and make a case distinction on the last rule applied. We consider a representative number of cases 
and remark that the remaining cases are similar. In each case, we translate a suitable $w$-partition of the input labeled polytree sequent. By \lem~\ref{lem:ld-partition-invariance} we are permitted to select any $w$-partition to translate.

\textit{Base case.} We show the $\id$ case below and let $\rel = \rel_{1},\rel_{2}$, $\Gamma = \Gamma_{1},\Gamma_{2}$, and $\Delta = \Delta_{1},\Delta_{2}$. 

\begin{flushleft}
\begin{tabular}{c c}
\AxiomC{}
\RightLabel{$\id$}
\UnaryInfC{$\rel,\Gamma, w:p \sar w:p,\Delta$}
\RightLabel{=}
\dottedLine
\UnaryInfC{$(\rel_{1},\Gamma_{1}, w:p \sar \Delta_{1}) \ptcomp{w} (\rel_{2},\Gamma_{2} \sar  w:p, \Delta_{2})$}
\DisplayProof

&

$\leadsto$
\end{tabular}
\end{flushleft}
\begin{flushright}
\AxiomC{}
\RightLabel{$\id$}
\UnaryInfC{$p \dar p$}
\RightLabel{$\dwl$}
\UnaryInfC{$p \circ \ldl^{w}(\lseq_{1}) \dar p$}
\RightLabel{$\dwr$}
\UnaryInfC{$p \circ \ldl^{w}(\lseq_{1}) \dar p \circ \ldr^{w}(\lseq_{2})$}
\RightLabel{=}
\dottedLine
\UnaryInfC{$\ldl^{w}(\rel_{1},\Gamma_{1}, w:p \sar \Delta_{1}) \dar \ldr^{w}(\rel_{2},\Gamma_{2} \sar  w:p, \Delta_{2})$}
\RightLabel{=}
\dottedLine
\UnaryInfC{$\ld_{w}(\rel,\Gamma, w:p \sar w:p,\Delta)$}
\DisplayProof
\end{flushright}

\textit{Inductive step.} We show the $\negr$, $\conl$, $\gr$, $\fr$, and $\pl$ cases.\\

$\negr$. In the second to third step in the target proof, we apply the definition of $\ld_{w}$ and let $(\rel, \Gamma, w :A \sar \Delta) = \lseq_{1} \ptcomp{w} (\rel_{2}, \Gamma_{2}, w:A \sar \Delta_{2})$.

\begin{flushleft}
\begin{tabular}{c c}
\AxiomC{$\rel, \Gamma, w :A \sar \Delta$}
\RightLabel{$\negr$}
\UnaryInfC{$\rel, \Gamma \sar w : \neg A, \Delta$}
\DisplayProof

&

$\leadsto$
\end{tabular}
\end{flushleft}
\begin{flushright}
\AxiomC{}
\RightLabel{IH}
\UnaryInfC{$\ld_{u}(\rel, \Gamma, w :A \sar \Delta)$}
\RightLabel{$\equiv$~(\lem~\ref{lem:ld-equivalence-relative-to-labels})}
\dashedLine
\UnaryInfC{$\ld_{w}(\rel, \Gamma, w :A \sar \Delta)$}
\RightLabel{=}
\dottedLine
\UnaryInfC{$\ldl^{w}(\lseq_{1}) \dar \ldr^{w}(\rel_{2}, \Gamma_{2}, w:A \sar \Delta_{2})$}
\RightLabel{=}
\dottedLine
\UnaryInfC{$\ldl^{w}(\lseq_{1}) \dar \dneg A \circ \ldr^{w}(\rel_{2}, \Gamma_{2} \sar \Delta_{2})$}
\RightLabel{$\driii$}
\UnaryInfC{$\ldl^{w}(\lseq_{1}) \circ \dneg \ldr^{w}(\rel_{2}, \Gamma_{2} \sar \Delta_{2}) \dar \dneg A$}
\RightLabel{$\negr$}
\UnaryInfC{$ \ldl^{w}(\lseq_{1}) \circ \dneg \ldr^{w}(\rel_{2}, \Gamma_{2} \sar \Delta_{2}) \dar \neg A$}
\RightLabel{$\driii$}
\UnaryInfC{$\ldl^{w}(\lseq_{1}) \dar \neg A \circ \ldr^{w}(\rel_{2}, \Gamma_{2} \sar \Delta_{2})$}
\RightLabel{=}
\dottedLine
\UnaryInfC{$\ldl^{w}(\lseq_{1}) \dar \ldr^{w}(\rel_{2}, \Gamma_{2}, w : \neg A \sar \Delta_{2})$}
\RightLabel{=}
\dottedLine
\UnaryInfC{$\ld_{w}(\rel, \Gamma \sar w : \neg A, \Delta)$}
\RightLabel{$\equiv$~(\lem~\ref{lem:ld-equivalence-relative-to-labels})}
\dashedLine
\UnaryInfC{$\ld_{u}(\rel, \Gamma \sar w : \neg A, \Delta)$}
\DisplayProof
\end{flushright}

$\conl$. In the second to third step in the target proof, we apply the definition of $\ld_{w}$ and let $(\rel, \Gamma, w :A, w :B \sar \Delta) =  (\rel_{1}, \Gamma_{1}, w :A, w :B \sar \Delta_{1}) \ptcomp{w} \lseq_{2}$.

\begin{flushleft}
\begin{tabular}{c c}
\AxiomC{$\rel, \Gamma, w :A, w :B \sar \Delta$}
\RightLabel{$\conl$}
\UnaryInfC{$\rel, \Gamma, w :A \wedge B  \sar \Delta$}
\DisplayProof

&

$\leadsto$
\end{tabular}
\end{flushleft}
\begin{flushright}
\AxiomC{}
\RightLabel{IH}
\UnaryInfC{$\ld_{u}(\rel, \Gamma, w :A, w :B \sar \Delta)$}
\RightLabel{$\equiv$~(\lem~\ref{lem:ld-equivalence-relative-to-labels})}
\dashedLine
\UnaryInfC{$\ld_{w}(\rel, \Gamma, w :A, w :B \sar \Delta)$}
\RightLabel{=}
\dottedLine
\UnaryInfC{$\ldl^{w}(\rel_{1}, \Gamma_{1}, w :A, w :B \sar \Delta_{1}) \dar \ldr^{w}(\lseq_{2})$}
\RightLabel{=}
\dottedLine
\UnaryInfC{$\ldl^{w}(\rel_{1}, \Gamma_{1} \sar \Delta_{1}) \circ A \circ B \dar \ldr^{w}(\lseq_{2})$}
\RightLabel{$\drii$}
\UnaryInfC{$A \circ B \dar \dneg \ldl^{w}(\rel_{1}, \Gamma_{1} \sar \Delta_{1}) \circ \ldr^{w}(\lseq_{2})$}
\RightLabel{$\conl$}
\UnaryInfC{$A \land B \dar \dneg \ldl^{w}(\rel_{1}, \Gamma_{1} \sar \Delta_{1}) \circ \ldr^{w}(\lseq_{2})$}
\RightLabel{$\drii$}
\UnaryInfC{$\ldl^{w}(\rel_{1}, \Gamma_{1} \sar \Delta_{1}) \circ A \land B \dar \ldr^{w}(\lseq_{2})$}
\RightLabel{=}
\dottedLine
\UnaryInfC{$\ldl^{w}(\rel_{1}, \Gamma_{1}, w : A \land B \sar \Delta_{1}) \dar \ldr^{w}(\lseq_{2})$}
\RightLabel{=}
\dottedLine
\UnaryInfC{$\ld_{w}(\rel, \Gamma, w : A \land B \sar \Delta)$}
\RightLabel{$\equiv$~(\lem~\ref{lem:ld-equivalence-relative-to-labels})}
\dashedLine
\UnaryInfC{$\ld_{u}(\rel, \Gamma, w : A \land B \sar \Delta)$}
\DisplayProof
\end{flushright}

$\gr$. In the second to third step in the target proof, we apply the definition of $\ld_{w}$ and let $(\rel, Rwu, \Gamma \sar u : A, \Delta) =  \lseq_{1} \ptcomp{w} (\rel_{2}, Rwv, \Gamma_{2} \sar v : A, \Delta_{2})$.

\begin{flushleft}
\begin{tabular}{c c}
\AxiomC{$\rel, Rwu, \Gamma \sar u : A, \Delta$}
\RightLabel{$\gr$}
\UnaryInfC{$\rel, \Gamma \sar w : \g A, \Delta$}
\DisplayProof

&

$\leadsto$
\end{tabular}
\end{flushleft}
\begin{flushright}
\AxiomC{}
\RightLabel{IH}
\UnaryInfC{$\ld_{u}(\rel, Rwv, \Gamma \sar v : A, \Delta)$}
\RightLabel{$\equiv$~(\lem~\ref{lem:ld-equivalence-relative-to-labels})}
\dashedLine
\UnaryInfC{$\ld_{w}(\rel, Rwv, \Gamma \sar v : A, \Delta)$}
\RightLabel{=}
\dottedLine
\UnaryInfC{$\ldl^{w}(\lseq_{1}) \dar \ldr^{w}(\rel_{2}, Rwv, \Gamma_{2} \sar v : A, \Delta_{2})$}
\RightLabel{=}
\dottedLine
\UnaryInfC{$\ldl^{w}(\lseq_{1}) \dar \bull A \circ \ldr^{w}(\rel_{2}, \Gamma_{2} \sar \Delta_{2})$}
\RightLabel{$\driii$}
\UnaryInfC{$\ldl^{w}(\lseq_{1}) \circ \dneg \ldr^{w}(\rel_{2}, \Gamma_{2} \sar \Delta_{2}) \dar \bull A$}
\RightLabel{$\drix$}
\UnaryInfC{$\bull (\ldl^{w}(\lseq_{1}) \circ \dneg \ldr^{w}(\rel_{2}, \Gamma_{2} \sar \Delta_{2})) \dar A$}
\RightLabel{$\gr$}
\UnaryInfC{$\ldl^{w}(\lseq_{1}) \circ \dneg \ldr^{w}(\rel_{2}, \Gamma_{2} \sar \Delta_{2}) \dar \g A$}
\RightLabel{$\driii$}
\UnaryInfC{$\ldl^{w}(\lseq_{1}) \dar \g A \circ \ldr^{w}(\rel_{2}, \Gamma_{2} \sar \Delta_{2})$}
\RightLabel{=}
\dottedLine
\UnaryInfC{$\ldl^{w}(\lseq_{1}) \dar \ldr^{w}(\rel_{2}, \Gamma_{2} \sar w : \g A, \Delta_{2})$}
\RightLabel{=}
\dottedLine
\UnaryInfC{$\ld_{w}(\rel, \Gamma \sar w : \g A, \Delta)$}
\RightLabel{$\equiv$~(\lem~\ref{lem:ld-equivalence-relative-to-labels})}
\dashedLine
\UnaryInfC{$\ld_{u}(\rel, \Gamma \sar w : \g A, \Delta)$}
\DisplayProof
\end{flushright}

$\fr$. In the second to third step in the target proof, we apply the definition of $\ld_{w}$ and let $(\rel, Rwv, \Gamma \sar w : \f A, v : A, \Delta) =  \lseq_{1} \ptcomp{w} (\rel_{2}, \Gamma_{2} \sar v : A, \Delta_{2})$. By \prp~\ref{prop:dis-derivable-rules} we are allowed to apply the $\deri$ and $\deriii$ rules in the target proof. Also, in the ninth and tenth lines, we let $\Gamma_{2} = \Gamma_{2} \restriction w, \Gamma_{2}'$ and $\Delta_{2} = \Delta_{2} \restriction w, \Delta_{2}'$.

\begin{flushleft}
\begin{tabular}{c c}
\AxiomC{$\rel, Rwv, \Gamma \sar w : \f A, v : A, \Delta$}
\RightLabel{$\fr$}
\UnaryInfC{$\rel, Rwv, \Gamma \sar w : \f A, \Delta$}
\DisplayProof

&

$\leadsto$
\end{tabular}
\end{flushleft}
\begin{flushright}
\AxiomC{}
\RightLabel{IH}
\UnaryInfC{$\ld_{u}(\rel, \Gamma \sar v : A, \Delta)$}
\RightLabel{$\equiv$~(\lem~\ref{lem:ld-equivalence-relative-to-labels})}
\dashedLine
\UnaryInfC{$\ld_{v}(\rel, \Gamma \sar v : A, \Delta)$}
\RightLabel{=}
\dottedLine
\UnaryInfC{$\ldl^{v}(\lseq_{1}) \dar \ldr^{v}(\rel_{2}, \Gamma_{2} \sar v : A, \Delta_{2})$}
\RightLabel{=}
\dottedLine
\UnaryInfC{$\ldl^{v}(\lseq_{1}) \dar A \circ \ldr^{v}(\rel_{2}, \Gamma_{2} \sar \Delta_{2})$}
\RightLabel{$\driii$}
\UnaryInfC{$\dneg \ldr^{v}(\rel_{2}, \Gamma_{2} \sar \Delta_{2}) \circ \ldl^{v}(\lseq_{1}) \dar A$}
\RightLabel{$\fr$}
\UnaryInfC{$\dneg\bull\dneg (\dneg \ldr^{v}(\rel_{2}, \Gamma_{2} \sar \Delta_{2}) \circ \ldl^{v}(\lseq_{1})) \dar \f A$}
\RightLabel{$\deri$}
\UnaryInfC{$\dneg \ldr^{v}(\rel_{2}, \Gamma_{2} \sar \Delta_{2}) \circ \ldl^{v}(\lseq_{1}) \dar \dneg \bull \dneg \f A$}
\RightLabel{$\driii$}
\UnaryInfC{$\ldl^{v}(\lseq_{1}) \dar \dneg \bull \dneg \f A \circ  \ldr^{v}(\rel_{2}, \Gamma_{2} \sar \Delta_{2})$}
\RightLabel{=}
\dottedLine
\UnaryInfC{$\ldl^{v}(\lseq_{1}) \dar \dneg \bull \dneg \f A \circ \dneg \bull \dneg (\dneg (\Gamma_{2} \restriction w) \circ (\Delta_{2} \restriction w)) \circ \ldr^{v}(\rel, \Gamma_{2}' \sar \Delta_{2}')$}
\RightLabel{$\deriii$}
\dashedLine
\UnaryInfC{$\ldl^{v}(\lseq_{1}) \dar \dneg \bull \dneg (\f A \circ \dneg (\Gamma_{2} \restriction w) \circ (\Delta_{2} \restriction w)) \circ \ldr^{v}(\rel, \Gamma_{2}' \sar \Delta_{2}')$}
\RightLabel{=}
\dottedLine
\UnaryInfC{$\ldl^{v}(\lseq_{1}) \dar \ldr^{v}(\rel_{2}, Rwv, \Gamma_{2} \sar w : \f A, \Delta_{2})$}
\RightLabel{=}
\dottedLine
\UnaryInfC{$\ld_{v}(\rel, Rwv, \Gamma \sar w : \f A, \Delta)$}
\RightLabel{$\equiv$~(\lem~\ref{lem:ld-equivalence-relative-to-labels})}
\dashedLine
\UnaryInfC{$\ld_{u}(\rel, Rwv, \Gamma \sar w : \f A, \Delta)$}
\DisplayProof
\end{flushright}

$\pl$. In the second to third step in the target proof, we apply the definition of $\ld_{w}$ and let $(\rel, Rvw, \Gamma, v : A \sar \Delta) =  (\rel_{1}, Rvw, \Gamma_{1} \sar v : A, \Delta_{1}) \ptcomp{w} \lseq_{2}$.

\begin{flushleft}
\begin{tabular}{c c}
\AxiomC{$\rel, Rvw, \Gamma, v : A \sar \Delta$}
\RightLabel{$\pl$}
\UnaryInfC{$\rel, \Gamma, w : \p A \sar \Delta$}
\DisplayProof

&

$\leadsto$
\end{tabular}
\end{flushleft}
\begin{flushright}
\AxiomC{}
\RightLabel{IH}
\UnaryInfC{$\ld_{u}(\rel, Rvw, \Gamma \sar v : A, \Delta)$}
\RightLabel{$\equiv$~(\lem~\ref{lem:ld-equivalence-relative-to-labels})}
\dashedLine
\UnaryInfC{$\ld_{w}(\rel, Rvw, \Gamma \sar v : A, \Delta)$}
\RightLabel{=}
\dottedLine
\UnaryInfC{$\ldl^{w}(\rel_{1}, Rvw, \Gamma_{1} \sar v : A, \Delta_{1}) \dar \ldr^{w}(\lseq_{2})$}
\RightLabel{=}
\dottedLine
\UnaryInfC{$\bull A \circ \ldl^{w}(\rel_{1}, \Gamma_{1} \sar \Delta_{1}) \dar \ldr^{w}(\lseq_{2})$}
\RightLabel{$\dri$}
\UnaryInfC{$\bull A \dar \ldr^{w}(\lseq_{2}) \circ \dneg \ldl^{w}(\rel_{1}, \Gamma_{1} \sar \Delta_{1})$}
\RightLabel{$\drix$}
\UnaryInfC{$A \dar \bull (\ldr^{w}(\lseq_{2}) \circ \dneg \ldl^{w}(\rel_{1}, \Gamma_{1} \sar \Delta_{1}))$}
\RightLabel{$\pl$}
\UnaryInfC{$\p A \dar \ldr^{w}(\lseq_{2}) \circ \dneg \ldl^{w}(\rel_{1}, \Gamma_{1} \sar \Delta_{1})$}
\RightLabel{$\dri$}
\UnaryInfC{$\ldl^{w}(\rel_{1}, \Gamma_{1} \sar \Delta_{1}) \circ \p A \dar \ldr^{w}(\lseq_{2})$}
\RightLabel{=}
\dottedLine
\UnaryInfC{$\ldl^{w}(\rel_{1}, \Gamma_{1}, w : \p A \sar \Delta_{1}) \dar \ldr^{w}(\lseq_{2})$}
\RightLabel{=}
\dottedLine
\UnaryInfC{$\ld_{w}(\rel, \Gamma, w : \p A \sar \Delta)$}
\RightLabel{$\equiv$~(\lem~\ref{lem:ld-equivalence-relative-to-labels})}
\dashedLine
\UnaryInfC{$\ld_{u}(\rel, \Gamma,  w : \p A \sar \Delta)$}
\DisplayProof
\end{flushright}
 This concludes the proof of the inductive step. Let us now analyze the complexity of the translation and the relative sizes of the input and target proofs.
 
 Observe that each (logical) rule application in $\lproof$ is considered once and adds a single logical rule application in the target proof. Furthermore, each application of \lem~\ref{lem:ld-equivalence-relative-to-labels} adds only polynomially many inferences to the target proof and only a constant number of display and structural rules are added to the target proof beyond these inferences, meaning for some polynomial $p'$, $\qty{f(\lproof)} = \mathcal{O}(p'(\qty{\lproof}))$. Also, each display sequent in the target proof is polynomial in $\wdth{\lproof}$ by \rmk~\ref{rmk:dis-seq-polynomial-in-lpt}. Therefore, the translation takes place in $\ptime$ and $\size{f(\lproof)} = \mathcal{O}(p(\size{\lproof}))$, for some polynomial $p$.
\end{proof}

We now show how to extend the above translation to cover \emph{strict} primitive tense structural rules. To carry out this translation, we suppose that we have a strict instance of $\ptsrl$ under a substitution $\sub$, and show that this rule instance can be translated via $\ld_{w}$ to an instance of $\ptsrd$ under a substitution $\sub^{\ld}$. Therefore, let the following be a strict instance of $\ptsrl$:
\begin{center}
\AxiomC{$(\phi_{\lw}(A) \seqcomp \phi_{\lw}(B_{1}) \seqcomp \vL)\sub$}
\AxiomC{$\ldots$}
\AxiomC{$(\phi_{\lw}(A) \seqcomp \phi_{\lw}(B_{m}) \seqcomp \vL)\sub$}
\RightLabel{$\ptsrl$}
\TrinaryInfC{$(\phi_{\lw}(A) \seqcomp \vL)\sub$}
\DisplayProof
\end{center}
 where the substitution $\sub$ is defined accordingly:
$$
\sub := [\lseq_{1} / \ulseq_{p_{1}}^{\lu_{1}}, \ldots, \lseq_{k} / \ulseq_{p_{1}}^{\lu_{k}},
 \ldots,
\lseq_{m} / \ulseq_{p_{n}}^{\lu_{m}}, \ldots, \lseq_{\ell} / \ulseq_{p_{n}}^{\lu_{\ell}},
\lseq / \ulseq, u_{1} / \lu_{1}, \ldots, u_{n} / \lu_{n}]
$$ 
 To translate the above instance of $\ptsrl$ to an instance of $\ptsrd$ under $\ld_{w}$, we define the substitution $\sub^{\ld}$ accordingly:
$$
\sub^{\ld} := [\ld_{1}^{u_{1}}(\lseq_{1}) / \ux_{p_{1}}, \ldots, \ld_{1}^{u_{m}}(\lseq_{m}) / \ux_{p_{n}}, \ld_{2}^{w}(\lseq) / \ux]
$$
 We note that multiple labeled sequent variables $\varlseq_{p_{i}}^{\lv_{1}}, \ldots, \varlseq_{p_{i}}^{\lv_{k}}$ annotated with the same propositional atom $p_{i}$ may occur in $\sub$, while in $\sub^{\ld}$ only a single structure variable $\vX_{p_{i}}$ annotated with $p_{i}$ will occur. By condition $\conii$ 
 (see \dfn~\ref{def:primitive-tense-rule-labeled}), we know that $\varlseq_{p_{i}}^{\lv_{1}}\sub \iso  \cdots \iso \varlseq_{p_{i}}^{\lv_{k}}\sub$. 
 Therefore, $\ld_{1}^{v_{1}}(\varlseq_{p_{i}}^{\lv_{1}}\sub) = \cdots = \ld_{1}^{v_{k}}(\varlseq_{p_{i}}^{\lv_{k}}\sub)$, meaning we can arbitrarily choose any labeled sequent $\varlseq_{p_{i}}^{\lv_{j}}\sub$ with $1 \leq j \leq k$ to use to substitute $\ld_{1}^{v_{j}}(\varlseq_{p_{i}}^{\lv_{j}}\sub)$ for $\vX_{p_{i}}$ in $\sub^{\ld}$. As shown above, we simply choose the first labeled sequent mentioned in $\sub$ in defining the substitution $\sub^{\ld}$. The interested reader should consult \eg~\ref{eg:ptsrl-ptsrd} in the previous section to see how an instance of $\ptsrd$ is obtained from an instance of $\ptsrl$ by reversing the example and viewing $\sub$ in the example as $(\sub^{\dl})^{\ld}$. 

\begin{lemma}\label{lem:translating-antecedents-ii} Let us consider a strict instance of a primitive tense structural rule as shown below.
\begin{center}
\AxiomC{$(\phi_{\lw}(A) \seqcomp \phi_{\lw}(B_{1}) \seqcomp \vL)\sub$}
\AxiomC{$\ldots$}
\AxiomC{$(\phi_{\lw}(A) \seqcomp \phi_{\lw}(B_{m}) \seqcomp \vL)\sub$}
\RightLabel{$\ptsrl$}
\TrinaryInfC{$(\phi_{\lw}(A) \seqcomp \vL)\sub$}
\DisplayProof
\end{center}
and let $C \in \{A, B_{1}, \ldots, B_{m}\}$. Then, 
\begin{enumerate}

\item $\ld^{w}_{1}(\phi_{\lw}(C)\sub) = \psi(C)\subld$

\item $\ld_{w}((\phi_{\lw}(C) \seqcomp \ulseq)\sub) = (\psi(C) \dar \ux)\subld$

\end{enumerate}
\end{lemma}

\begin{proof} We prove claim 1 by induction on the complexity of $C$. Claim 2 follows from claim 1. 

\textit{Base case.} If $A$ is an atom $p_{i}$, so that $\phi_{\lw}(p_{i}) = \ulseq_{p_{i}}^{\lw}$ and $\sub(\ulseq_{p_{i}}^{\lw}) = \lseq_{i}$, then the proof shown below right proves the case. If $A$ is the logical constant $\top$, so that $\phi_{\lw}(A) = \phi_{\lw}(\top) = (\empdata \sar \empdata)$, then the proof shown below right establishes the case.

\begin{minipage}{.45\textwidth}
\begin{eqnarray*}
& & \ld^{w}_{1}(\phi_{\lw}(p_{i})\sub) \\
& & \ld^{w}_{1}(\ulseq_{p_{i}}^{\lw}\sub)\\
& & \ld^{w}_{1}(\lseq_{i})\\
& & \ux_{p_{i}}\subld\\
& & \psi(p_{i})\subld
\end{eqnarray*}
\end{minipage}
\begin{minipage}{.45\textwidth}
\begin{eqnarray*}
& & \ld^{w}_{1}(\phi_{\lw}(\top)\sigma) \\
& & \ld^{w}_{1}(\empdata \sar \empdata)\\
& & I \\
& & I\subld\\
& & \psi(\top)\subld
\end{eqnarray*}
\end{minipage}\\

\textit{Inductive step.} We consider the cases where $C = D \land E$, $C = \f D$, and $C = \p D$ in turn. In the first case, we have that $\phi_{\lw}(C) = \phi_{\lw}(D \land E) = \phi_{\lw}(D) \scomp \phi_{\lw}(E)$, and the proof below establishes the claim in this case. The third to fourth step below follows from the fact that $\ptsrl$ is strict; in particular, $\sub$ satisfies conditions $\coniv$, $\conv$, and $\convii$ (see \dfn~\ref{def:strictness}).\\
\begin{minipage}{.95\textwidth}
\begin{eqnarray*}
& & \ld^{w}_{1}(\phi_{\lw}(D \land E)\sigma)\\
& & \ld^{w}_{1}((\phi_{\lw}(D) \seqcomp \phi_{\lw}(E))\sigma)\\
& & \ld^{w}_{1}(\phi_{\lw}(D)\sub \seqcomp \phi_{\lw}(E)\sub)\\
& & \ld^{w}_{1}(\phi_{\lw}(D)\sigma) \circ \ld^{w}_{1}(\phi_{\lw}(E)\sigma)\\
& & \psi(D)\subld \circ \psi(E)\subld\\
& & (\psi(D) \circ \psi(E))\subld\\
& & \psi(D \land E)\subld
\end{eqnarray*}
\end{minipage}\\

The proof of the second case is shown below left, where $\phi_{\lw}(A) = \phi_{\lw}(\f D) = (R\lw\lu \sar \empdata) \scomp \phi_{\lu}(D)$, and the proof of the final case is shown below right, where $\phi_{\lw}(A) = \phi_{\lw}(\p D) = (R\lu\lw \sar \empdata) \scomp \phi_{\lu}(D)$.

\begin{minipage}{.45\textwidth}
\begin{eqnarray*}
& & \ld_{w}(\phi_{\lw}(\f D)\sigma)\\
& & \ld_{w}(((R\lw\lu \sar \empdata) \scomp \phi_{\lu}(D))\sigma)\\
& & \ld_{w}((R\lw\lu \sar \empdata) \scomp \phi_{\lu}(D)\sigma)\\
& & \dneg \bull  \dneg  \ld_{u}(\phi_{\lu}(D)\sigma)\\
& & \dneg \bull  \dneg  \psi(D)\subld\\
& & (\dneg \bull  \dneg  \psi(D))\subld\\
& & \psi(\f D)\subld
\end{eqnarray*}
\end{minipage}
\begin{minipage}{.45\textwidth}
\begin{eqnarray*}
& & \ld_{w}(\phi_{w}(\p D)\sigma)\\
& & \ld_{w}(((Ruw \sar \empdata) \scomp \phi_{u}(D))\sub)\\
& & \ld_{w}((Ruw \sar \empdata) \scomp \phi_{u}(D)\sub)\\
& & \bull   \ld_{u}(\phi_{u}(D)\sub)\\
& & \bull   \psi(D)\subld\\
& & (\bull   \psi(D))\subld\\
& & \psi(\p D)\subld
\end{eqnarray*}
\end{minipage}\\
\end{proof}


\begin{theorem}\label{thm:gtktp-to-dktp} If $\proves{\gtktp}{\lproof}{\lseq}$ and $\lproof$ is a strict labeled polytree proof, then there exists a $\ptime$ function $f$ in $\size{\lproof}$ and polynomial $p$ such that $\proves{\dktp}{f(\lproof)}{\ld_{w}(\lseq)}$ and $\size{f(\lproof)} = \mathcal{O}(p(\size{\lproof}))$.
\end{theorem}

\begin{proof} We prove the result by extending the inductive step of \thm~\ref{thm:gtkt-to-dkt} with the $\ptsrl$ case. Let us suppose we are translating a strict labeled polytree proof $\lproof$, and let us consider a strict instance of $\ptsrl$ as shown below top. The desired translation is obtained as shown below bottom, where the $\iso$ steps follow from \lem~\ref{lem:ld-equivalence-relative-to-labels} and the equality steps are obtained by \lem~\ref{lem:translating-antecedents-ii}.
\begin{flushleft}
\begin{tabular}{c c}
\AxiomC{$(\phi_{\lw}(A) \seqcomp \phi_{w}(B_{1}) \seqcomp \ulseq)\sub$}
\AxiomC{$\ldots$}
\AxiomC{$(\phi_{\lw}(A) \seqcomp \phi_{w}(B_{m}) \seqcomp \ulseq)\sub$}
\RightLabel{$\ptsrl$}
\TrinaryInfC{$(\phi_{\lw}(A) \seqcomp \ulseq)\sub$}
\DisplayProof

&

$\leadsto$
\end{tabular}
\end{flushleft}
\begin{flushright}
\AxiomC{$\ld_{u}((\phi_{\lw}(A) \seqcomp \phi_{\lw}(B_{1}) \seqcomp \ulseq)\sub)$}
\RightLabel{$\iso$}
\dottedLine
\UnaryInfC{$\ld_{w}((\phi_{\lw}(A) \seqcomp \phi_{\lw}(B_{1}) \seqcomp \ulseq)\sub)$}
\RightLabel{=}
\dottedLine
\UnaryInfC{$(\psi(A) \circ \psi(B_{1}) \dar \vX)\subld$}
\RightLabel{$\drii$}
\UnaryInfC{$(\psi(B_{1}) \dar \dneg \psi(A)  \circ \vX)\subld$}

\AxiomC{$\ldots$}

\AxiomC{$\ld_{u}((\phi_{\lw}(A) \seqcomp \phi_{w}(B_{m}) \seqcomp \ulseq)\sub)$}
\RightLabel{$\iso$}
\dottedLine
\UnaryInfC{$\ld_{w}((\phi_{\lw}(A) \seqcomp \phi_{w}(B_{m}) \seqcomp \ulseq)\sub)$}
\RightLabel{=}
\dottedLine
\UnaryInfC{$(\psi(A) \circ \psi(B_{m}) \dar \vX)\subld$}
\RightLabel{$\drii$}
\UnaryInfC{$(\psi(B_{m}) \dar \dneg \psi(A)  \circ \vX)\subld$}

\RightLabel{$\ptsrd$}
\TrinaryInfC{$(\psi(A) \dar \dneg \psi(A) \circ \vX)\subld$}
\RightLabel{$\drii$}
\UnaryInfC{$(\psi(A) \circ \psi(A)  \dar \vX)\subld$}
\RightLabel{$\dcl$}
\UnaryInfC{$(\psi(A) \dar \vX)\subld$}
\RightLabel{=}
\dottedLine
\UnaryInfC{$\ld_{w}((\psi(A) \dar \uy)\sub)$}
\DisplayProof
\end{flushright}
 By an argument similar to the one at the end of \thm~\ref{thm:gtkt-to-dkt}, it follows that $f(\lproof)$ is computable in $\ptime$ and $\size{f(\lproof)} = \mathcal{O}(p(\size{\lproof}))$, for some polynomial $p$.
\end{proof}

\section{Concluding Remarks}\label{conclusion}

We have established a $\ptime$ correspondence between the maximal analytic class of display calculi for tense logics (namely, all display calculi for primitive tense logics) and labeled sequent systems, thus solving an open problem posed in~\cite{CiaLyoRamTiu21,Res06}. Our translations show that the space of cut-free display proofs is polynomially equivalent to the space of strict labeled proofs. We also find that translating from display to labeled preserves (or decreases) the quantity of the proof, while potentially increasing the width; the reverse translation brings about a polynomial increase in the size of the proof. Moreover, we made the interesting observation that display structures translate to labeled polytree sequents and that all display equivalent sequents translate to the same labeled sequent (up to isomorphism), showing that labeled sequents are a canonical representation of display sequents that enjoy less bureaucracy.

A fascinating avenue of future research would be to study translations and the computational relationship between strict labeled proofs and non-strict labeled proofs (which might include labeled sequents \emph{not} in a polytree form). Composing such translations with the translations in this paper would confirm if it is possible or not to translate any analytic display proof into \emph{any} analytic labeled proof, or vice versa, in $\ptime$. Transforming a strict labeled proof into a non-strict one, where labeled sequents may include disconnectivity or (un)directed cycles, is simple to obtain by applying (1) the hp-admissibility of $\wk$ to add disjoint regions to labeled sequents in the proof, and (2) by applying the hp-admissibility of $\lsub$ to identify labels, bringing about the formation of cycles. Showing the converse, i.e. that any non-strict labeled proof can be obtained from a strict proof via applications of $\wk$ and $\lsub$ (up to admissible permutations of rules in a proof), appears far more challenging and may require new proof-theoretic techniques (or, operations beyond weakenings, label substitutions, and permutations).

Nevertheless, the above suggested correspondence between strict and non-strict labeled proofs would imply that every non-strict labeled proof is essentially a homomorphic image of a strict proof. This would tell us that display calculi are a special fragment of labeled calculi that generate proofs homomorphically mappable  into any other proof of the same conclusion (up to admissible permutations of rules). It is not clear if transforming non-strict proofs into strict proofs is possible in $\ptime$ however (assuming $\ptime \neq \np$) as finding the correct `homomorphic mapping' from a strict to a non-strict proof may require an $\np$ algorithm (cf.~\cite{HelBes04}).



\bibliographystyle{abbrv}
\bibliography{bibliography}

\begin{thebibliography}{10}

\bibitem{Avr96}
A.~Avron.
\newblock The method of hypersequents in the proof theory of propositional
  non-classical logics.
\newblock In W.~Hodges, M.~Hyland, C.~Steinhorn, and J.~Truss, editors, {\em
  Logic: From Foundations to Applications: European Logic Colloquium}, page
  1–32. Clarendon Press, USA, 1996.

\bibitem{BaaHorLutSat17}
F.~Baader, I.~Horrocks, C.~Lutz, and U.~Sattler.
\newblock {\em Introduction to Description Logic}.
\newblock Cambridge University Press, 2017.

\bibitem{Bel82}
N.~D. Belnap.
\newblock Display logic.
\newblock {\em Journal of philosophical logic}, 11(4):375--417, 1982.

\bibitem{BlaRijVen01}
P.~Blackburn, M.~de~Rijke, and Y.~Venema.
\newblock {\em Modal Logic}, volume~53 of {\em Cambridge Tracts in Theoretical
  Computer Science}.
\newblock Cambridge University Press, 2001.

\bibitem{Bor08}
B.~Boretti.
\newblock {\em Proof Analysis in Temporal Logic}.
\newblock PhD thesis, University of Milan, 2008.

\bibitem{Bro12}
J.~Brotherston.
\newblock Bunched logics displayed.
\newblock {\em Stud. Log.}, 100(6):1223–1254, Dec. 2012.

\bibitem{BuiGor07}
L.~Buisman and R.~Gor{\'e}.
\newblock A cut-free sequent calculus for bi-intuitionistic logic.
\newblock In N.~Olivetti, editor, {\em Automated Reasoning with Analytic
  Tableaux and Related Methods}, pages 90--106, Berlin, Heidelberg, 2007.
  Springer Berlin Heidelberg.

\bibitem{Bul92}
R.~A. Bull.
\newblock Cut elimination for propositional dynamic logic without *.
\newblock {\em Z. Math. Logik Grundlag. Math.}, 38(2):85--100, 1992.

\bibitem{CasSma02}
C.~Castellini and A.~Smaill.
\newblock A systematic presentation of quantified modal logics.
\newblock {\em Logic Journal of the IGPL}, 10(6):571--599, 2002.

\bibitem{CiaLyoRamTiu21}
A.~Ciabattoni, T.~S. Lyon, R.~Ramanayake, and A.~Tiu.
\newblock Display to labeled proofs and back again for tense logics.
\newblock {\em ACM Trans. Comput. Logic}, 22(3), July 2021.

\bibitem{DeuNasRem08}
A.~Deutsch, A.~Nash, and J.~B. Remmel.
\newblock The chase revisited.
\newblock In M.~Lenzerini and D.~Lembo, editors, {\em Proceedings of the 27th
  {ACM} {SIGMOD-SIGACT-SIGART} Symposium on Principles of Database Systems
  (PODS'08)}, pages 149--158. {ACM}, 2008.

\bibitem{Dyc92}
R.~Dyckhoff.
\newblock Contraction-free sequent calculi for intuitionistic logic.
\newblock {\em The Journal of Symbolic Logic}, 57(3):795--807, 1992.

\bibitem{Fit72}
M.~Fitting.
\newblock Tableau methods of proof for modal logics.
\newblock {\em Notre Dame Journal of Formal Logic}, 13(2):237--247, 1972.

\bibitem{Gab96}
D.~M. Gabbay.
\newblock {\em Labelled deductive systems}, volume~33 of {\em Oxford Logic
  guides}.
\newblock Clarendon Press/Oxford Science Publications, 1996.

\bibitem{Gen35a}
G.~Gentzen.
\newblock Untersuchungen {\"u}ber das logische schlie{\ss}en. i.
\newblock {\em Mathematische zeitschrift}, 39(1):176--210, 1935.

\bibitem{Gen35b}
G.~Gentzen.
\newblock Untersuchungen {\"u}ber das logische schlie{\ss}en. ii.
\newblock {\em Mathematische Zeitschrift}, 39(1):405--431, 1935.

\bibitem{Gir87}
J.-Y. Girard.
\newblock Linear logic.
\newblock {\em Theoretical computer science}, 50(1):1--101, 1987.

\bibitem{Gor98}
R.~Goré.
\newblock {Substructural logics on display}.
\newblock {\em Logic Journal of the IGPL}, 6(3):451--504, 05 1998.

\bibitem{Hei05}
R.~Hein.
\newblock Geometric theories and proof theory of modal logic.
\newblock Master's thesis, Technische Universit\"at Dresden, 2005.

\bibitem{HelBes04}
P.~Hell and J.~Nesetril.
\newblock {\em {Graphs and Homomorphisms}}.
\newblock Oxford University Press, 07 2004.

\bibitem{Kan57}
S.~Kanger.
\newblock {\em Provability in logic}.
\newblock Almqvist \& Wiksell, 1957.

\bibitem{Kas94}
R.~Kashima.
\newblock Cut-free sequent calculi for some tense logics.
\newblock {\em Studia Logica}, 53(1):119--135, 1994.

\bibitem{Kra96}
M.~Kracht.
\newblock Power and weakness of the modal display calculus.
\newblock In H.~Wansing, editor, {\em Proof Theory of Modal Logic}, pages
  93--121. Springer Netherlands, Dordrecht, 1996.

\bibitem{Kri63}
S.~A. Kripke.
\newblock Semantical considerations on modal logic.
\newblock {\em Acta Philosophica Fennica}, 16:83--94, 1963.

\bibitem{Lel15}
B.~Lellmann.
\newblock Linear nested sequents, 2-sequents and hypersequents.
\newblock In H.~De~Nivelle, editor, {\em Automated Reasoning with Analytic
  Tableaux and Related Methods}, volume 9323 of {\em Lecture Notes in Computer
  Science}, pages 135--150, Cham, 2015. Springer International Publishing.

\bibitem{LyoTiuGorClo20}
T.~Lyon, A.~Tiu, R.~Gor{\'{e}}, and R.~Clouston.
\newblock Syntactic interpolation for tense logics and bi-intuitionistic logic
  via nested sequents.
\newblock In M.~Fern{\'{a}}ndez and A.~Muscholl, editors, {\em 28th {EACSL}
  Annual Conference on Computer Science Logic (CSL)}, volume 152 of {\em
  LIPIcs}, pages 28:1--28:16. Schloss Dagstuhl - Leibniz-Zentrum f{\"{u}}r
  Informatik, 2020.

\bibitem{Mae60}
S.~Maehara.
\newblock On the interpolation theorem of craig.
\newblock {\em S{\^u}gaku}, 12(4):235--237, 1960.

\bibitem{McM18}
K.~L. McMillan.
\newblock Interpolation and model checking.
\newblock In E.~M. Clarke, T.~A. Henzinger, H.~Veith, and R.~Bloem, editors,
  {\em Handbook of Model Checking}, pages 421--446. Springer International
  Publishing, Cham, 2018.

\bibitem{Min97}
G.~Mints.
\newblock Indexed systems of sequents and cut-elimination.
\newblock {\em Journal of Philosophical Logic}, 26(6):671--696, 1997.

\bibitem{Pot83}
G.~Pottinger.
\newblock Uniform, cut-free formulations of t, s4 and s5.
\newblock {\em Journal of Symbolic Logic}, 48(3):900, 1983.

\bibitem{Pri03}
A.~N. Prior.
\newblock {\em Time and Modality}.
\newblock Oxford University Press, 1957.

\bibitem{Res06}
G.~Restall.
\newblock Comparing modal sequent systems.
\newblock {\em unpublished}, pages 1--13, 2006.

\bibitem{Sim94}
A.~K. Simpson.
\newblock {\em The proof theory and semantics of intuitionistic modal logic}.
\newblock PhD thesis, University of Edinburgh. College of Science and
  Engineering. School of Informatics, 1994.

\bibitem{Sla97}
J.~K. Slaney.
\newblock Minlog: {A} minimal logic theorem prover.
\newblock In W.~McCune, editor, {\em Automated Deduction - CADE-14,
  Proceedings}, volume 1249 of {\em Lecture Notes in Computer Science}, pages
  268--271. Springer, 1997.

\bibitem{Smu68}
R.~M. Smullyan.
\newblock {\em First-Order Logic}.
\newblock Springer-Verlag, 1968.

\bibitem{Sto04}
P.~Stouppa.
\newblock The design of modal proof theories: The case of s5.
\newblock Master's thesis, Technische Universit\"at Dresden, 2004.

\bibitem{BerStr22}
K.~van Berkel and C.~Stra{\ss}er.
\newblock Reasoning with and about norms in logical argumentation.
\newblock In F.~Toni, S.~Polberg, R.~Booth, M.~Caminada, and H.~Kido, editors,
  {\em Computational Models of Argument - Proceedings of {COMMA} 2022}, volume
  353 of {\em Frontiers in Artificial Intelligence and Applications}, pages
  332--343. {IOS} Press, 2022.

\bibitem{Vig00}
L.~Vigan{\`o}.
\newblock {\em Labelled Non-Classical Logics}.
\newblock Springer Science \& Business Media, 2000.

\bibitem{Wan94}
H.~Wansing.
\newblock Sequent calculi for normal modal propositional logics.
\newblock {\em Journal of Logic and Computation}, 4(2):125--142, 1994.

\bibitem{Wol98}
F.~Wolter.
\newblock On logics with coimplication.
\newblock {\em Journal of Philosophical Logic}, 27(4):353--387, 1998.

\end{thebibliography}

\end{document}